%% file: main.tex
\documentclass{article}
\usepackage[a4paper,width=150mm,top=25mm,bottom=25mm]{geometry}
\usepackage{natbib}
\usepackage{amsthm}
\newenvironment{keywords}{Keywords: }{\newline}
\newtheorem{definition}{Definition}
\newtheorem{theorem}{Theorem}
\newtheorem{corollary}{Corollary}

\newcommand{\email}[1]{}
\newcommand{\doi}[1]{}
\renewcommand{\date}[1]{}
\newcommand{\pagerange}[1]{}
\newcommand{\volume}[1]{}
\newcommand{\pubyear}[1]{}
\newcommand{\artmonth}[1]{}
\newcommand{\backmatter}{}
\usepackage{amsmath}

\usepackage{times}
\usepackage{bm}

\usepackage{amsfonts}
\usepackage{tikz}
\usepackage{caption} 
\usepackage{enumerate}

\graphicspath{{./figures/}}

\def\FDA{FDA}
\def\SRV{SRV}
\def\REML{REML}

\def\im{{\mathbf{i}}}
\def\ind{1}

\def\Supplement{Supplement}

\def\bSig\mathbf{\Sigma}

\newtheorem{proposition}{Proposition}

\title
{Elastic Full Procrustes Analysis of Plane Curves via Hermitian Covariance Smoothing}

\author
{Almond St\"ocker$^{1,2,*}$\email{almond.stoecker@epfl.ch}, 
Manuel Pfeuffer$^{1}$, 
Lisa Steyer$^{1}$, 
and Sonja Greven$^{1}$\\
$^{1}$Chair of Statistics, School of Business and Economics, Humboldt-Universit{\"a}t zu Berlin,\\ Unter den Linden 6, 10099 Berlin, Germany \\
$^{2}$Department of Mathematics, École polytechnique fédérale de Lausanne (EPFL),\\ Station 8, CH-1015 Lausanne, Switzerland}

\begin{document}


\date{{\it Received November} 2022. {\it Revised XX} XX.  {\it Accepted March} XX.}



\pagerange{\pageref{firstpage}--\pageref{lastpage}} 
\volume{XX}
\pubyear{XX}
\artmonth{XXXX}


\doi{10.1111/j.1541-0420.2005.00454.x}


\label{firstpage}

\maketitle

\begin{abstract}
	Determining the mean shape of a collection of curves is not a trivial task, in particular when curves are only irregularly/sparsely sampled at discrete points. 
	We newly propose an elastic full Procrustes mean of shapes of (oriented) plane curves, which are considered equivalence classes of parameterized curves 
	with respect to translation, rotation, scale, and re-parameterization (warping), based on the square-root-velocity (\SRV) framework.
	Identifying the real plane with the complex numbers, we establish a connection to covariance estimation in irregular/sparse functional data analysis. 
    We introduce Hermitian covariance smoothing and show how to employ this extension of existing covariance estimation methods for obtaining an estimator of the (in)elastic full Procrustes mean, also in the sparse case not yet covered by existing (intrinsic) elastic shape means.
    For this, we provide different groundwork results which are also of independent interest:
    we characterize (the decomposition of) the covariance structure of rotation-invariant bivariate stochastic processes using complex representations, and we identify sampling schemes that allow for exact observation of derivatives/\SRV\ transforms of sparsely sampled curves.
    We demonstrate the performance of the approach 
	in a phonetic study on tongue shapes and in different realistic simulation settings, inter alia based on handwriting data.
\end{abstract}

%





%

\section{Introduction}
\label{sec:intro}

When comparing the shape of, say, a specific outline marked on medical images across different patients, the concrete coordinate system used for recording is often arbitrary and not of interest: 
the shape neither depends on positioning in space, nor on orientation or size. Analogously, the outline can be mathematically represented via a parameterized curve $\beta: [0,1] \rightarrow \mathbb{R}^2$, but the particular parameterization of the outline curve is often not of interest, only its image. We study datasets where an observational unit is the shape of a plane curve, 
defined as equivalence class a) over  the shape invariances translation, rotation and scale and b) over re-parameterization.
More specifically, we generalize the notion of a full Procrustes mean from discrete landmark shape analysis \citep{DrydenMardia2016ShapeAnalysisWithApplications} to elastic shape analysis of curves, in particular to achieve improved estimation properties in irregular/sparsely measured scenarios compared to existing 
``intrinsic'' elastic mean shape estimation methods relying on  geodesic distances.  
To allow this generalization of landmark shape means to curves, we also present two results characterizing the covariance structure of rotation-invariant bivariate stochastic processes via their complex representations. To enable derivative-based elastic analysis of sparsely/irregularly sampled curves, which are common in practice but for which existing methods have problems, we provide a result on the feasibility to exactly observe the necessary derivatives under such sampling. 
While these results are important building blocks in preparing the proposed elastic Full Procrustes mean estimation, they are also of independent interest in their own right.

For landmark shapes, 
different notions of mean shape are well-established including, in addition to the full Procrustes mean, in particular also the intrinsic shape mean, i.e.\, the Riemannian center of mass in the shape space. \citet{Dryden2014MeanShapes} discuss properties of different shape mean concepts, pointing out that the full Procrustes mean is more robust with respect to outliers than the intrinsic mean or the \emph{partial} Procrustes mean fixing scale to unit size. Further discussion of these three mean concepts, which all present Fréchet means based on different distances, can be found in \citet{Huckemann2012meaning}. 
The full Procrustes mean also arises as the mode 
of a complex Bingham distribution \citep{Kent1994}  on (unit-norm) landmark configurations  $\mathbf{X} \in \mathbb{C}^k$ of $k$ landmarks, which is commonly used to model planar landmark shapes, identifying the real plane $\mathbb{R}^2 \cong \mathbb{C}$ with the complex numbers.
Moreover, it corresponds to the leading eigenvector of the complex covariance matrix of $\mathbf{X}$,  an important point we generalize for the estimation strategy proposed for curve mean shapes in this paper. 

Compared to landmark shapes, different additional challenges arise for shapes of curves: 
invariance with respect to re-parameterization (warping) 
is one that is highly related to the registration problem in function data analysis \citep[\FDA,][]{RamsaySilverman2005}. In the context of shape analysis of curves, \citet{Srivastava2011ShapeAnaofElasticCurves} propose an \emph{elastic} re-parameterization invariant metric, 
allowing to define a proper distance between two curves via optimal warping alignment. Greatly simplifying the formulation of the metric by working with square-root-velocity (\SRV) transformations of the curves, 
this lead to a rapidly growing literature on \emph{functional} shape analysis of curves in the \SRV-framework \citep[see e.g.,][]{SrivastavaKlassen2016}. 
However, so far the focus lay on elastic generalization of the intrinsic shape mean instead of the (potentially more robust) full Procrustes mean which we generalize here. In simulations, we  illustrate how the novel elastic full Procrustes mean estimation yields improved mean estimates in irregularly/sparsely sampled data (sometimes even as an estimator of the intrinsic shape mean, compared to existing estimators designed for this alternative mean). 

\emph{Sparsely/irregularly} observed curves have been considered in the \SRV-framework by \citet{steyer2021elastic},  
however, only restricting to re-parameterization invariance and not investigating shape means. 
Such data with a comparatively low number of samples per curve often results in practice when the sampling rate of a measurement device is limited, or the resolution of images used for curve segmentation is coarse. In  \FDA, 
sparse/irregular functional data is commonly distinguished from dense/regular data, as it requires  explicit treatment. Models for sparse/irregular data are often based on smooth (spline) function bases 
and commonly involve 
assumption of (small) measurement errors on the discrete curve evaluations 
\citep{GrevenScheipl2017}. 

Focusing on shape analysis of sparsely/irregularly measured curves, 
we consider the full Procrustes mean concept particularly attractive due to its robustness known from landmark shape analysis, and due to its direct connection to the covariance structure of the data, which allows relying on a core estimation strategy in sparse/irregular \FDA: 
following \citet{Yao:etal:2005}, covariance smoothing has become a major tool
for sparse/irregular \FDA, allowing to reconstruct the functional covariance structure based on sparse evaluations. \citet{cederbaum2018fast, reiss2020tensorFPC} discuss (symmetric) tensor-product spline smoothing for this purpose, considering univariate functional data. \citet{HappGreven2018} generalize univariate approaches to conduct functional principal component analysis also for multivariate sparse/irregular data.

In this paper, our contributions are to  1.\ characterize the complex covariance (decomposition) of rotation-invariant bivariate stochastic processes. This gives us the  basis to  2.\ develop
Hermitian covariance smoothing, 
which we 3.\ use to propose a covariance-based  
estimation method for  the 4.\ novel (elastic) full Procrustes means we propose as a more robust notion of elastic shape mean, with a particular focus also on sparsely/irregularly sampled curves. For such realistic curve measurements  we 5.\ characterize scenarios where exact sampling of the necessary SRVs/derivatives is feasible. 

In the following, we first discuss in Section 2 complex stochastic processes as random elements of Hilbert spaces, illustrating their convenience for rotation-invariant bivariate \FDA\, and propose Hermitian tensor-product smoothing 
for complex functional principle component analysis. 
This lays the groundwork for the second part 
of the paper in Section 3, where 
we introduce the notion of elastic (and inelastic) full Procrustes mean shapes of plane curves based on the \SRV-framework. 
We show conditions 
under which exactly  observing \SRV s (i.e., curve derivatives) of sparsely/irregularly measured curves  is feasible and propose estimation of their full Procrustes means via Hermitian covariance smoothing.
Finally, we present an elastic full Procrustes analysis of tongue outlines observed from participants of a phonetic study and validate the proposed approach in three simulation scenarios in Sections 5 and 4.
Proofs for all propositions are given in 
an online supplement. 
A ready to use implementation is offered in the R-package \texttt{elastes} (\texttt{github.com/mpff/elastes}).

\section{Hermitian covariance smoothing}
\label{sec:covariance}

\subsection{Complex processes and rotation invariance}
\label{sec:covariance:complex_processes}

Although functional data analysis traditionally focuses on Hilbert spaces over $\mathbb{R}$ \citep[compare, e.g.,][]{Hsing2015TheoreticalFoundations}, underlying functional analytic statements cover Hilbert spaces over $\mathbb{C}$ as well \citep[e.g.,][]{Rynne2007linearFuncAna}.
This lets us formulate principal component analysis for complex-valued functional data and underlying concepts 
in  analogy to the real case in the following. 
Subsequently, we present two results on the relation of complex to bivariate (real) functional data and on the convenience of a complex viewpoint under rotation invariance that will be key in our estimation approach.
Although complex stochastic processes have been discussed in the literature \citep{neeser1993proper}, we are not aware of any previous discussion of the results we present in this section.
In the complex viewpoint, the real plane $\mathbb{R}^2$ is identified with the complex numbers $\mathbb{C}$ via the canonical vector space isomorphism $\kappa: \mathbb{C} \rightarrow \mathbb{R}^2$, $z \mapsto \mathbf{z} = (\Re(z), \Im(z))^\top$ mapping $z\in\mathbb{C}$ to its real part $\Re(z)$ and imaginary part $\Im(z)$.
By $z^\dagger$ we denote the complex conjugate $\Re(z) - \im\, \Im(z)$ of $z\in\mathbb{C}$, with $\im^2 = -1$, or more generally the Hermitian adjoint (conjugate transpose) for complex matrices or operators.
Rotation of $\mathbf{z} \in \mathbb{R}^2$ by $\omega \in \mathbb{R}$ radians simplifies to scalar multiplication $\exp(\im\,\omega)\, z \in \mathbb{C}$ in complex representation. 

Let $Y$ be a complex-valued stochastic process with realizations $y: \mathcal{T} \rightarrow \mathbb{C}$  in $\mathbb{L}^2(\mathcal{T}, \mathbb{C})$,
where $\mathcal{T}$ is a compact metric space with finite measure $\nu$. Here, $\mathcal{T}=[0,1]$ is typically  the unit interval with $\nu$ the Lebesgue measure, and $t\in \mathcal{T}$ is referred to as ``time". 
The complex, separable Hilbert space $\mathbb{L}^2(\mathcal{T}, \mathbb{C})$ of square-integrable complex-valued functions is equipped with the inner product $\langle x, y \rangle = \int x^\dagger(t) y(t) \,d\nu(t)$ for $x,y \in \mathbb{L}^2(\mathcal{T}, \mathbb{C})$ and the corresponding norm $\|\cdot\|$.

\begin{definition}
	\begin{enumerate}[i)]
		\item  $Y$ is called \emph{random element} in a real or complex Hilbert space $\mathbb{H}$ if $\langle x, Y\rangle$ is measurable for all $x\in \mathbb{H}$ and the distribution of $Y$ is uniquely determined by the (marginal) distributions of $\langle x, Y\rangle$ over $x \in \mathbb{H}$.
		\item The mean $\mu\in\mathbb{H}$ and covariance operator $\Sigma: \mathbb{H} \rightarrow \mathbb{H}$ of a random element $Y$ are defined via $\langle \mu, x \rangle = \mathbb{E}\left( \langle Y, x \rangle \right)$ and $\langle \Sigma(x), y \rangle = \mathbb{E}\left( \langle x, Y - \mu \rangle \langle Y - \mu, y \rangle  \right)$ for all $x,y \in \mathbb{H}$.
	\end{enumerate}	
\end{definition}

In the following, we assume $Y$ is a random element of $\mathbb{L}^2(\mathcal{T}, \mathbb{C})$. Being self-adjoint and compact, its covariance operator $\Sigma$ admits a representation $\Sigma(f) = \sum_{k\geq1} \lambda_k \langle e_k, f\rangle e_k$ via countably many eigenfunctions $e_1, e_2, \dots \in \mathbb{L}^2(\mathcal{T}, \mathbb{C})$, $\Sigma(e_k) = \lambda_k e_k$, with real eigenvalues $\lambda_1 \geq \lambda_2 \geq \dots \geq 0$ of $\Sigma$ (see \Supplement). The $\{e_k\}_k$ form an orthonormal basis of the Hilbert subspace formed by the closure of the image of $\Sigma$. 
The random element can be represented as $Y = \mu + \sum_{k\geq1} \langle e_k, Y - \mu \rangle e_k$ with probability one. The scores $Z_k = \langle e_k, Y - \mu \rangle$, $k\geq 1$, are complex random variables with mean zero and covariance $\operatorname{Cov}\left(Z_k, Z_{k'}\right) =\mathbb{E}\left(\langle Y - \mu, e_{k}\rangle \langle e_{k'}, Y - \mu \rangle\right) = \lambda_k \ind_{\{k'\}}(k)$, where $\ind_{\mathcal{S}}(t) = 1$ if $t\in\mathcal{S}$ and 0 else for a set $\mathcal{S}$.

$Y$ is canonically identified with the bivariate real process $\mathbf{Y}= \kappa(Y) = (\Re(Y), \Im(Y))^\top$, random element in the Hilbert space $\mathbb{L}^2(\mathcal{T}, \mathbb{R}^2)$ with the inner product of $\mathbf{x} = \kappa(x), \mathbf{y} = \kappa(y)$, $x,y \in \mathbb{L}^2(\mathcal{T}, \mathbb{C})$, defined by $\langle \mathbf{x}, \mathbf{y} \rangle =  \int \Re\left(x(t)\right) \Re\left(y(t)\right) \,d\nu(t) + \int \Im\left(x(t)\right) \Im\left(y(t)\right) \,d\nu(t)= \Re\left(\langle x, y \rangle\right)$.

\begin{theorem}
	\label{lem:bivariateCov}
	Define the pseudo-covariance operator $\Omega$ of $Y$ with mean $\mu$ by $\langle \Omega(x), y \rangle = \mathbb{E}\left(\langle Y-\mu, x \rangle \langle Y-\mu, y \rangle \right)$ for all $x,y\in \mathbb{L}^2(\mathcal{T}, \mathbb{C})$, and let $\mathbf{\Sigma}$ denote the covariance operator of $\mathbf{Y} = \kappa(Y)$. Then the covariance and pseudo-covariance operators $\Sigma$ and $\Omega$ of $Y$ together determine $\mathbf{\Sigma}$ via 
	$$ \kappa^{-1} \circ \mathbf{\Sigma} \circ \kappa = {(\Sigma + \Omega)}/{2}.$$
\end{theorem}

Aiming at shape analysis, we are particularly interested in rotation-invariant distributions $\mathfrak{L}(\mathbf{Y})$ of $\mathbf{Y} = \kappa(Y)$, corresponding to $\mathfrak{L}(Y) = \mathfrak{L}(\exp(\im\,\omega)Y)$ for all $\omega\in\mathbb{R}$. In this case, $\mathfrak{L}(Y)$ is typically referred to as `proper', `circular' or `complex symmetric' \citep{neeser1993proper, picinbono1996second, Kent1994} and the simplification by taking a complex approach becomes evident: 

\begin{theorem}
	\label{lem:bivariateCov_cplxsymmetry}
	A random element $Y$ in $\mathbb{L}^2(\mathcal{T}, \mathbb{C})$ with covariance operator $\Sigma$ with eigenbasis $\{e_k\}_k$ and corresponding eigenvalues $\{\lambda_k\}_k$  follows a complex symmetric distribution 
	if and only if all scores $Z_k = \langle e_k, Y - \mu\rangle$ with $\lambda_k>0$ do, and additionally the mean of  $Y$ is $\mu=0$. 
	In this case,
	\begin{enumerate}[i)]
		\item \label{lem:bivariateCov1} the pseudo-covariance $\Omega$ of $Y$ vanishes, i.e.\ $\Omega(y) = 0$ for all $y\in\mathbb{L}^2(\mathcal{T}, \mathbb{C})$, and the covariance operator $\mathbf{\Sigma}$ of the bivariate process $\mathbf{Y}=\kappa(Y)$ is completely determined by $\Sigma$;
		\item \label{lem:bivariateCov2} the pairs $\mathbf{e}_k = \kappa(2^{-1/2} e_k), \mathbf{e}_{-k} = \kappa(\im\, 2^{-1/2} e_k) \in \mathbb{L}^2(\mathcal{T}, \mathbb{R}^2)$ yield an eigendecomposition $\mathbf{\Sigma}(\mathbf{f}) = \sum_{k\neq 0} \lambda_k  \langle \mathbf{e}_k, \mathbf{f} \rangle \mathbf{e}_k$ of $\mathbf{\Sigma}$. With probability one, $\mathbf{Y} = \sum_{k\neq 0} \mathbf{e}_k \mathbf{Z}_k$ with uncorrelated real scores $\mathbf{Z}_k$ with mean zero, variance $\operatorname{var}\left(\mathbf{Z}_k\right) = \lambda_{k}$ and $\kappa(Z_k)=(\mathbf{Z}_k, \mathbf{Z}_{-k})^\top$.  	
	\end{enumerate}
\end{theorem}

While rotation invariance of $\mathfrak{L}(\mathbf{Y})$ leads to even multiplicities in the eigenvalues of the bivariate covariance operator $\mathbf{\Sigma}$, it does not pose a constraint on the complex eigenvalues and eigenfunctions of $\Sigma$, which would complicate the eigendecomposition. Here, rotation invariance of $\mathfrak{L}(\mathbf{Y})$ instead translates to complex symmetry of the distribution of the scores $Z_k$. 

Mean and covariance structure of $Y$ can also be approached from the point-wise mean $\mu^*(t) = \mathbb{E}\left(Y(t)\right)$ and Hermitian covariance surface $C(s,t) = \mathbb{E}\left((Y(s) - \mu^*(t))^\dagger(Y(t) - \mu^*(t))\right) = C(t,s)^\dagger$. Under complex symmetry, we obtain again $\mu^*(t) = 0$, while the auto-covariances $\mathbb{E}(\Re(Y(s)) \Re(Y(t))) = \mathbb{E}(\Im(Y(s)) \Im(Y(t))) =\Re\left(C(s,t)\right)$ and  cross-covariances $\mathbb{E}(\Re(Y(s)) \Im(Y(t))) = - \mathbb{E}(\Im(Y(s)) \Re(Y(t))) = \Im\left(C(s,t)\right)$ of the bivariate $\mathbf{Y}$ are completely determined by $C(s,t)$, as shown in the \Supplement. 
The integral operator $\Sigma^*(f)(t) = \int C(s,t) f(s) \, d\nu(s)$ on $\mathbb{L}^2(\mathcal{T}, \mathbb{C})$ induced by the covariance surface again constitutes a compact and self-adjoint operator and admits, as such, an eigendecomposition. In fact, under standard assumptions such as continuity of $\mu^*(t)$ and $C(s,t)$, Fubini allows switching integrals and the point-wise mean $\mu^* = \mu$ coincides with the mean element and the operator $\Sigma^* = \Sigma$ with the covariance operator. In this case, the eigendecomposition of $\Sigma$ also yields a decomposition 
\begin{equation}\nonumber
	\label{eigendecomposition}
	C(s,t) = \sum_{k\geq1} \lambda_k e_k^\dagger(s) e_k(t)
\end{equation}
of the covariance surface.

\subsection{Hermitian covariance estimation via tensor-product smoothing}
\label{sec:covariance:covariance_smoothing}

Based on a densely/regularly sampled collection of realizations $y_1, \dots, y_n: \mathcal{T} \rightarrow \mathbb{C}$ (with equal grids) of a complex symmetric process $Y$, the covariance surface $C(s,t)$ of $Y$ can be estimated by the empirical covariance surface
$\hat{C}_{emp.}(s,t) = \frac{1}{n} \sum_{i=1}^{n}  y_i^\dagger(s) y_i(t)$ for each pair of grid-points $s,t$.
This is, however, not possible in a sparse/irregular setting where only a limited number of evaluations $y_i(t_{i1}) = y_{i1}, \dots, y_i(t_{in_i}) = y_{in_i}$ are available for $i=1,\dots, n$ such that, for a given $(s,t)$-tuple, $\hat{C}_{emp.}(s,t)$ would only be based on few observations if computable at all. Consequently, some kind of smoothing over samples becomes necessary and, following the seminal work of \citet{Yao:etal:2005}, covariance estimation in the sparse/irregular functional case has widely been approached as a non-/semi-parametric regression problem. We proceed accordingly in the complex case and model $\mathbb{E}\left(Y^\dagger(s) Y(t) \right) = C(s,t)$ with a (smooth) regression estimator $\hat{C}(s,t)$ fitted to response products $y_{ij}^\dagger y_{i\ddot{\jmath}}$ at respective tuples $(t_{ij}, t_{i\ddot{\jmath}})\in\mathcal{T}^2$, for $j,\ddot{\jmath} = 1, \dots, n_i$ and $i=1,\dots, n$. 
Here, it is often reasonable to assume that, in fact, only measurements $\tilde{y}_{ij} = y_{ij} + \varepsilon_{ij}$ are observed with $\varepsilon_{ij} = \varepsilon_i(t_{ij})$ uncorrelated measurement errors originating from a white noise error process $\varepsilon(t)$, $t\in \mathcal{T}$. This leads to a combined covariance $\widetilde{C}(s,t) = C(s,t) + \tau^2(t)\, \ind_{\{s\}}(t)$ with $\tau^2(t) = \operatorname{var}(\varepsilon(t))$ the variance function of $\varepsilon(t)$. Assuming $C(s,t)$ continuous, $\tau^2(t)$ can be distinguished as a discontinuous ``nugget effect" at $s=t$. 

Generalizing the approach of \citet{cederbaum2018fast} for real covariance surfaces to the complex case, we propose to model $C(s,t)$ using a Hermitian tensor-product smooth 
\begin{equation}\nonumber
	C(s,t) \approx \sum_{g=1}^{m} \sum_{k=1}^{m} \xi_{gk} f_g(s) f_k(t) = \mathbf{f}^\top\!(s)\, \mathbf{\Xi}\, \mathbf{f}(t) = \operatorname{vec}(\mathbf{\Xi})^\top (\mathbf{f}(t) \otimes \mathbf{f}(s)) 
\end{equation}
with real-valued basis functions $f_k: \mathcal{T} \rightarrow \mathbb{R}$, $k=1, \dots, m$, stacked to a vector $\mathbf{f}(t) = (f_1(t), \dots, f_m(t))^\top$, and a Hermitian coefficient matrix $\mathbf{\Xi} = \{\xi_{kk'}\}_{kk'} = \mathbf{\Xi}^\dagger \in \mathbb{C}^{m\times m}$ ensuring ${C}(s,t)$ is Hermitian as required, with $\operatorname{vec}$ stacking the columns of a matrix to a vector.
Both the symmetry of the real part $\Re(\mathbf{\Xi})=\Re(\mathbf{\Xi})^\top$ and the anti-symmetry of the imaginary part $\Im(\mathbf{\Xi})=-\Im(\mathbf{\Xi})^\top$ present linear constraints. As such they can be implemented via suitable basis transforms $\mathbf{D}_{\Re}\, ({\mathbf{f} \otimes \mathbf{f}})(s,t)$ and $\mathbf{D}_{\Im}\, ({\mathbf{f} \otimes \mathbf{f}})(s,t)$ of the tensor-product basis $({\mathbf{f} \otimes \mathbf{f}})(s,t) = (f_1(s) \mathbf{f}^\top\!(t), \dots, f_m(s) \mathbf{f}^\top\!(t))^\top$ with transformation matrices $\mathbf{D}_{\Re} \in \mathbb{R}^{{(m^2+m)}/{2} \times m^2}$ and $\mathbf{D}_{\Im} \in \mathbb{R}^{{(m^2-m)}/{2} \times m^2}$ for the symmetric and anti-symmetric part, respectively. Since $\mathbb{R}^{m\times m}$ is a direct sum of the vector spaces of symmetric and antisymmetric $m\times m$ matrices, $\mathbf{D}_{\Im}$ can be obtained, e.g., as basis matrix of the null space of $\mathbf{D}_{\Re}$. A possible construction of $\mathbf{D}_{\Re}$ is described by \citet{cederbaum2018fast}. 
In addition to the covariance, we also model the error variance $\tau^2(t) \approx {\boldsymbol{\xi}}_{\tau}^\top \mathbf{f}_\tau(t)$ expanded in a real function basis $\mathbf{f}_\tau(t)$. Here, it might be convenient to employ the same basis $\mathbf{f}_\tau(t) = \mathbf{f}(t)$, or to assume constant error variance by setting $\mathbf{f}_\tau(t) = 1$ for all $t$. At any $t$ with $\tau^2(t) = 0$, the measurement error is excluded from the model.  
The coefficients $\operatorname{vec}\big(\widehat{\mathbf{\Xi}}\big) = \mathbf{D}_{\Re} \hat{\boldsymbol{\xi}}_{\Re} + \im\, \mathbf{D}_{\Im} \hat{\boldsymbol{\xi}}_{\Im}$ of the covariance estimator $\hat{C}(s,t)$ minimize the penalized least-squares criterion 
\begin{equation}\nonumber
	\textsc{pls}(\mathbf{\Xi}, \boldsymbol{\xi}_{\tau}) = \sum_{i, j, \ddot{\jmath}} \left| \mathbf{f}^\top\!(t_{ij})\, \mathbf{\Xi}\, \mathbf{f}( t_{i\ddot{\jmath}}) + {\boldsymbol{\xi}}_{\tau}^\top \mathbf{f}_\tau(t_{ij}) \, \ind_{\{\ddot{\jmath}\}}(j) - y_{ij}^\dagger y_{i\,\ddot{\jmath}} \right|^2 + \textsc{pen}(\mathbf{\Xi}, \boldsymbol{\xi}_{\tau})
\end{equation}
with quadratic penalty term $\textsc{pen}$. They are seperately obtained for the real and imaginary part  of the covariance using  $\textsc{pls} = \textsc{pls}_{\Re} + \textsc{pls}_{\Im}$ via the well-known linear estimators $\hat{\boldsymbol{\xi}}_{\Re} \in \mathbb{R}^{{(m^2+m)}/{2}}$, $\hat{\boldsymbol{\xi}}_{\tau} \in \mathbb{R}^{m_\tau}$ minimizing $\textsc{pls}_{\Re} = \sum_{i, j, \ddot{\jmath}} 
(
{\boldsymbol{\xi}}_{\Re}^\top\, \mathbf{D}_{\Re}\,({\mathbf{f} \otimes \mathbf{f}})(t_{ij},t_{i\ddot{\jmath}}) + \boldsymbol{\xi}_\tau^\top \mathbf{f}(t_{ij}) \, \ind_{\{\ddot{\jmath}\}}(j)
- \Re(y_{ij}^\dagger y_{i\,\ddot{\jmath}})
)^2 + 
\eta_{\Re}\, {\boldsymbol{\xi}}_{\Re}^\top \mathbf{D}_{\Re} \mathbf{P}_{\otimes}\,\mathbf{D}_{\Re}^\top {\boldsymbol{\xi}}_{\Re}
+ 
\eta_{\tau}\, {\boldsymbol{\xi}}_{\tau}^\top \mathbf{P}_{\tau}\, {\boldsymbol{\xi}}_{\tau}$, and $\hat{\boldsymbol{\xi}}_{\Im} \in \mathbb{R}^{({m^2-m})/{2}}$ minimizing
$\textsc{pls}_{\Im} = \sum_{i, j, \ddot{\jmath}} 
(
{\boldsymbol{\xi}}_{\Im}^\top\, \mathbf{D}_{\Im}\,({\mathbf{f} \otimes \mathbf{f}})(t_{ij},t_{i\ddot{\jmath}})
- \Im(y_{ij}^\dagger y_{i\,\ddot{\jmath}})
)^2 + 
\eta_{\Im}\, {\boldsymbol{\xi}}_{\Im}^\top \mathbf{D}_{\Im} \mathbf{P}_{\otimes}\,\mathbf{D}_{\Im}^\top {\boldsymbol{\xi}}_{\Im}$.
Smoothing parameters $\eta_{\tau},\eta_{\Re}, \eta_{\Im} > 0$ control the penalty induced by the matrices $\mathbf{P}_\tau$ and  $\mathbf{P}_{\otimes} = \mathbf{P} \otimes \mathbf{I}_{m} + \mathbf{I}_{m} \otimes \mathbf{P}$ constructed from a suitable penalty matrix $\mathbf{P} \in \mathbb{R}^{m\times m}$ for the basis coefficients of $\mathbf{f}(t)$ and the $m\times m$ identity matrix $\mathbf{I}_m$.
Assuming the error variance not too heterogeneous over $t$, the matrix $\mathbf{P}_\tau$ should typically penalize deviations from the constant. 
Based on a working normality assumption,  $\eta_{\Re}, \eta_{\tau}$ and $\eta_{\Im}$ are obtained via restricted maximum likelihood (\REML) estimation \citep{Wood2017}, avoiding computationally intense hyper-parameter tuning. 
For practical use, we extended the R package \texttt{sparseFLMM} \citep{sparseFLMM} to also offer  anti-symmetric tensor-product smooths for the package \texttt{mgcv} \citep{Wood2017} used for estimation. For asymptotic theory on the used penalized spline estimators, please see \citet{WoodEtAl2016GeneralSmoothModels}.

After estimation, eigenfunctions $e_k$ and eigenvalues $\lambda_k$ of the covariance operator $\Sigma$ of $Y$ are estimated by the corresponding eigendecomposition $\hat{C}(s,t) = \sum_{k\geq1}\hat{\lambda}_k \hat{e}_k^\dagger(s) \hat{e}_k(t)$ of the respective covariance operator $\hat{\Sigma}$. 
Based on $\hat{\mathbf{\Xi}}$ and the Gram matrix $\mathbf{G} = \{\langle f_k, f_{k'} \rangle\}_{k,k'=1}^m$, the right eigenvalues of the matrix $\mathbf{G}^{-1}\hat{\mathbf{\Xi}}$ yield the eigenvalues $\hat{\lambda}_k$ of $\hat{\Sigma}$. The corresponding eigenvectors $\hat{\boldsymbol{\theta}}_k$ yield the eigenfunctions $\hat{e}_k(t) = \hat{\boldsymbol{\theta}}_k^\top \mathbf{f}(t)$ of $\hat{\Sigma}$ for $k=1, \dots, m$. 
To ensure positive-definiteness, eigenfunctions with $\lambda_{k} \leq 0$ are omitted from the basis. Nonnegativity of $\tau^2$ is enforced post-hoc by setting negative values to zero.

\section{Elastic full Procrustes analysis \label{sec:fullproc}}

\subsection{Full Procrustes analysis in the square-root-velocity framework}
\label{sec:fullproc:general}

\newcommand{\arc}[1]{\check{#1}}

To now  propose (elastic) full Procrustes means  for plane curves, we first introduce some underlying concepts and notation.
We understand a \textit{parameterized} curve as a function $\beta: [0,1] \rightarrow \mathbb{C}$, which is assumed absolutely continuous such that the component-wise derivative $\dot{\beta}(t) = \frac{d}{dt} \Re\circ \beta(t) + \im \frac{d}{dt} \Im\circ \beta(t)$ exists almost everywhere and also the integral $\varphi_\beta(t) = \int_{0}^{t} |\dot{\beta}(s)|\, ds < \infty$ exists for $t\in[0,1]$. 
Denoting the set of absolutely continuous functions $[0,1] \rightarrow \mathbb{C}$ by $\mathcal{AC}([0,1], \mathbb{C})$, 
we further assume $\beta \in \mathcal{AC}^*([0,1], \mathbb{C}) = \mathcal{AC}([0,1], \mathbb{C}) \setminus \{t \mapsto z : z \in \mathbb{C}\}$ excluding constant functions as degenerate curves.
Then $\beta$ has positive length $L(\beta) = \varphi_\beta(1) > 0$, and a constant-speed parameterization $\alpha = \beta \circ \varphi_\beta^{-1}$ always exists, when taking the generalized inverse $\varphi_\beta^{-1}(s) = \inf\{t \in [0,1]: s\,L(\beta) \leq \varphi_{\beta}(t)\}$, $s\in [0,1]$. 
Two parameterized curves $\beta_1, \beta_2\in \mathcal{AC}^*([0,1], \mathbb{C})$ are said to describe the same curve if they have the same constant-speed parameterization $\alpha_1 = \alpha_2$, which yields an equivalence relation $\beta_1 \approx \beta_2$.
An \emph{oriented} curve is then defined as equivalence class with respect to `$\approx$'. If the context allows it, we commonly refer to both oriented plane curves and their parameterized curve representatives $\beta$ simply as ``curve".  
A diffeomorphism $\gamma: [0,1] \rightarrow [0,1]$ which is orientation-preserving, i.e., with derivative $\dot{\gamma}(t) > 0$ for $t\in[0,1]$, is called warping function and the set of such warping functions is denoted by $\Gamma$. With obviously $\beta \circ \gamma \approx \beta$, warping can equivalently be used to define equivalence of parameterized curves \citep[see, e.g,][which we also recommend for further details]{bruveris2016optimal}.
Abstracting also from the particular coordinate system for $\mathbb{C}$, the shape of an (oriented) curve with parameterization $\beta$ is then defined by $[ \beta ] = \{ \tilde{\beta} \in \mathcal{AC}([0,1], \mathbb{C}) :  u\, \tilde{\beta} + v \approx \beta \text{ for some }u,v\in\mathbb{C}\} $, its equivalence class under translation by any $v$, rotation by $u/|u| = \exp(\im\, \omega), \omega\in\mathbb{R}$, re-scaling by $|u|$, and warping.
This presents our ultimate object of interest.
In establishing a metric on the quotient space $\mathfrak{B} = \{[\beta] : \beta \in \mathcal{AC}^*([0,1], \mathbb{C})\}$, we follow and extend the idea of the full Procrustes distance in landmark shape analysis and define
\begin{equation}
	\label{eq:FullProc}
	d_{\Psi}([ \beta_1 ], [ \beta_2]) = \inf_{\substack{a \geq 0, v_i \in \mathbb{C},\\ \omega_i \in \mathbb{R}, \gamma_i \in \Gamma}}
\|\ \Psi\left(\exp(\im\,\omega_1)\ \beta_1\circ\gamma_1 + v_1\right) - a\ \Psi\left(\exp(\im\,\omega_2)\ \beta_2\circ\gamma_2 + v_2\right)\ \|  
\end{equation}
for $\beta_1, \beta_2 \in \mathcal{AC}^*([0,1], \mathbb{C})$, with a pre-shape map $\Psi:\mathcal{AC}^*([0,1], \mathbb{C}) \rightarrow \mathbb{L}^2([0,1], \mathbb{C}),\ \beta \mapsto q$ discussed below allowing to base computation on the $\mathbb{L}^2$-metric while optimizing over all involved invariances.
Acting differently than the other curve-shape preserving transformations \citep[see, e.g.,][Chap. 3.7]{SrivastavaKlassen2016}, scale invariance is generally accounted for by a normalization constraint $\|\Psi(\beta)\| = \|q\| = 1$ for all $\beta$. Fixing $a=1$ in (\ref{eq:FullProc}) would yield a partial-Procrustes-type distance instead. Replacing also the norm by the arc length on the $\mathbb{L}^2$-sphere would correspond to an intrinsic shape distance.
To obtain a proper and sound metric, $\Psi$ has to be carefully chosen. It is well-known that directly applying the $\mathbb{L}^2$-metric on the level of parameterized curves $\beta$ is problematic, since in this case the warping action of $\gamma \in \Gamma$ is not by isometries \citep{SrivastavaKlassen2016}. 

We set $\tilde{\Psi}(\beta)$ to the \SRV-transformation \citep{Srivastava2011ShapeAnaofElasticCurves}, representing a curve $\beta$ by its square-root-velocity (\SRV) transform $q:[0,1] \rightarrow \mathbb{C}$ given by $q(t) = {\dot{\beta}(t)}/{|\dot{\beta}(t)|^{1/2}}$ wherever this is defined and $q(t) = 0$ elsewhere. 
Indeed, $q$ is square-integrable with $\|q\|^2 = \int_{0}^{1} |q(t)|^2\, dt = L(\beta)$. Since $\tilde{\Psi}\left(u\, \beta \circ \gamma + v\right)(t) = ({u}/{|u|^{1/2}})\, q \circ \gamma(t) {\dot{\gamma}(t)}^{1/2}$, warping and rotation act by isometries with $\|\tilde{\Psi}(a \exp(\im\, \omega)\, \beta_1 \circ \gamma + v) - \tilde{\Psi}(a \exp(\im\, \omega)\, \beta_2  \circ \gamma + v)\| = {a}^{1/2} \|\tilde{\Psi}(\beta_1) - \tilde{\Psi}(\beta_2)\|$ for any two curves $\beta_1, \beta_2$ and $\gamma\in\Gamma$, $a\geq 0$, $\omega\in\mathbb{R}$, $u,v\in\mathbb{C}$. 
The $\mathbb{L}^2$-metric on the \SRV-transforms induces a metric on the space of parameterized curves modulo translation \citep{bruveris2016optimal}. It is commonly referred to as ``elastic" metric due to the isometric action of $\gamma$ allowing to construct a metric on oriented curves via optimal warping alignment.
$\tilde{\Psi}$ is surjective but not injective, with $\tilde{\Psi}^{-1}(\{\tilde{\Psi}(\beta)\}) = \{\beta + v : v \in \mathbb{C} \} \subset [\beta]$. Without loss of generality, we can, thus, set $\tilde{\Psi}^{-1}(q)(t) = \int_{0}^{t} \dot{\beta}(s) \, ds = \int_{0}^{t} q(s) |q(s)| \, ds$ 
when discussing shapes $[\beta]$. 

\begin{proposition}
\label{lem:efullProcDist}
With $\Psi(\beta) = \tilde{\Psi}({\beta}/{L(\beta)}) = {\tilde{\Psi}(\beta)}/{\|\tilde{\Psi}(\beta)\|}$ the normalized \SRV-transform, $d_\Psi$ defines a metric on $\mathfrak{B}$, referred to as \textit{elastic} full Procrustes distance $d_\mathcal{E}$. It takes the form
\begin{equation}\nonumber
	d_\mathcal{E}^2([ \beta_1 ], [ \beta_2]) = \inf_{u\in \mathbb{C}, \gamma \in \Gamma} \|q_1 - u\, q_2\circ\gamma\, \dot{\gamma}^{1/2}\|^2 =  1 - \sup_{\gamma \in \Gamma} | \langle q_1, q_2 \circ \gamma´\, \dot{\gamma}^{1/2} \rangle |^2
\end{equation}
for $q_i = \Psi(\beta_i)$ unit-norm \SRV-transforms of curve shape representatives $\beta_1, \beta_2 \in \mathcal{AC}^*([0,1], \mathbb{C})$. 
\end{proposition}

With a metric at hand, we may proceed by considering random shapes and define the concept of a Fréchet mean induced by the metric \citep[compare, e.g.,][]{Huckemann2012meaning, ziezold1977expected}. A random element $A$ in a metric space $(\mathfrak{A}, d)$ is a Borel-measurable random variable taking values in $\mathfrak{A}$. A (population) Fréchet mean or expected element $\mathfrak{m} \in \mathfrak{A}$ is defined as a minimizer of the expected square distance
\begin{equation}\nonumber
\mathbb{E}\left(d^2(\mathfrak{m}, A) \right) = \sigma^2 = \inf_{\mathfrak{a} \in \mathfrak{A}} \mathbb{E}\left(d^2(\mathfrak{a}, A) \right).
\end{equation}
assuming a finite variance $\sigma^2 < \infty$.

\begin{definition}
A \textit{random (plane curve) shape} $[B]$ is a random element in the shape space $\mathfrak{B}$ equipped with the elastic full Procrustes distance $d_\mathcal{E}$.
We call a Fréchet mean $[\mu_\mathcal{E}] \in \mathfrak{B}$ of $[B]$, represented by $\mu_\mathcal{E} \in \mathcal{AC}^*([0,1], \mathbb{C})$, an \textit{elastic full Procrustes mean} of the random shape $[B]$. 
\end{definition}

As distance computation is carried out on \SRV-transforms, 
it is, however, typically more convenient to consider the mean shape on \SRV-level, i.e.\, via a distribution $\mathfrak{L}\left(Q\right)$ of a random element $Q=\Psi(B)$ in the Hilbert space $\mathbb{L}^2([0,1], \mathbb{C})$ inducing the shape distribution $\mathfrak{L}([B])$.

\begin{proposition}
\label{lem:efullProcMean}
Consider a random element $Q$ in $\mathbb{L}^2([0,1], \mathbb{C})$ with $\|Q\|=1$ almost surely. The elastic full Procrustes means $[\mu_\mathcal{E}]$ of the induced random shape $[B] = [\Psi^{-1}(Q)]$ are determined by their \SRV-transform $\psi_\mathcal{E} = \Psi(\mu_\mathcal{E})$ fulfilling
\begin{equation}
	\label{eq:SRV_FullProc_elastic}
	\psi_\mathcal{E} \in \underset{y : \|y\| = 1}{\operatorname{argmax}} \, \mathbb{E}\big( \sup_{\gamma \in \Gamma} | \langle y, Q \circ \gamma\, \dot{ \gamma}^{1/2} \rangle |^2  \big) = \underset{y : \|y\| = 1}{\operatorname{argmax}} \, \mathbb{E}\big( \sup_{\gamma \in \Gamma}  \langle y, Q \circ \gamma\, \dot{ \gamma}^{1/2} \rangle \langle Q \circ \gamma\ \dot{ \gamma}^{1/2}, y  \rangle  \big).
\end{equation}
\end{proposition}

When fixing $\gamma$ in Equation \eqref{eq:SRV_FullProc_elastic}, the maximum of the quadratic form is obtained at the leading eigenvector of the covariance operator of $Q\circ\gamma\, \dot{\gamma}^{1/2}$, which is carried out in detail in Proposition \ref{lem:ifullProcMean} 
considering \emph{inelastic} full Procrustes means of shapes of parameterized plane curves.
This allows use of Hermitian covariance smoothing, introduced in Section \ref{sec:covariance:covariance_smoothing}, for shape mean estimation. 
Inelastic mean estimation will present a building block in elastic mean estimation but is also interesting in its own right, especially in data scenarios involving natural curve parameterizations. 

\begin{proposition}
\label{lem:ifullProcMean}
For $\beta \in \mathcal{AC}^*([0,1], \mathbb{C})$ define the shape of a \emph{parameterized} plane curve as $(\beta) = \{u\, \beta + v : u,v\in\mathbb{C}\}$. Then
\begin{enumerate}[i)]
	\item the \emph{inelastic} full Procrustes distance $d_{\not \mathcal{E}}((\beta_1), (\beta_2)) = \inf_{u\in \mathbb{C}} \| q_1 - u q_2\|$ with $\|q_i\| = 1$ for $\Psi(\beta_i) = q_i$, $i=1,2$, defines a metric on the shape space $\tilde{\mathfrak{B}} = \{(\beta) : \mathcal{AC}^*([0,1], \mathbb{C})\}$ of parameterized plane curves and can be expressed as $d^2_{\not \mathcal{E}}((\beta_1), (\beta_2)) = 1 - |\langle q_1,  q_2\rangle|^2$;
	\item multiplication by $\langle q_1, q_2\rangle^\dagger/|\langle q_1, q_2\rangle| = \operatorname{argmin}_{u:|u|=1} \|q_1 - u q_2\|$ yields rotation alignment of $\beta_2$ to $\beta_1$;
	\item \label{lem:ifullProcMean:eigen} for a complex symmetric random element $Q$ in $\mathbb{L}^2([0,1], \mathbb{C})$ with covariance operator $\Sigma$, let $\mathcal{Y}_1 = \{y : \Sigma(y) = \lambda_1 y\}$ denote the spectrum of the leading eigenvalue $\lambda_1$ of $\Sigma$. Then, $(\mathcal{Y}_1) = \{(y) : y\in\mathcal{Y}_1\}$ is the set of Fréchet means of the random shape $(B)=(\Psi^{-1}(Q))$ in $\tilde{\mathfrak{B}}$ with respect to $d_{\not \mathcal{E}}$, which we refer to as \textit{inelastic} full Procrustes means.
	In particular, the leading eigenfunction $\psi_{\not \mathcal{E}} = e_1$ of an eigendecomposition of $\Sigma$ yields an inelastic full Procrustes mean $(\mu_{\not \mathcal{E}})$ of $(B)$ with \SRV-transform $\psi_{\not \mathcal{E}} = \Psi(\mu_{\not \mathcal{E}})$. It is unique if $\lambda_1$ has multiplicity 1. The variance of $(B)$ is $\sigma_{\not \mathcal{E}}^2 = \mathbb{E}(d^2_{\not \mathcal{E}}((\mu_{\not \mathcal{E}}), (B))) = 1-\lambda_1$.
\end{enumerate}
\end{proposition}

Motivated by Proposition \ref{lem:ifullProcMean} \ref{lem:ifullProcMean:eigen}), we propose to estimate $\psi_{\not \mathcal{E}}$ as leading eigenfunction $\hat{e}_1$ of $\hat{C}(s,t)$ obtained by Hermitian covariance smoothing in Section \ref{sec:fullproc:estimation}, as part of the estimation procedure of $\psi_{\mathcal{E}}$. 
However, before that
we first address the question of how it is still possible to work with derivative-based SRV-curves even in the sparsely observed setting so common in practice. 


\subsection{The square-root-velocity representation in a sparse/irregular setting}
\label{sec:fullproc:sparseSRV}

In practice, the shape of an (oriented) plane curve is observed via a vector $\mathbf{b} = (b_{0},\dots b_{n_0})^\top\in\mathbb{C}^{n_0+1}$ of points, which can be considered evaluations $\beta^*(t^*_j) = b_j$ of some continuous parameterization $\beta^*: [0,1] \rightarrow \mathbb{C}$ of the curve at arbitrary time points $t^*_0 < \dots < t^*_{n_0}$. 
However, fixing the time grid, the derivatives $\dot{\beta}^*(t^*_i)$ are not observable. Instead, evaluations of an \SRV-transform describing the curve can be directly obtained from the finite differences $\Delta_j = b_{j} - b_{j-1}$, 
if the curve segments $\beta^*\left((t^*_{j-1}, t^*_j)\right) \subset \mathbb{C}$ between the observed points in $\mathbf{b}$ have no edges or loops: 

\begin{theorem}[Feasible sampling]
\label{lemma:feasible_sampling}
If $\beta^*$ is continuous and $\beta^*: (t_{j-1}^*, t_j^*) \rightarrow \mathbb{C}$ is injective and continuously differentiable with $\dot{\beta^*}(t)\neq 0$ for all $t\in(t_{j-1}^*, t_j^*)$, for $j = 1, \dots, n_0$, then for any time points $0 < t_1 < \dots < t_{n_0}< 1$ and speeds $w_1, \dots, w_{n_0}>0$, there exists a $\gamma \in \Gamma$ such that 
\begin{equation}\nonumber
	q(t_j) = w_j^{1/2} \, (\beta^*(t^*_{j}) - \beta^*(t^*_{j-1})) = w_j^{1/2} \, \Delta_j \quad  (j = 1, \dots n_0)
\end{equation}
for the \SRV-transform $q$ of $\beta = \beta^* \circ \gamma$.	
\end{theorem}

We call a vector of sampling points $\mathbf{b}$ of a curve \emph{feasible} if the conditions of Lemma \ref{lemma:feasible_sampling} hold. This is always fulfilled if there is a $\beta^* \in (\beta)$ such that $\beta^*$ is continuously differentiable with non-vanishing derivative on all $(0,1)$ and, in particular, if it describes an embedded one-dimensional differentiable submanifold of $\mathbb{R}^2$. If, instead, the curve has edges, they must be contained in $\mathbf{b}$, as well as a point inside of each loop (i.e.\ within each closed curve segment). 

Note that while discrete observations often result in approximate derivative computations, Theorem \ref{lemma:feasible_sampling} ensures that  the derivative-based SRV-transform can be {\it exactly} recovered on a  desired grid - up to a re-parameterization not essential in an analysis invariant to re-parameterization. 
Selected time points $t_1< \dots< t_{n_0}$ and speeds $w_1, \dots, w_{n_0}>0$ implicitly determine the parameterization. In principle, they could be arbitrarily selected due to parameterization invariance of the analysis, but with regard to mean estimation it is desirable to initialize them in a coherent way. 
Without any prior knowledge, constant speed parameterization of underlying curves $\beta$ presents a canonical choice. To approximate this, we borrow from constant speed parameterization $\hat{\beta}$ of the sample polygon with vertices $\mathbf{b}$, implying a piece-wise constant \SRV-transform $\hat{q}(t) = \sum_{j=1}^{n_0} q_j\, \ind_{[s_{j-1}, s_j)}(t)$ of $\hat{\beta}$ with {\SRV}s $q_j = \Delta_j |\Delta_j|^{-1} {L^{1/2}(\hat{\beta})}$, with $L(\hat{\beta}) = \sum_{j=1}^{n_0}  |\Delta_j|$ the length of the polygon. The nodes $s_j = \sum_{l=1}^{j} {|\Delta_j|}/{L(\hat{\beta})}$ indicate the vertices $\hat{\beta}(s_j) = b_j$, $j=0,\dots, n_0$. 
In accordance with that, we set $q(t_j) = q_j$ and select time points $t_j = {(s_j + s_{j-1})}/{2}$ in the center of the edges, for $j=1,\dots, n_0$. Depending on the context other choices might be preferable, but we generally expect this choice to imply reasonable starting parameterizations.

\subsection{Estimating elastic full Procrustes means via Hermitian covariance smoothing}
\label{sec:fullproc:estimation}

Consider a collection of sample vectors $\mathbf{b}_i \in \mathbb{C}^{n_i+1}$ of $n$ curves $\beta_i \in \mathcal{AC}^*([0,1], \mathbb{C}$), $i=1,\dots, n$, realizations of a random plane curve shape $[B]$.
For scale-invariance, sample polygons are normalized to unit-length. Moreover, the $\mathbf{b}_i$ are assumed feasibly sampled to represent them by evaluations $q_i(t_{ij}) = q_{ij}$ at time points $t_{ij}$, $j=1,\dots,n_i$, of the \SRV-transform $q_i$ of $\beta_i$ as described in the previous Section \ref{sec:fullproc:sparseSRV}.
We model an elastic full Procrustes mean $[\mu]$ of $[B]$ via the \SRV-transform $\psi$ of $\mu\in\mathcal{AC}^*([0,1], \mathbb{C})$ expanded as $\psi(t) = \sum_{k=1}^{m} \theta_k f_k(t) = \boldsymbol{\theta}^\top\mathbf{f}(t)$ in a basis $\mathbf{f}(t) = (f_1(t), \dotsm f_m(t))^\top$ of functions $f_k \in \mathbb{L}^2([0,1], \mathbb{R})$, $k=1,\dots, m$, 
with complex coefficient vector $\boldsymbol{\theta} = (\theta_1, \dots, \theta_m)^\top \in \mathbb{C}^m$. For the basis, piece-wise linear B-splines of order 1 present an attractive choice, since they have been proven identifiable under warping-invariance \citep{steyer2021elastic} while still implying continuity of $\psi$ and a differentiable mean curve $\mu$.

The idea of alternating between a) mean estimation on aligned data and b) alignment of the data to the current mean is used for estimation of landmark full Procrustes means \citep[p. 139]{DrydenMardia2016ShapeAnalysisWithApplications} and intrinsic elastic mean curve shapes \citep[p. 319]{SrivastavaKlassen2016}. 
We follow a similar strategy to find an estimator $\hat{\psi}(t) = \hat{\boldsymbol{\theta}}^\top\mathbf{f}(t)$ for $\psi$ but estimate an inelastic full Procrustes mean in a) and base the estimate on Hermitian covariance smoothing for irregularly/sparsely sampled curves. The covariance estimate is also used for estimating normalization and rotation alignment multipliers, which are not directly computable for sparse curve data. For warping alignment in b), we utilize the approach of \citet{steyer2021elastic}, which has proven suitable also for irregularly/sparsely sampled curves.
The single steps of the algorithm are detailed in the following and a discussion of its empirical performance is given in the next section.

\textbf{Initialize} in iteration $h=0$ \SRV-representations  $q^{[h]}_{i}(t_{ij}^{[h]}) = q^{[h]}_{ij}$  with $q^{[0]}_{ij} = q_{ij}$ and $t_{ij}^{[0]} = t_{ij}$ as in Section \ref{sec:fullproc:sparseSRV} for all $i,j$, and repeat the following steps for $h=1, 2, \dots$:
\begin{enumerate}[I.]
\item \label{alg:elasticFullProc:covariance} \textbf{Covariance estimation:} We estimate the covariance surface $C^{[h]}(s,t)$ of a complex symmetric process $Q$ underlying $q_1^{[h]}, \dots, q_n^{[h]}$ with a tensor-product estimator  $\hat{C}^{[h]}(s,t) = \mathbf{f}(s)^\top \hat{\mathbf{\Xi}}^{[h]}\, \mathbf{f}(t)$ with coefficient matrix $\hat{\mathbf{\Xi}}^{[h]} \in \mathbb{C}^{m\times m}$. 
While for dense sampling, an estimate can be directly obtained from the covariance of the $\langle q_i^{[h]}, f_{k}\rangle$ (see \Supplement), we propose Hermitian covariance smoothing as described in Section \ref{sec:covariance} for sparse/irregular data. This yields eigenfunctions $\hat{e}^{[h]}_k$ and eigenvalues $\hat{\lambda}_k^{[h]}$, $k=1,\dots,m$, of the corresponding covariance operator $\hat{\Sigma}^{[h]}$, as well as an estimate $\hat{\tau}^{2[h]}(t) \geq 0$ of the variance of a white noise zero mean residual process $\varepsilon(t)$ at $t\in [0,1]$, if measurement uncertainty on observations $Q(t_{ij}) + \varepsilon(t_{ij})$ is assumed. 

\item \label{alg:elasticFullProc:mean} \textbf{Mean estimation:} Set $\hat{\psi}^{[h]}(t)=\hat{e}^{[h]}_1(t) = \hat{\boldsymbol{\theta}}_1^{[h]\top} \mathbf{f}(t)$ to the leading eigenfunction of $\hat{\Sigma}^{[h]}$ obtained from the leading right eigenvector $\hat{\boldsymbol{\theta}}_1^{[h]}$ of $\mathbf{G}^{-1} \hat{\mathbf{\Xi}}^{[h]}$ with Gramian $\mathbf{G}$ of $\mathbf{f}$. 
This yields an inelastic full Procrustes mean estimate $[\hat{\mu}^{[h]}]=[\Psi^{-1}(\hat{\psi}^{[h]})]$ of the curves with the current parameterization (Proposition \ref{lem:ifullProcMean}), presenting the current estimate of the elastic full Procrustes mean.

\item \label{alg:elasticFullProc:rotation} \textbf{Rotation alignment and re-normalization:} 
For $u_i^{[h]} = \big({z^{[h]}_{i1}}/{|z_{i1}^{[h]}|}\big)^\dagger {(L^{[h]}(\beta_i))^{-1/2}}$ with $z^{[h]}_{i1}=\langle\hat{e}^{[h]}_1, q_i\rangle$, the multiplied $u_i^{[h]} q_i^{[h]}$ has norm $1$ and is rotation aligned to $\hat{\psi}^{[h]}$.
We estimate $u_i^{[h]}$ by $\hat{u}_i^{[h]}$ for $i=1,\dots, n$ based on the covariance estimation by plugging in conditional expectations $\hat{z}_{i1}^{[h]} = \mathbb{E}\big(\langle \hat{e}^{[h]}_1, Q \rangle \mid Q(t_{ij}) + \varepsilon(t_{ij}) = q^{[h]}_{ij},\, j=1,\dots, n_i \big)$ and $\hat{L}^{[h]}(\beta_i) = \mathbb{E}\big(\| Q \|^2 \mid Q(t_{ij}) + \varepsilon(t_{ij}) = q^{[h]}_{ij},\, j=1,\dots, n_i \big)$ under a working normality assumption, an estimation approach in the spirit of \citet{Yao:etal:2005}. Expressions can be found in the \Supplement.

\item \label{alg:elasticFullProc:warping} \textbf{Warping alignment:}
Based on its rotation aligned {\SRV} evaluations, the $i$th curve is (approximately) warping aligned to $\hat\mu^{[h]}$ using the approach of \citet{steyer2021elastic}, where \SRV-transforms are approximated as piece-wise constant functions $\hat{q}_i^{[h]}(t) \approx q_i^{[h]}(t)$ to find the infima of $\|\hat\mu^{[h]} - \hat{q}_i^{[h]} \circ \gamma_i \, \dot{\gamma}_i^{1/2}\|$ over $\gamma_1, \dots, \gamma_n\in\Gamma$. This yields new parameterization time-points $t_{ij}^{[h+1]}$, $j=1,\dots, n_i$, and corresponding {\SRV}s $q_{ij}^{[h+1]} = w_{ij}^{[h]} \hat{u}_{i}^{[h]} q_{ij}^{[h]}$, with  $w_{ij}^{[h]} > 0$ depending on the $t_{ij}^{[h]}$ and $t_{ij}^{[h+1]}$, passed forward to proceed with the next iteration at Step \ref{alg:elasticFullProc:covariance}.
Details can be found in the \Supplement.

\end{enumerate}

\textbf{Stop} the algorithm when $\|\hat{\psi}^{[h]} - \hat{\psi}^{[h-1]}\|$ is below a specified threshold in Step \ref{alg:elasticFullProc:mean}. An additional execution of Steps \ref{alg:elasticFullProc:rotation} and \ref{alg:elasticFullProc:warping} then yields rotation aligned representations of approximately unit-length curves and current time points.

\section{Adequacy and robustness of elastic full Procrustes mean estimation in realistic curve shape data}\label{sec:threespirals}

Familiar everyday shapes offer an ideal platform for evaluation of shape mean estimation, allowing for intuitive visual assessment of results. 
We consider three different such datasets for investigating the performance of elastic full Procrustes mean shape estimation 
and comparing it to other mean concepts: 
1. \texttt{digit3.dat} from \citet{DrydenMardia2016ShapeAnalysisWithApplications}, in R package \texttt{shapes}, comprising a total of 30 handwritten digits ``\textit{3}'' sampled at 13 landmarks each; 
2. irregularly sampled spirals $\beta(t) = t\,\exp(13\,\im\, t)$, $t\in[0,1]$, with random $n_i\in\{17, \dots, 22\}$ sampling points per spiral or with $n_i\in\{4, \dots, 7\}$ in a very sparse setting, additionally provided with small measurement errors and random rotation, translation and scaling;
and 3. handwritten letters ``\textit{f}'' extracted from the \texttt{handwrit} data in \citet{RamsaySilverman2005}, in R package \texttt{fda}, comprising 20 repetitions of the letter with a total of 501 samples per curve.
While we focus on one letter here for simplicity, example fits on the entire ``\textit{fda}'' writings can be found in Figure S1 in the Online Supplement.

Based on \texttt{digit3.dat}, 
we compare our elastic full Procrustes mean estimator $\hat{\mu}_{\mathcal{E}}$ with its inelastic analog $\hat{\mu}_{\not \mathcal{E}}$ and with an elastic curve mean estimator $\hat{\mu}_{\mathcal{C}}$ taking shape invariances not into account (fitted with R package \texttt{elasdics}). 
Moreover, we investigate fitting performance of $\hat{\mu}_{\mathcal{E}}$ for $n=4, 10, 30$ observed digits in a  simulation. 
All estimators are fitted using piece-wise constant and piece-wise linear B-splines with 13 equally spaced knots on \SRV-level applying 2nd order difference penalties in the covariance estimation for  $\hat{\mu}_{\mathcal{E}}$ and $\hat{\mu}_{\not \mathcal{E}}$. No penalty is available for $\hat{\mu}_{\mathcal{C}}$.
Figure \ref{fig:digitspirals} shows the estimates fitted on the first $n=4$ digits in the dataset. Without warping alignment, $\hat{\mu}_{\not \mathcal{E}}$ does not capture the pronounced central nose in the digit ``\textit{3}'' as distinctly as $\hat{\mu}_{\mathcal{E}}$. The difference is somewhat smaller when fitting on all $n=30$ digits (not shown), yet only marginally.
Since the data is roughly rotation and scaling aligned, $\hat{\mu}_{\mathcal{C}}$ is very close to $\hat{\mu}_{\mathcal{E}}$ when fitting on all digits. 
When fitting only on the first $n=4$ digits in the data, however, $\hat{\mu}_{\mathcal{C}}$ substantially deviates, in particular for the smooth estimator using linear splines, as shown in Figure \ref{fig:digitspirals} (top left). 
This can presumably be attributed to a) $\hat{\mu}_{\mathcal{C}}$ being more affected by the one outlying ``\textit{3}'' (top-left) than $\hat{\mu}_{\mathcal{E}}$, and b) the nose pointing into different directions depending on the handwriting.
Overall, deficiencies in warping and rotation alignment tend to mask features in the curve shapes by averaging over different orientations and parameterizations, similarly to the effect of measurement error in covariates in a regression model. 
With missing scale alignment, the shape of the estimated mean is mainly driven by the shape of the largest curve(s) in the data.
Good estimation quality is also confirmed in simulations that compare  elastic full Procrustes mean estimates $\hat{\mu}_l$, $l=1,\dots, 101$, estimated on independently drawn bootstrap samples of the digits (with $n=4, 10, 30$), with the mean $\mu$ estimated on the original dataset and taken as true mean. 
While single mean estimates for as few curves as $n=4$ might considerably deviate, the majority  visually resembles $\mu$ well, including $\hat{\mu}_{(0.75)}$ where $\hat{\mu}_{(a)}$ denotes the bootstrap estimator with $d_{(a)}$ the $a$-quantile of the distances $d_l = d_\mathcal{E}([\hat{\mu}_{l}], [\mu])$, $l=1,\dots,101$. 
Except for two outliers, all estimates with $n=10$ and $n=30$ are better than $\hat{\mu}_{(0.75)}$ for $n=4$ (Figure \ref{fig:digitspirals}, top middle). 

\begin{figure}
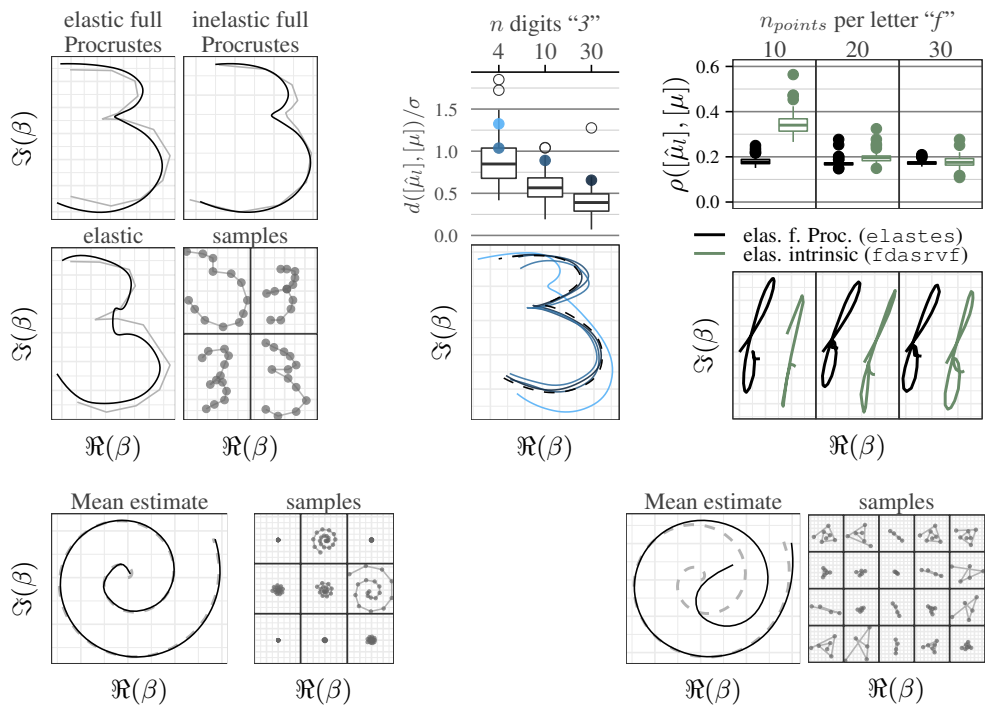

	\input{figures/digits.tex}
	\input{figures/digitsim.tex}
	\input{figures/f_comparison_dists}
	
	\vspace{-.3cm}
	\input{figures/spirals.tex}
	
	\caption{
		\textit{Top left:} Different digit ``\textit{3}'' mean curves (\textit{black}: order 1, \textit{grey}: order 0 B-splines on \SRV-level) estimated on the first $n=4$ sample polygons in \texttt{digit3.dat} shown in the bottom-right. 
		\textit{Top center:} Simulation results from 101-fold bootstrap samples of different sample sizes on \texttt{digit3.dat}. 
		Four bootstrap estimates as examples of cases with relatively high deviations from $\mu$ (95\% and for $n=4$ also 75\% distance quantiles) are depicted in the bottom and marked in the top panel (\textit{filled dots}).
		Here, distances to $[\mu]$ are provided relative to the standard deviation $\sigma$ estimated on the original dataset (as described below in Section \ref{sec:tongues}). However, in some sense, $\sigma$ is an underestimate as it does not include variation induced by irregular/sparse sampling.
		\textit{Top right:} Performance comparison of our elastic full Procrustes mean and the \texttt{fdasrvf} elastic intrinsic mean estimator based on 101-fold boostrap with $n_{points} = 10, 20, 30$ points sampled per letter ``\textit{f}''.
		Top shows the distribution of geodesic distances of estimated means to the overall intrinsic mean ``\textit{f}'' $[\mu]$ (computed with \texttt{fdasrvf}). For \texttt{fdasrvf}, three outliers for $n_{points} = 10$ and one for $n_{points} = 20$ above $0.6$ are omitted for the sake of visibility.    
		Bottom shows example means of median geodesic distance in each setting.
		\textit{Bottom:} Elastic full Procrustes means estimated on the spiral samples polygons displayed to their right, 
		in front of the original spiral (\textit{grey, dashed line}).
		\label{fig:digitspirals}
	}
\end{figure}

We illustrate the role of sparsity in shape mean estimation in the spiral data with its varying level of detail over the curve (i.e.\ varying curvature) and random irregular grids sampled roughly at constant angle distances (Figure \ref{fig:digitspirals}, bottom). Elastic full Procrustes mean estimates are based on piece-wise linear splines on \SRV-level with 20 knots and 2nd order penalties in covariance smoothing.
With a moderate number of sample points $n_i\in\{17, \dots, 22\}$, the estimate based on $n=9$ curves regains the original spiral shape close to perfectly. Only the inner end of the spiral with the most curvature shows some deviation.
With $n_i\in\{4,\dots,7\}$ 
and $n=20$, the estimator does not capture the higher curvature in the inner part of the spiral but otherwise fits its shape well despite extreme sparsity. 
In sparse functional data analysis, borrowing of strength across curves allows for consistent estimation of principle components based on a minimum number of sampling points $n_i$ for each curve under mild conditions \citep{Yao:etal:2005}.
However, this cannot equally be expected under shape invariances, as indicated by the fact that no shape information remains when curves are observed at $n_i < 3$ points, and in particular when warping-alignment can only be approximated on sparse samples.  
Still, we observe that bias becomes vanishingly small when the sampling points cover the curve sufficiently well. As this is often the case in real data, 
elastic full Procrustes mean estimation performs reliably well in practice already for comparably sparse data in our experience.

Based on $n=20$ handwritten letters ``\textit{f}'', we compare to R package \texttt{fdasrvf} \citep{fdasrvf}, which offers state-of-the-art elastic (intrinsic, not full Procrustes) shape mean estimation for regularly and densely observed curves.
To test different degrees of sparsity, we consider three scenarios with $n_{points} = 10, 20, 30$ sampling points per curve. For each, we draw $l = 1, \dots, 101$ bootstrap samples with $n_1 = \dots = n_{20} = n_{points}$ points subsampled from the total recorded points of each ``\textit{f}'' giving a higher acceptance probability to points important for curve reconstruction. This leads to datasets of sparse but still recognizable letters.
For all three settings our elastic full Procrustes mean estimator is fitted using piece-wise constant B-splines with 30 equally spaced knots on \SRV-level and applying a 2nd order difference penalty in the covariance estimation. This leads to polygonal means on curve level as in \texttt{fdasrvf} where the number of knots is, unlike in our approach, always equal to $n_{points}$.
As they estimate a different, intrinsic shape mean based on the elastic geodesic shape distance $\rho$, a fair comparison is not possible. We thus tailor the comparison to favor \texttt{fdasrvf} by comparing (also our full Procrustes) to their intrinsic shape mean on the full data, and using their  distance $\rho$.
Figure \ref{fig:digitspirals} (top right) illustrates performance based on their ``true mean'' $[\mu]$, estimated on the complete original data. 
In the very sparse $n_{points} = 10$ setting, differences in the mean concept are clearly dominated by the gain of using our mean estimator, which shows stable estimates gradually improving with $n_{points}$. 
With more densely observed curves the differences in fitting performance become smaller and the \texttt{fdasrvf} implementation gains a distinct computational advantage due to quadratic increase of the design matrix dimension in Hermitian covariance smoothing. 
While also in the $n_{points}=30$ scenario fitting time remains below 1.5 minutes on a standard computer, it can dramatically increase with the numbers of knots and sampling points. 
In dense scenarios, we, thus, recommend utilizing an alternative covariance estimator for elastic full Procrustes mean estimation as described in the Online Supplement. Still, also in this denser setting, our approach  estimating the elastic full Procrustes mean is at least as good in recovering the  elastic intrinsic mean as  \texttt{fdasrvf} which is, unlike our estimator, designed to estimate this mean.

\section{Phonetic analysis of tongue shapes}\label{sec:tongues}

The modulation of tongue shape presents an integral part of articulation \citep{hoole1999lingual}. Several authors investigate the shape variation in different phonetic tasks by analyzing tongue surface contours during speech production \citep{stone2001tongueSurfaceContours, iskarous2005patterns, davidson2006tongueshapes} to obtain insights into speech mechanics. They model tongue contour shapes with (penalized) B-splines fitted through points marked on the tongue surface in ultrasound or MRT images of the speaker profile. While different measures to register/superimpose the tongue contour curves are undertaken, shape and warping invariances are not explicitly incorporated into their statistical analysis so far. In particular, reducing tongue shapes to one dimensional curves over an angle as in \cite{davidson2006tongueshapes} brings the problem that the different functions (due to different tongue shapes for different sounds) extend over different angle domains, which is ignored in the analysis.
We suggest elastic full Procrustes analysis to appropriately handle the inherently two-dimensional curves. This approach accounts for the lack of a coordinate system in the ultrasound image, different positioning of ultrasound devices and size differences of speakers (Procrustes analysis) as well as flexibility of the tongue muscle to adjust its shape (elastic analysis).  We illustrate the approach in experimental data kindly provided by Marianne Pouplier: tongue contour shapes are recorded in an experimental setting
from six native German speakers ($\mathcal{S} = \{1,\dots,6\}$) repeating the same set of fictitious words, such as ``$\mathtt{pada}$", ``$\mathtt{pidi}$", ``$\mathtt{pala}$" or ``$\mathtt{pili}$". The words implement different combinations of two flanking vowels in $\mathcal{V}=\{\mathtt{a*a}, \mathtt{i*i}\}$ around a consonant in $\mathcal{C} = \{\mathtt{d}, \mathtt{l}, \mathtt{n}, \mathtt{s}\}$.
Each combination is repeated multiple times by each of the speakers (1-8 times), observing tongue contour shapes formed at the central time point of consonant articulation (estimated from the acoustic signal). 
In total, this yields $n=299$ sample polygons with nodes $\mathbf{b}_i\in\mathbb{C}^{n_i}$, $i=1,\dots, n$, each sampled at $n_i = 29$ points from the tongue root to the tongue tip. 
A feature vector $X_i = (v_i, c_i, s_i)^\top \in \mathcal{X} = \mathcal{V} \times \mathcal{C} \times \mathcal{S}$ identifies the word-speaker combination of the $i$th curve.
We investigate the different sources of shape variability (consonants, vowel context, speakers, repetitions) by elastic Full Procrustes analysis on different levels of hierarchy. 
Let $[\hat{\mu}_\mathcal{A}] \in \mathfrak{B}$ denote the elastic full Procrustes mean estimated for all $i$ with $X_i\in\mathcal{A} \subset \mathcal{X}$. 
Figure \ref{fig:tongues} depicts the overall shape mean $[\hat{\mu}_\mathcal{X}]$, separate means $[\hat{\mu}_{\{(c,v)\} \times \mathcal{S}}]$ for the consonants $c\in \{\mathtt{d}, \mathtt{s}\}$ in both vowel contexts $v\in\mathcal{V}$, and speaker-word means $[\hat{\mu}_{\{(c,v,s)\}}]$ reflecting individual articulation by speaker $s\in\mathcal{S}$. 
Not displayed consonants ``$\mathtt{l}$'' and ``$\mathtt{n}$'' yield very similar shapes as ``$\mathtt{d}$''. 
Shape means are estimated using linear B-splines on \SRV\ level with 13 equidistant knots and a 2nd order difference penalty for the basis coefficients. 
Homogeneous measurement error variance is assumed. 
Fitting the overall mean in this setting takes about 3 minutes on a standard computer. 

\begin{figure}[t]
	\centerline{
		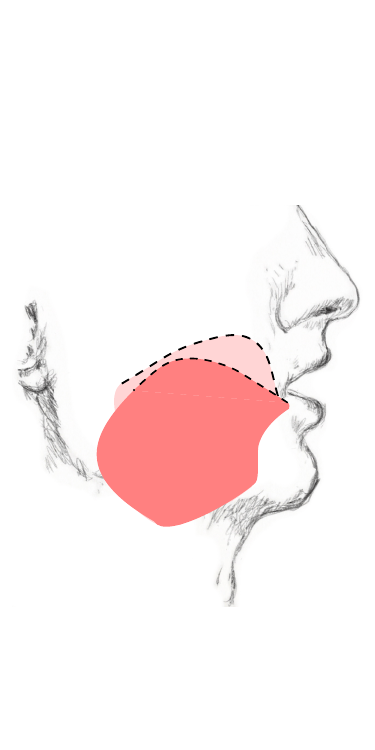
		\input{figures/tongues.tex}
	} 
	
	\caption{\label{fig:tongues}
		\textit{Left}: schematic illustrating the tongue muscle modulation when pronouncing ``$\mathtt{isi}$'' and ``$\mathtt{asa}$''. 
		Dashed lines correspond to the respective mean shapes in the right plot. With its multiple and multi-directional fibers, the tongue muscle almost fills the entire oral cavity and can flexibly adjust its shape. 
		In particular, not only tongue tip but also tongue root can move relatively freely. 
		\textit{Right}: elastic full Procrustes mean tongue shape estimates for different levels of aggregation. Tongue shapes are depicted in Bookstein coordinates, i.e.\,  with the tongue roots at $\beta(0) = 0$ and the tongue tips at $\beta(1) = 1$. Each panel shows the overall mean shape in the dataset \textit{(light gray, thick long-dashed line)}, the vowel-consonant mean shape \textit{(black, dashed line)}, and speaker-wise mean shapes \textit{(dark gray, solid lines)} for each combination. In each panel, original sample polygons \textit{(light red, thin lines, dots at sample points)} are added for the speaker with most intra-speaker variation (which is the same speaker except for ``$\mathtt{idi}$''). 
	}
\end{figure}

For quantitative assessment of the hierarchical variation structure, we consider the conditional variances $\sigma_{\mathcal{A}}^2 = \mathbb{E} (d^2_{\mathcal{E}}([B], [\mu_{\mathcal{A}}])\mid X \in \mathcal{A})$ 
with $X$ constrained on a subset $\mathcal{A} \subset \mathcal{X}$. Motivated by $\sigma_{\mathcal{A}}^2 = 1 - \lambda_{\mathcal{A}, 1}$ (Proposition \ref{lem:ifullProcMean} \ref{lem:ifullProcMean:eigen}) with $\lambda_{\mathcal{A}, 1}$ the largest eigenvalue of the respective conditional covariance operator, 
we estimate  $\hat{\sigma}^2_{\mathcal{A}} = 1-{\hat{\lambda}_{\mathcal{A},1}}{(\sum_{k=1}^m\hat{\lambda}_{\mathcal{A},k})^{-1}}$ with $\hat{\lambda}_{\mathcal{A},1}, \dots, \hat{\lambda}_{\mathcal{A},m}$ the positive eigenvalues of the covariance operator obtained in the final iteration of estimating $[\mu_{\mathcal{A}}]$. In a dense setting, where observations can be exactly normalized, the estimator $\check{\sigma}^2_{\mathcal{A}} = 1-\hat{\lambda}_{\mathcal{A},1}$ can be used directly, since when $\|Q\|=1$ almost surely also $\mathbb{E}(\|Q\|^2) = \sum_{k\geq1} \lambda_k = 1$. In a sparse setting, however, dividing by $\sum_{k=1}^m\hat{\lambda}_{\mathcal{A},k}$ in $\hat{\sigma}^2_{\mathcal{A}}$ ensures non-negative variance estimates.

In analogy to standard analysis of variance, we define the coefficient of determination for $\mathcal{A}_1$ in some decomposition $\mathcal{A}_1 \times \mathcal{A}_2 = \mathcal{X}$ as $\operatorname{R}^2_{\mathcal{A}_1} = 1 - {(|\mathcal{X}|\, \hat{\sigma}^2_{\mathcal{X}})^{-1}}{|\mathcal{A}_2| \sum_{a \in \mathcal{A}_1} \hat{\sigma}^2_{\{a\}\times\mathcal{A}_2}}$ reflecting the variance reduction achieved by conditioning on the features in $\mathcal{A}_1$.
Inspecting these measures underpins the visual impression from Figure \ref{fig:tongues}: although the tongue movement is induced by consonant pronunciation, the vowel context appears more dispositive for the tongue shape during articulation explaining more than half of the total variation ($\operatorname{R}^2_{\mathcal{V}} = 0.68$, $\operatorname{R}^2_{\mathcal{C}} = 0.11$), which increases only to $\operatorname{R}^2_{\mathcal{V}\times\mathcal{C}} = 0.73$ when also distinguishing consonants. Comparing the different vowel contexts, we observe nearly double variation for $\mathtt{a*a}$ than for $\mathtt{i*i}$ with ${\hat{\sigma}^2_{\{\mathtt{a*a}\}\times\mathcal{C}\times\mathcal{S}}}/{\hat{\sigma}^2_{\{\mathtt{i*i}\}\times\mathcal{C}\times\mathcal{S}}} = 1.95$, which might potentially relate to different pronunciations of ``$\mathtt{a}$" in German dialects.   
When considering single word articulation of a speaker ($\operatorname{R}^2_{\mathcal{V}\times\mathcal{C}\times\mathcal{S}} = 0.93$) about $7$ percent of the variation remain as residual variance, indicating that, while there is still non-negligible intra speaker variation, the inter speaker variance is considerably higher.

Recorded via ultrasound images, the shape of tongue surface contours modulo the respective invariances presents a natural object of analysis. Yet, if suitable reference landmarks allowed, the information on positioning, size, orientation and warping of the curve could also be separately investigated.  

\section{Discussion}\label{sec5}

While we find good performance of the proposed elastic full Procrustes mean estimator in realistic irregular/sparse curve data, future work should focus on theoretical assessment of estimation quality as well as inference. 
In particular, evaluation of the bias introduced by sub-optimal alignment of curves based on single discrete measurements is a topic of its own that would be of interest, as well as characterization of suitable sampling schemes where the bias is empirically negligible, 
which often  appears to be the case in practice.

In this paper, we focus on open rather than closed curves, since the presented covariance-based estimation approach is particularly natural in this case. Constraining curves $\beta$ to be closed, i.e.\ $\beta(0) = \beta(1)$, induces the non-linear constraint $\int_0^1 q(t) |q(t)| \,dt = 0$ on \SRV-level, which prevents direct application of Proposition \ref{lem:ifullProcMean} in the estimation. 
To still obtain a closed mean estimator for closed observations, the presented estimator could be closed along the lines of \citet[Chapter 10.6.2]{SrivastavaKlassen2016} either post-hoc or in each step of the fitting algorithm. 

As it can be analytically computed, inelastic full Procrustes analysis can also serve as a good starting point for estimating other types of shape means of plane curves. In addition, the estimated covariance structure supports estimation of inner products in sparse/irregular data scenarios, which are involved also in estimation of, e.g., other types of shape means. The presented results thus have relevance beyond the estimation of the (elastic) full Procrusted mean for plane shapes.


\backmatter


\section*{Acknowledgements}

We sincerely thank Marianne Pouplier and Philip Hoole for providing their carefully recorded phonetic tongue shape data and Paula Giesler and Sophia Schaffer for their help in understanding and visualizing its anatomical background. We gratefully acknowledge funding by grant GR 3793/3-1 from the German research foundation (DFG).\vspace*{-8pt}


%
 \bibliographystyle{biom} 
 \bibliography{literature}

%
%
%




\appendix

\newpage
\input{supplement.tex}
%
%
%

\label{lastpage}

\end{document}

%% file: figures/digits.tex
\begin{tikzpicture}[x=1pt,y=1pt]
\definecolor{fillColor}{RGB}{255,255,255}
\path[use as bounding box,fill=fillColor,fill opacity=0.00] (0,0) rectangle (155.81,180.67);
\begin{scope}
\path[clip] ( 14.88,  0.00) rectangle (140.93,180.67);
\definecolor{drawColor}{RGB}{255,255,255}
\definecolor{fillColor}{RGB}{255,255,255}

\path[draw=drawColor,line width= 0.6pt,line join=round,line cap=round,fill=fillColor] ( 14.88,  0.00) rectangle (140.93,180.68);
\end{scope}
\begin{scope}
\path[clip] ( 20.38, 93.49) rectangle ( 82.34,175.17);
\definecolor{drawColor}{RGB}{255,255,255}
\definecolor{fillColor}{RGB}{255,255,255}

\path[draw=drawColor,line width= 0.6pt,line join=round,line cap=round,fill=fillColor] ( 20.38, 93.49) rectangle ( 82.34,175.17);
\end{scope}
\begin{scope}
\path[clip] ( 35.70, 96.24) rectangle ( 82.24,157.90);
\definecolor{fillColor}{RGB}{255,255,255}

\path[fill=fillColor] ( 35.70, 96.24) rectangle ( 82.24,157.90);
\definecolor{drawColor}{gray}{0.92}

\path[draw=drawColor,line width= 0.3pt,line join=round] ( 35.70,104.19) --
	( 82.24,104.19);

\path[draw=drawColor,line width= 0.3pt,line join=round] ( 35.70,114.02) --
	( 82.24,114.02);

\path[draw=drawColor,line width= 0.3pt,line join=round] ( 35.70,123.85) --
	( 82.24,123.85);

\path[draw=drawColor,line width= 0.3pt,line join=round] ( 35.70,133.67) --
	( 82.24,133.67);

\path[draw=drawColor,line width= 0.3pt,line join=round] ( 35.70,143.50) --
	( 82.24,143.50);

\path[draw=drawColor,line width= 0.3pt,line join=round] ( 35.70,153.32) --
	( 82.24,153.32);

\path[draw=drawColor,line width= 0.3pt,line join=round] ( 38.17, 96.24) --
	( 38.17,157.90);

\path[draw=drawColor,line width= 0.3pt,line join=round] ( 48.00, 96.24) --
	( 48.00,157.90);

\path[draw=drawColor,line width= 0.3pt,line join=round] ( 57.82, 96.24) --
	( 57.82,157.90);

\path[draw=drawColor,line width= 0.3pt,line join=round] ( 67.65, 96.24) --
	( 67.65,157.90);

\path[draw=drawColor,line width= 0.3pt,line join=round] ( 77.48, 96.24) --
	( 77.48,157.90);

\path[draw=drawColor,line width= 0.6pt,line join=round] ( 35.70, 99.28) --
	( 82.24, 99.28);

\path[draw=drawColor,line width= 0.6pt,line join=round] ( 35.70,109.11) --
	( 82.24,109.11);

\path[draw=drawColor,line width= 0.6pt,line join=round] ( 35.70,118.93) --
	( 82.24,118.93);

\path[draw=drawColor,line width= 0.6pt,line join=round] ( 35.70,128.76) --
	( 82.24,128.76);

\path[draw=drawColor,line width= 0.6pt,line join=round] ( 35.70,138.59) --
	( 82.24,138.59);

\path[draw=drawColor,line width= 0.6pt,line join=round] ( 35.70,148.41) --
	( 82.24,148.41);

\path[draw=drawColor,line width= 0.6pt,line join=round] ( 43.08, 96.24) --
	( 43.08,157.90);

\path[draw=drawColor,line width= 0.6pt,line join=round] ( 52.91, 96.24) --
	( 52.91,157.90);

\path[draw=drawColor,line width= 0.6pt,line join=round] ( 62.74, 96.24) --
	( 62.74,157.90);

\path[draw=drawColor,line width= 0.6pt,line join=round] ( 72.56, 96.24) --
	( 72.56,157.90);
\definecolor{drawColor}{gray}{0.70}

\path[draw=drawColor,line width= 0.6pt,line join=round] ( 42.89,105.04) --
	( 58.06,101.44) --
	( 69.11,104.10) --
	( 77.41,111.07) --
	( 80.13,120.43) --
	( 73.99,130.27) --
	( 62.44,134.08) --
	( 54.83,134.29) --
	( 60.57,136.38) --
	( 68.10,142.87) --
	( 67.02,148.56) --
	( 58.32,152.52) --
	( 42.71,152.82);
\definecolor{drawColor}{RGB}{0,0,0}

\path[draw=drawColor,line width= 0.6pt,line join=round] ( 37.81,102.63) --
	( 39.44,101.96) --
	( 41.07,101.36) --
	( 42.71,100.83) --
	( 44.35,100.37) --
	( 46.00, 99.97) --
	( 47.67, 99.65) --
	( 49.35, 99.38) --
	( 51.05, 99.18) --
	( 52.76, 99.06) --
	( 54.45, 99.04) --
	( 56.11, 99.12) --
	( 57.76, 99.30) --
	( 59.40, 99.58) --
	( 61.03, 99.97) --
	( 62.67,100.46) --
	( 64.32,101.07) --
	( 65.97,101.79) --
	( 67.52,102.59) --
	( 68.95,103.46) --
	( 70.27,104.41) --
	( 71.48,105.43) --
	( 72.60,106.53) --
	( 73.63,107.73) --
	( 74.57,109.02) --
	( 75.45,110.41) --
	( 76.17,111.83) --
	( 76.69,113.19) --
	( 77.03,114.52) --
	( 77.19,115.83) --
	( 77.20,117.13) --
	( 77.04,118.46) --
	( 76.70,119.83) --
	( 76.18,121.25) --
	( 75.47,122.71) --
	( 74.68,124.07) --
	( 73.81,125.34) --
	( 72.84,126.51) --
	( 71.79,127.61) --
	( 70.63,128.64) --
	( 69.36,129.59) --
	( 67.98,130.48) --
	( 66.48,131.31) --
	( 65.11,132.03) --
	( 63.92,132.65) --
	( 62.90,133.18) --
	( 62.05,133.61) --
	( 61.34,133.97) --
	( 60.77,134.26) --
	( 60.31,134.48) --
	( 59.96,134.65) --
	( 59.72,134.78) --
	( 59.56,134.89) --
	( 59.48,134.97) --
	( 59.45,135.04) --
	( 59.44,135.10) --
	( 59.47,135.17) --
	( 59.55,135.26) --
	( 59.69,135.36) --
	( 59.91,135.49) --
	( 60.20,135.64) --
	( 60.55,135.82) --
	( 60.98,136.04) --
	( 61.49,136.28) --
	( 62.10,136.55) --
	( 62.80,136.87) --
	( 63.61,137.22) --
	( 64.52,137.61) --
	( 65.38,138.04) --
	( 66.15,138.53) --
	( 66.86,139.06) --
	( 67.49,139.64) --
	( 68.06,140.27) --
	( 68.57,140.97) --
	( 69.04,141.73) --
	( 69.44,142.58) --
	( 69.75,143.43) --
	( 69.92,144.21) --
	( 69.97,144.96) --
	( 69.91,145.68) --
	( 69.73,146.39) --
	( 69.42,147.11) --
	( 68.97,147.86) --
	( 68.36,148.64) --
	( 67.58,149.44) --
	( 66.70,150.19) --
	( 65.72,150.88) --
	( 64.63,151.53) --
	( 63.40,152.12) --
	( 62.04,152.66) --
	( 60.52,153.14) --
	( 58.82,153.57) --
	( 56.95,153.92) --
	( 54.98,154.23) --
	( 52.94,154.50) --
	( 50.83,154.71) --
	( 48.64,154.88) --
	( 46.37,155.01) --
	( 44.02,155.08) --
	( 41.59,155.10) --
	( 39.07,155.06);
\definecolor{drawColor}{gray}{0.20}

\path[draw=drawColor,line width= 0.6pt,line join=round,line cap=round] ( 35.70, 96.24) rectangle ( 82.24,157.90);
\end{scope}
\begin{scope}
\path[clip] (  0.00,  0.00) rectangle (155.81,180.67);
\definecolor{drawColor}{RGB}{0,0,0}

\node[text=drawColor,rotate= 90.00,anchor=base,inner sep=0pt, outer sep=0pt, scale=  1.10] at ( 28.06,127.07) {\small$\Im(\beta)$};
\end{scope}
\begin{scope}
\path[clip] (  0.00,  0.00) rectangle (155.81,180.67);
\definecolor{drawColor}{gray}{0.25}

\node[text=drawColor,anchor=base,inner sep=0pt, outer sep=0pt, scale=  0.88] at ( 58.97,169.11) {elastic full};

\node[text=drawColor,anchor=base,inner sep=0pt, outer sep=0pt, scale=  0.88] at ( 58.97,159.61) {Procrustes};
\end{scope}
\begin{scope}
\path[clip] ( 82.34, 93.49) rectangle (135.43,175.17);
\definecolor{drawColor}{RGB}{255,255,255}
\definecolor{fillColor}{RGB}{255,255,255}

\path[draw=drawColor,line width= 0.6pt,line join=round,line cap=round,fill=fillColor] ( 82.34, 93.49) rectangle (135.43,175.17);
\end{scope}
\begin{scope}
\path[clip] ( 85.19, 96.24) rectangle (135.33,157.90);
\definecolor{fillColor}{RGB}{255,255,255}

\path[fill=fillColor] ( 85.19, 96.24) rectangle (135.33,157.90);
\definecolor{drawColor}{gray}{0.92}

\path[draw=drawColor,line width= 0.3pt,line join=round] ( 85.19,106.74) --
	(135.33,106.74);

\path[draw=drawColor,line width= 0.3pt,line join=round] ( 85.19,121.34) --
	(135.33,121.34);

\path[draw=drawColor,line width= 0.3pt,line join=round] ( 85.19,135.95) --
	(135.33,135.95);

\path[draw=drawColor,line width= 0.3pt,line join=round] ( 85.19,150.55) --
	(135.33,150.55);

\path[draw=drawColor,line width= 0.3pt,line join=round] ( 97.16, 96.24) --
	( 97.16,157.90);

\path[draw=drawColor,line width= 0.3pt,line join=round] (111.76, 96.24) --
	(111.76,157.90);

\path[draw=drawColor,line width= 0.3pt,line join=round] (126.36, 96.24) --
	(126.36,157.90);

\path[draw=drawColor,line width= 0.6pt,line join=round] ( 85.19, 99.44) --
	(135.33, 99.44);

\path[draw=drawColor,line width= 0.6pt,line join=round] ( 85.19,114.04) --
	(135.33,114.04);

\path[draw=drawColor,line width= 0.6pt,line join=round] ( 85.19,128.64) --
	(135.33,128.64);

\path[draw=drawColor,line width= 0.6pt,line join=round] ( 85.19,143.25) --
	(135.33,143.25);

\path[draw=drawColor,line width= 0.6pt,line join=round] ( 85.19,157.85) --
	(135.33,157.85);

\path[draw=drawColor,line width= 0.6pt,line join=round] ( 89.85, 96.24) --
	( 89.85,157.90);

\path[draw=drawColor,line width= 0.6pt,line join=round] (104.46, 96.24) --
	(104.46,157.90);

\path[draw=drawColor,line width= 0.6pt,line join=round] (119.06, 96.24) --
	(119.06,157.90);

\path[draw=drawColor,line width= 0.6pt,line join=round] (133.66, 96.24) --
	(133.66,157.90);
\definecolor{drawColor}{gray}{0.70}

\path[draw=drawColor,line width= 0.6pt,line join=round] ( 87.47,103.80) --
	(105.80, 99.04) --
	(120.94,101.73) --
	(131.15,110.37) --
	(133.05,121.33) --
	(128.11,130.28) --
	(124.07,134.92) --
	(123.45,136.57) --
	(125.59,138.34) --
	(127.91,141.23) --
	(125.79,147.15) --
	(116.09,152.97) --
	( 98.36,154.64);
\definecolor{drawColor}{RGB}{0,0,0}

\path[draw=drawColor,line width= 0.6pt,line join=round] ( 91.40,103.09) --
	( 93.34,102.29) --
	( 95.26,101.59) --
	( 97.15,100.97) --
	( 99.03,100.45) --
	(100.89,100.01) --
	(102.75, 99.65) --
	(104.59, 99.38) --
	(106.43, 99.18) --
	(108.26, 99.08) --
	(110.04, 99.08) --
	(111.78, 99.18) --
	(113.47, 99.39) --
	(115.14, 99.71) --
	(116.79,100.12) --
	(118.42,100.65) --
	(120.03,101.29) --
	(121.64,102.03) --
	(123.14,102.86) --
	(124.53,103.76) --
	(125.80,104.73) --
	(126.98,105.78) --
	(128.06,106.92) --
	(129.06,108.14) --
	(129.98,109.46) --
	(130.83,110.88) --
	(131.54,112.33) --
	(132.07,113.73) --
	(132.44,115.08) --
	(132.66,116.40) --
	(132.73,117.70) --
	(132.67,119.00) --
	(132.46,120.32) --
	(132.10,121.66) --
	(131.61,123.02) --
	(131.04,124.28) --
	(130.40,125.44) --
	(129.70,126.53) --
	(128.92,127.53) --
	(128.08,128.47) --
	(127.15,129.35) --
	(126.13,130.17) --
	(125.04,130.93) --
	(124.00,131.62) --
	(123.06,132.24) --
	(122.21,132.78) --
	(121.46,133.25) --
	(120.79,133.67) --
	(120.20,134.02) --
	(119.68,134.32) --
	(119.23,134.57) --
	(118.87,134.79) --
	(118.60,134.98) --
	(118.41,135.14) --
	(118.28,135.28) --
	(118.20,135.40) --
	(118.16,135.50) --
	(118.15,135.60) --
	(118.17,135.69) --
	(118.21,135.79) --
	(118.29,135.91) --
	(118.41,136.04) --
	(118.59,136.20) --
	(118.82,136.38) --
	(119.12,136.59) --
	(119.50,136.83) --
	(119.98,137.10) --
	(120.54,137.40) --
	(121.09,137.74) --
	(121.60,138.11) --
	(122.08,138.53) --
	(122.51,138.99) --
	(122.92,139.51) --
	(123.30,140.07) --
	(123.65,140.70) --
	(123.97,141.40) --
	(124.22,142.12) --
	(124.37,142.83) --
	(124.41,143.54) --
	(124.36,144.26) --
	(124.21,145.01) --
	(123.95,145.80) --
	(123.56,146.64) --
	(123.03,147.54) --
	(122.36,148.48) --
	(121.60,149.35) --
	(120.76,150.16) --
	(119.81,150.92) --
	(118.75,151.62) --
	(117.56,152.26) --
	(116.24,152.85) --
	(114.76,153.39) --
	(113.13,153.86) --
	(111.40,154.26) --
	(109.61,154.59) --
	(107.74,154.84) --
	(105.79,155.01) --
	(103.76,155.10) --
	(101.63,155.10) --
	( 99.41,155.01) --
	( 97.08,154.83);
\definecolor{drawColor}{gray}{0.20}

\path[draw=drawColor,line width= 0.6pt,line join=round,line cap=round] ( 85.19, 96.24) rectangle (135.33,157.90);
\end{scope}
\begin{scope}
\path[clip] (  0.00,  0.00) rectangle (155.81,180.67);
\definecolor{drawColor}{gray}{0.25}

\node[text=drawColor,anchor=base,inner sep=0pt, outer sep=0pt, scale=  0.88] at (110.26,169.11) {inelastic full};

\node[text=drawColor,anchor=base,inner sep=0pt, outer sep=0pt, scale=  0.88] at (110.26,159.61) {Procrustes};
\end{scope}
\begin{scope}
\path[clip] ( 20.38,  5.50) rectangle ( 82.34, 93.49);
\definecolor{drawColor}{RGB}{255,255,255}
\definecolor{fillColor}{RGB}{255,255,255}

\path[draw=drawColor,line width= 0.6pt,line join=round,line cap=round,fill=fillColor] ( 20.38,  5.50) rectangle ( 82.34, 93.49);
\end{scope}
\begin{scope}
\path[clip] ( 35.70, 20.71) rectangle ( 82.24, 85.71);
\definecolor{fillColor}{RGB}{255,255,255}

\path[fill=fillColor] ( 35.70, 20.71) rectangle ( 82.24, 85.71);
\definecolor{drawColor}{gray}{0.92}

\path[draw=drawColor,line width= 0.3pt,line join=round] ( 35.70, 26.41) --
	( 82.24, 26.41);

\path[draw=drawColor,line width= 0.3pt,line join=round] ( 35.70, 37.30) --
	( 82.24, 37.30);

\path[draw=drawColor,line width= 0.3pt,line join=round] ( 35.70, 48.18) --
	( 82.24, 48.18);

\path[draw=drawColor,line width= 0.3pt,line join=round] ( 35.70, 59.06) --
	( 82.24, 59.06);

\path[draw=drawColor,line width= 0.3pt,line join=round] ( 35.70, 69.95) --
	( 82.24, 69.95);

\path[draw=drawColor,line width= 0.3pt,line join=round] ( 35.70, 80.83) --
	( 82.24, 80.83);

\path[draw=drawColor,line width= 0.3pt,line join=round] ( 45.09, 20.71) --
	( 45.09, 85.71);

\path[draw=drawColor,line width= 0.3pt,line join=round] ( 55.98, 20.71) --
	( 55.98, 85.71);

\path[draw=drawColor,line width= 0.3pt,line join=round] ( 66.86, 20.71) --
	( 66.86, 85.71);

\path[draw=drawColor,line width= 0.3pt,line join=round] ( 77.75, 20.71) --
	( 77.75, 85.71);

\path[draw=drawColor,line width= 0.6pt,line join=round] ( 35.70, 20.97) --
	( 82.24, 20.97);

\path[draw=drawColor,line width= 0.6pt,line join=round] ( 35.70, 31.85) --
	( 82.24, 31.85);

\path[draw=drawColor,line width= 0.6pt,line join=round] ( 35.70, 42.74) --
	( 82.24, 42.74);

\path[draw=drawColor,line width= 0.6pt,line join=round] ( 35.70, 53.62) --
	( 82.24, 53.62);

\path[draw=drawColor,line width= 0.6pt,line join=round] ( 35.70, 64.50) --
	( 82.24, 64.50);

\path[draw=drawColor,line width= 0.6pt,line join=round] ( 35.70, 75.39) --
	( 82.24, 75.39);

\path[draw=drawColor,line width= 0.6pt,line join=round] ( 39.65, 20.71) --
	( 39.65, 85.71);

\path[draw=drawColor,line width= 0.6pt,line join=round] ( 50.54, 20.71) --
	( 50.54, 85.71);

\path[draw=drawColor,line width= 0.6pt,line join=round] ( 61.42, 20.71) --
	( 61.42, 85.71);

\path[draw=drawColor,line width= 0.6pt,line join=round] ( 72.30, 20.71) --
	( 72.30, 85.71);
\definecolor{drawColor}{gray}{0.70}

\path[draw=drawColor,line width= 0.6pt,line join=round] ( 45.41, 27.42) --
	( 56.32, 23.67) --
	( 67.45, 26.06) --
	( 76.23, 32.15) --
	( 80.13, 41.42) --
	( 70.12, 56.58) --
	( 60.41, 58.91) --
	( 52.15, 58.54) --
	( 61.61, 62.87) --
	( 67.14, 71.29) --
	( 64.88, 77.44) --
	( 53.85, 80.25) --
	( 42.78, 80.47);
\definecolor{drawColor}{RGB}{0,0,0}

\path[draw=drawColor,line width= 0.6pt,line join=round] ( 37.81, 38.12) --
	( 39.63, 35.47) --
	( 41.39, 33.28) --
	( 43.08, 31.50) --
	( 44.71, 30.09) --
	( 46.29, 28.98) --
	( 47.84, 28.15) --
	( 49.39, 27.57) --
	( 50.95, 27.22) --
	( 52.54, 27.04) --
	( 54.12, 26.95) --
	( 55.70, 26.92) --
	( 57.28, 26.96) --
	( 58.86, 27.07) --
	( 60.44, 27.25) --
	( 62.03, 27.50) --
	( 63.63, 27.82) --
	( 65.23, 28.21) --
	( 66.73, 28.69) --
	( 68.11, 29.26) --
	( 69.38, 29.93) --
	( 70.55, 30.68) --
	( 71.64, 31.53) --
	( 72.65, 32.49) --
	( 73.59, 33.55) --
	( 74.47, 34.75) --
	( 75.21, 35.98) --
	( 75.75, 37.14) --
	( 76.12, 38.27) --
	( 76.32, 39.36) --
	( 76.37, 40.45) --
	( 76.27, 41.55) --
	( 76.02, 42.68) --
	( 75.59, 43.87) --
	( 75.01, 45.08) --
	( 74.40, 46.19) --
	( 73.75, 47.21) --
	( 73.07, 48.14) --
	( 72.36, 48.99) --
	( 71.62, 49.75) --
	( 70.85, 50.45) --
	( 70.03, 51.08) --
	( 69.18, 51.65) --
	( 68.32, 52.12) --
	( 67.47, 52.49) --
	( 66.62, 52.77) --
	( 65.76, 52.97) --
	( 64.88, 53.08) --
	( 63.98, 53.11) --
	( 63.04, 53.05) --
	( 62.07, 52.90) --
	( 61.20, 52.81) --
	( 60.54, 52.89) --
	( 60.03, 53.10) --
	( 59.61, 53.46) --
	( 59.26, 54.01) --
	( 58.98, 54.82) --
	( 58.77, 55.97) --
	( 58.69, 57.55) --
	( 58.76, 59.40) --
	( 58.91, 60.77) --
	( 59.11, 61.71) --
	( 59.34, 62.30) --
	( 59.57, 62.64) --
	( 59.81, 62.81) --
	( 60.10, 62.87) --
	( 60.47, 62.80) --
	( 60.97, 62.57) --
	( 61.44, 62.43) --
	( 61.85, 62.44) --
	( 62.23, 62.57) --
	( 62.64, 62.87) --
	( 63.09, 63.38) --
	( 63.58, 64.18) --
	( 64.11, 65.34) --
	( 64.69, 66.94) --
	( 65.18, 68.68) --
	( 65.47, 70.19) --
	( 65.60, 71.51) --
	( 65.57, 72.67) --
	( 65.41, 73.68) --
	( 65.13, 74.58) --
	( 64.74, 75.39) --
	( 64.21, 76.14) --
	( 63.56, 76.84) --
	( 62.77, 77.53) --
	( 61.84, 78.22) --
	( 60.76, 78.89) --
	( 59.49, 79.56) --
	( 58.02, 80.21) --
	( 56.34, 80.85) --
	( 54.42, 81.47) --
	( 52.28, 82.06) --
	( 50.22, 82.47) --
	( 48.35, 82.70) --
	( 46.63, 82.76) --
	( 45.06, 82.66) --
	( 43.60, 82.41) --
	( 42.23, 82.02) --
	( 40.95, 81.49) --
	( 39.72, 80.82);
\definecolor{drawColor}{gray}{0.20}

\path[draw=drawColor,line width= 0.6pt,line join=round,line cap=round] ( 35.70, 20.71) rectangle ( 82.24, 85.71);
\end{scope}
\begin{scope}
\path[clip] (  0.00,  0.00) rectangle (155.81,180.67);
\definecolor{drawColor}{RGB}{0,0,0}

\node[text=drawColor,rotate= 90.00,anchor=base,inner sep=0pt, outer sep=0pt, scale=  1.10] at ( 28.06, 53.21) {\small$\Im(\beta)$};
\end{scope}
\begin{scope}
\path[clip] (  0.00,  0.00) rectangle (155.81,180.67);
\definecolor{drawColor}{RGB}{0,0,0}

\node[text=drawColor,anchor=base,inner sep=0pt, outer sep=0pt, scale=  1.10] at ( 58.97,  7.64) {\small$\Re(\beta)$};
\end{scope}
\begin{scope}
\path[clip] (  0.00,  0.00) rectangle (155.81,180.67);
\definecolor{drawColor}{gray}{0.25}

\node[text=drawColor,anchor=base,inner sep=0pt, outer sep=0pt, scale=  0.88] at ( 58.97, 87.42) {elastic};
\end{scope}
\begin{scope}
\path[clip] ( 82.34,  5.50) rectangle (135.43, 93.49);
\definecolor{drawColor}{RGB}{255,255,255}
\definecolor{fillColor}{RGB}{255,255,255}

\path[draw=drawColor,line width= 0.6pt,line join=round,line cap=round,fill=fillColor] ( 82.34,  5.50) rectangle (135.43, 93.49);
\end{scope}
\begin{scope}
\path[clip] ( 85.19, 53.25) rectangle (110.24, 85.71);
\definecolor{fillColor}{RGB}{255,255,255}

\path[fill=fillColor] ( 85.19, 53.25) rectangle (110.24, 85.71);
\definecolor{drawColor}{gray}{0.92}

\path[draw=drawColor,line width= 0.3pt,line join=round] ( 85.19, 56.78) --
	(110.24, 56.78);

\path[draw=drawColor,line width= 0.3pt,line join=round] ( 85.19, 65.14) --
	(110.24, 65.14);

\path[draw=drawColor,line width= 0.3pt,line join=round] ( 85.19, 73.50) --
	(110.24, 73.50);

\path[draw=drawColor,line width= 0.3pt,line join=round] ( 85.19, 81.86) --
	(110.24, 81.86);

\path[draw=drawColor,line width= 0.3pt,line join=round] ( 86.56, 53.25) --
	( 86.56, 85.71);

\path[draw=drawColor,line width= 0.3pt,line join=round] ( 94.15, 53.25) --
	( 94.15, 85.71);

\path[draw=drawColor,line width= 0.3pt,line join=round] (101.74, 53.25) --
	(101.74, 85.71);

\path[draw=drawColor,line width= 0.3pt,line join=round] (109.33, 53.25) --
	(109.33, 85.71);

\path[draw=drawColor,line width= 0.6pt,line join=round] ( 85.19, 60.96) --
	(110.24, 60.96);

\path[draw=drawColor,line width= 0.6pt,line join=round] ( 85.19, 69.32) --
	(110.24, 69.32);

\path[draw=drawColor,line width= 0.6pt,line join=round] ( 85.19, 77.68) --
	(110.24, 77.68);

\path[draw=drawColor,line width= 0.6pt,line join=round] ( 90.36, 53.25) --
	( 90.36, 85.71);

\path[draw=drawColor,line width= 0.6pt,line join=round] ( 97.95, 53.25) --
	( 97.95, 85.71);

\path[draw=drawColor,line width= 0.6pt,line join=round] (105.54, 53.25) --
	(105.54, 85.71);
\definecolor{drawColor}{gray}{0.70}

\path[draw=drawColor,line width= 0.6pt,line join=round] ( 87.09, 65.85) --
	( 89.37, 62.50) --
	( 93.16, 58.32) --
	( 99.99, 55.82) --
	(106.06, 57.49) --
	(107.58, 60.83) --
	(109.10, 65.85) --
	(106.82, 72.53) --
	(103.03, 75.88) --
	( 96.20, 76.71) --
	( 96.20, 81.73) --
	( 92.40, 83.40) --
	( 86.33, 84.24);
\definecolor{drawColor}{RGB}{102,102,102}
\definecolor{fillColor}{RGB}{102,102,102}

\path[draw=drawColor,draw opacity=0.70,line width= 0.4pt,line join=round,line cap=round,fill=fillColor,fill opacity=0.70] ( 87.09, 65.85) circle (  1.43);

\path[draw=drawColor,draw opacity=0.70,line width= 0.4pt,line join=round,line cap=round,fill=fillColor,fill opacity=0.70] ( 89.37, 62.50) circle (  1.43);

\path[draw=drawColor,draw opacity=0.70,line width= 0.4pt,line join=round,line cap=round,fill=fillColor,fill opacity=0.70] ( 93.16, 58.32) circle (  1.43);

\path[draw=drawColor,draw opacity=0.70,line width= 0.4pt,line join=round,line cap=round,fill=fillColor,fill opacity=0.70] ( 99.99, 55.82) circle (  1.43);

\path[draw=drawColor,draw opacity=0.70,line width= 0.4pt,line join=round,line cap=round,fill=fillColor,fill opacity=0.70] (106.06, 57.49) circle (  1.43);

\path[draw=drawColor,draw opacity=0.70,line width= 0.4pt,line join=round,line cap=round,fill=fillColor,fill opacity=0.70] (107.58, 60.83) circle (  1.43);

\path[draw=drawColor,draw opacity=0.70,line width= 0.4pt,line join=round,line cap=round,fill=fillColor,fill opacity=0.70] (109.10, 65.85) circle (  1.43);

\path[draw=drawColor,draw opacity=0.70,line width= 0.4pt,line join=round,line cap=round,fill=fillColor,fill opacity=0.70] (106.82, 72.53) circle (  1.43);

\path[draw=drawColor,draw opacity=0.70,line width= 0.4pt,line join=round,line cap=round,fill=fillColor,fill opacity=0.70] (103.03, 75.88) circle (  1.43);

\path[draw=drawColor,draw opacity=0.70,line width= 0.4pt,line join=round,line cap=round,fill=fillColor,fill opacity=0.70] ( 96.20, 76.71) circle (  1.43);

\path[draw=drawColor,draw opacity=0.70,line width= 0.4pt,line join=round,line cap=round,fill=fillColor,fill opacity=0.70] ( 96.20, 81.73) circle (  1.43);

\path[draw=drawColor,draw opacity=0.70,line width= 0.4pt,line join=round,line cap=round,fill=fillColor,fill opacity=0.70] ( 92.40, 83.40) circle (  1.43);

\path[draw=drawColor,draw opacity=0.70,line width= 0.4pt,line join=round,line cap=round,fill=fillColor,fill opacity=0.70] ( 86.33, 84.24) circle (  1.43);
\definecolor{drawColor}{gray}{0.20}

\path[draw=drawColor,line width= 0.6pt,line join=round,line cap=round] ( 85.19, 53.25) rectangle (110.24, 85.71);
\end{scope}
\begin{scope}
\path[clip] ( 85.19, 20.71) rectangle (110.24, 53.18);
\definecolor{fillColor}{RGB}{255,255,255}

\path[fill=fillColor] ( 85.19, 20.71) rectangle (110.24, 53.18);
\definecolor{drawColor}{gray}{0.92}

\path[draw=drawColor,line width= 0.3pt,line join=round] ( 85.19, 24.25) --
	(110.24, 24.25);

\path[draw=drawColor,line width= 0.3pt,line join=round] ( 85.19, 32.61) --
	(110.24, 32.61);

\path[draw=drawColor,line width= 0.3pt,line join=round] ( 85.19, 40.97) --
	(110.24, 40.97);

\path[draw=drawColor,line width= 0.3pt,line join=round] ( 85.19, 49.33) --
	(110.24, 49.33);

\path[draw=drawColor,line width= 0.3pt,line join=round] ( 86.56, 20.71) --
	( 86.56, 53.18);

\path[draw=drawColor,line width= 0.3pt,line join=round] ( 94.15, 20.71) --
	( 94.15, 53.18);

\path[draw=drawColor,line width= 0.3pt,line join=round] (101.74, 20.71) --
	(101.74, 53.18);

\path[draw=drawColor,line width= 0.3pt,line join=round] (109.33, 20.71) --
	(109.33, 53.18);

\path[draw=drawColor,line width= 0.6pt,line join=round] ( 85.19, 28.43) --
	(110.24, 28.43);

\path[draw=drawColor,line width= 0.6pt,line join=round] ( 85.19, 36.79) --
	(110.24, 36.79);

\path[draw=drawColor,line width= 0.6pt,line join=round] ( 85.19, 45.15) --
	(110.24, 45.15);

\path[draw=drawColor,line width= 0.6pt,line join=round] ( 90.36, 20.71) --
	( 90.36, 53.18);

\path[draw=drawColor,line width= 0.6pt,line join=round] ( 97.95, 20.71) --
	( 97.95, 53.18);

\path[draw=drawColor,line width= 0.6pt,line join=round] (105.54, 20.71) --
	(105.54, 53.18);
\definecolor{drawColor}{gray}{0.70}

\path[draw=drawColor,line width= 0.6pt,line join=round] ( 93.69, 24.89) --
	( 96.72, 26.56) --
	(100.52, 28.23) --
	(101.28, 31.58) --
	( 99.76, 34.92) --
	( 97.48, 35.76) --
	( 95.20, 34.09) --
	( 99.00, 37.43) --
	(100.52, 41.61) --
	(102.04, 44.95) --
	(100.52, 47.46) --
	( 94.45, 46.63) --
	( 92.17, 44.12);
\definecolor{drawColor}{RGB}{102,102,102}
\definecolor{fillColor}{RGB}{102,102,102}

\path[draw=drawColor,draw opacity=0.70,line width= 0.4pt,line join=round,line cap=round,fill=fillColor,fill opacity=0.70] ( 93.69, 24.89) circle (  1.43);

\path[draw=drawColor,draw opacity=0.70,line width= 0.4pt,line join=round,line cap=round,fill=fillColor,fill opacity=0.70] ( 96.72, 26.56) circle (  1.43);

\path[draw=drawColor,draw opacity=0.70,line width= 0.4pt,line join=round,line cap=round,fill=fillColor,fill opacity=0.70] (100.52, 28.23) circle (  1.43);

\path[draw=drawColor,draw opacity=0.70,line width= 0.4pt,line join=round,line cap=round,fill=fillColor,fill opacity=0.70] (101.28, 31.58) circle (  1.43);

\path[draw=drawColor,draw opacity=0.70,line width= 0.4pt,line join=round,line cap=round,fill=fillColor,fill opacity=0.70] ( 99.76, 34.92) circle (  1.43);

\path[draw=drawColor,draw opacity=0.70,line width= 0.4pt,line join=round,line cap=round,fill=fillColor,fill opacity=0.70] ( 97.48, 35.76) circle (  1.43);

\path[draw=drawColor,draw opacity=0.70,line width= 0.4pt,line join=round,line cap=round,fill=fillColor,fill opacity=0.70] ( 95.20, 34.09) circle (  1.43);

\path[draw=drawColor,draw opacity=0.70,line width= 0.4pt,line join=round,line cap=round,fill=fillColor,fill opacity=0.70] ( 99.00, 37.43) circle (  1.43);

\path[draw=drawColor,draw opacity=0.70,line width= 0.4pt,line join=round,line cap=round,fill=fillColor,fill opacity=0.70] (100.52, 41.61) circle (  1.43);

\path[draw=drawColor,draw opacity=0.70,line width= 0.4pt,line join=round,line cap=round,fill=fillColor,fill opacity=0.70] (102.04, 44.95) circle (  1.43);

\path[draw=drawColor,draw opacity=0.70,line width= 0.4pt,line join=round,line cap=round,fill=fillColor,fill opacity=0.70] (100.52, 47.46) circle (  1.43);

\path[draw=drawColor,draw opacity=0.70,line width= 0.4pt,line join=round,line cap=round,fill=fillColor,fill opacity=0.70] ( 94.45, 46.63) circle (  1.43);

\path[draw=drawColor,draw opacity=0.70,line width= 0.4pt,line join=round,line cap=round,fill=fillColor,fill opacity=0.70] ( 92.17, 44.12) circle (  1.43);
\definecolor{drawColor}{gray}{0.20}

\path[draw=drawColor,line width= 0.6pt,line join=round,line cap=round] ( 85.19, 20.71) rectangle (110.24, 53.18);
\end{scope}
\begin{scope}
\path[clip] (110.29, 53.25) rectangle (135.33, 85.71);
\definecolor{fillColor}{RGB}{255,255,255}

\path[fill=fillColor] (110.29, 53.25) rectangle (135.33, 85.71);
\definecolor{drawColor}{gray}{0.92}

\path[draw=drawColor,line width= 0.3pt,line join=round] (110.29, 56.78) --
	(135.33, 56.78);

\path[draw=drawColor,line width= 0.3pt,line join=round] (110.29, 65.14) --
	(135.33, 65.14);

\path[draw=drawColor,line width= 0.3pt,line join=round] (110.29, 73.50) --
	(135.33, 73.50);

\path[draw=drawColor,line width= 0.3pt,line join=round] (110.29, 81.86) --
	(135.33, 81.86);

\path[draw=drawColor,line width= 0.3pt,line join=round] (111.66, 53.25) --
	(111.66, 85.71);

\path[draw=drawColor,line width= 0.3pt,line join=round] (119.25, 53.25) --
	(119.25, 85.71);

\path[draw=drawColor,line width= 0.3pt,line join=round] (126.84, 53.25) --
	(126.84, 85.71);

\path[draw=drawColor,line width= 0.3pt,line join=round] (134.43, 53.25) --
	(134.43, 85.71);

\path[draw=drawColor,line width= 0.6pt,line join=round] (110.29, 60.96) --
	(135.33, 60.96);

\path[draw=drawColor,line width= 0.6pt,line join=round] (110.29, 69.32) --
	(135.33, 69.32);

\path[draw=drawColor,line width= 0.6pt,line join=round] (110.29, 77.68) --
	(135.33, 77.68);

\path[draw=drawColor,line width= 0.6pt,line join=round] (115.45, 53.25) --
	(115.45, 85.71);

\path[draw=drawColor,line width= 0.6pt,line join=round] (123.04, 53.25) --
	(123.04, 85.71);

\path[draw=drawColor,line width= 0.6pt,line join=round] (130.63, 53.25) --
	(130.63, 85.71);
\definecolor{drawColor}{gray}{0.70}

\path[draw=drawColor,line width= 0.6pt,line join=round] (117.79, 58.39) --
	(121.58, 58.39) --
	(125.38, 62.57) --
	(126.14, 65.91) --
	(123.86, 70.09) --
	(119.31, 69.26) --
	(116.27, 68.42) --
	(123.86, 70.09) --
	(126.14, 71.76) --
	(126.90, 74.27) --
	(126.90, 76.78) --
	(123.86, 77.61) --
	(121.58, 77.61);
\definecolor{drawColor}{RGB}{102,102,102}
\definecolor{fillColor}{RGB}{102,102,102}

\path[draw=drawColor,draw opacity=0.70,line width= 0.4pt,line join=round,line cap=round,fill=fillColor,fill opacity=0.70] (117.79, 58.39) circle (  1.43);

\path[draw=drawColor,draw opacity=0.70,line width= 0.4pt,line join=round,line cap=round,fill=fillColor,fill opacity=0.70] (121.58, 58.39) circle (  1.43);

\path[draw=drawColor,draw opacity=0.70,line width= 0.4pt,line join=round,line cap=round,fill=fillColor,fill opacity=0.70] (125.38, 62.57) circle (  1.43);

\path[draw=drawColor,draw opacity=0.70,line width= 0.4pt,line join=round,line cap=round,fill=fillColor,fill opacity=0.70] (126.14, 65.91) circle (  1.43);

\path[draw=drawColor,draw opacity=0.70,line width= 0.4pt,line join=round,line cap=round,fill=fillColor,fill opacity=0.70] (123.86, 70.09) circle (  1.43);

\path[draw=drawColor,draw opacity=0.70,line width= 0.4pt,line join=round,line cap=round,fill=fillColor,fill opacity=0.70] (119.31, 69.26) circle (  1.43);

\path[draw=drawColor,draw opacity=0.70,line width= 0.4pt,line join=round,line cap=round,fill=fillColor,fill opacity=0.70] (116.27, 68.42) circle (  1.43);

\path[draw=drawColor,draw opacity=0.70,line width= 0.4pt,line join=round,line cap=round,fill=fillColor,fill opacity=0.70] (123.86, 70.09) circle (  1.43);

\path[draw=drawColor,draw opacity=0.70,line width= 0.4pt,line join=round,line cap=round,fill=fillColor,fill opacity=0.70] (126.14, 71.76) circle (  1.43);

\path[draw=drawColor,draw opacity=0.70,line width= 0.4pt,line join=round,line cap=round,fill=fillColor,fill opacity=0.70] (126.90, 74.27) circle (  1.43);

\path[draw=drawColor,draw opacity=0.70,line width= 0.4pt,line join=round,line cap=round,fill=fillColor,fill opacity=0.70] (126.90, 76.78) circle (  1.43);

\path[draw=drawColor,draw opacity=0.70,line width= 0.4pt,line join=round,line cap=round,fill=fillColor,fill opacity=0.70] (123.86, 77.61) circle (  1.43);

\path[draw=drawColor,draw opacity=0.70,line width= 0.4pt,line join=round,line cap=round,fill=fillColor,fill opacity=0.70] (121.58, 77.61) circle (  1.43);
\definecolor{drawColor}{gray}{0.20}

\path[draw=drawColor,line width= 0.6pt,line join=round,line cap=round] (110.29, 53.25) rectangle (135.33, 85.71);
\end{scope}
\begin{scope}
\path[clip] (110.29, 20.71) rectangle (135.33, 53.18);
\definecolor{fillColor}{RGB}{255,255,255}

\path[fill=fillColor] (110.29, 20.71) rectangle (135.33, 53.18);
\definecolor{drawColor}{gray}{0.92}

\path[draw=drawColor,line width= 0.3pt,line join=round] (110.29, 24.25) --
	(135.33, 24.25);

\path[draw=drawColor,line width= 0.3pt,line join=round] (110.29, 32.61) --
	(135.33, 32.61);

\path[draw=drawColor,line width= 0.3pt,line join=round] (110.29, 40.97) --
	(135.33, 40.97);

\path[draw=drawColor,line width= 0.3pt,line join=round] (110.29, 49.33) --
	(135.33, 49.33);

\path[draw=drawColor,line width= 0.3pt,line join=round] (111.66, 20.71) --
	(111.66, 53.18);

\path[draw=drawColor,line width= 0.3pt,line join=round] (119.25, 20.71) --
	(119.25, 53.18);

\path[draw=drawColor,line width= 0.3pt,line join=round] (126.84, 20.71) --
	(126.84, 53.18);

\path[draw=drawColor,line width= 0.3pt,line join=round] (134.43, 20.71) --
	(134.43, 53.18);

\path[draw=drawColor,line width= 0.6pt,line join=round] (110.29, 28.43) --
	(135.33, 28.43);

\path[draw=drawColor,line width= 0.6pt,line join=round] (110.29, 36.79) --
	(135.33, 36.79);

\path[draw=drawColor,line width= 0.6pt,line join=round] (110.29, 45.15) --
	(135.33, 45.15);

\path[draw=drawColor,line width= 0.6pt,line join=round] (115.45, 20.71) --
	(115.45, 53.18);

\path[draw=drawColor,line width= 0.6pt,line join=round] (123.04, 20.71) --
	(123.04, 53.18);

\path[draw=drawColor,line width= 0.6pt,line join=round] (130.63, 20.71) --
	(130.63, 53.18);
\definecolor{drawColor}{gray}{0.70}

\path[draw=drawColor,line width= 0.6pt,line join=round] (115.98, 24.70) --
	(120.53, 22.19) --
	(127.36, 23.86) --
	(131.16, 28.04) --
	(127.36, 33.06) --
	(124.33, 36.40) --
	(118.26, 39.75) --
	(124.33, 39.75) --
	(127.36, 41.42) --
	(125.09, 44.76) --
	(122.81, 46.43) --
	(119.02, 48.10) --
	(115.98, 49.78);
\definecolor{drawColor}{RGB}{102,102,102}
\definecolor{fillColor}{RGB}{102,102,102}

\path[draw=drawColor,draw opacity=0.70,line width= 0.4pt,line join=round,line cap=round,fill=fillColor,fill opacity=0.70] (115.98, 24.70) circle (  1.43);

\path[draw=drawColor,draw opacity=0.70,line width= 0.4pt,line join=round,line cap=round,fill=fillColor,fill opacity=0.70] (120.53, 22.19) circle (  1.43);

\path[draw=drawColor,draw opacity=0.70,line width= 0.4pt,line join=round,line cap=round,fill=fillColor,fill opacity=0.70] (127.36, 23.86) circle (  1.43);

\path[draw=drawColor,draw opacity=0.70,line width= 0.4pt,line join=round,line cap=round,fill=fillColor,fill opacity=0.70] (131.16, 28.04) circle (  1.43);

\path[draw=drawColor,draw opacity=0.70,line width= 0.4pt,line join=round,line cap=round,fill=fillColor,fill opacity=0.70] (127.36, 33.06) circle (  1.43);

\path[draw=drawColor,draw opacity=0.70,line width= 0.4pt,line join=round,line cap=round,fill=fillColor,fill opacity=0.70] (124.33, 36.40) circle (  1.43);

\path[draw=drawColor,draw opacity=0.70,line width= 0.4pt,line join=round,line cap=round,fill=fillColor,fill opacity=0.70] (118.26, 39.75) circle (  1.43);

\path[draw=drawColor,draw opacity=0.70,line width= 0.4pt,line join=round,line cap=round,fill=fillColor,fill opacity=0.70] (124.33, 39.75) circle (  1.43);

\path[draw=drawColor,draw opacity=0.70,line width= 0.4pt,line join=round,line cap=round,fill=fillColor,fill opacity=0.70] (127.36, 41.42) circle (  1.43);

\path[draw=drawColor,draw opacity=0.70,line width= 0.4pt,line join=round,line cap=round,fill=fillColor,fill opacity=0.70] (125.09, 44.76) circle (  1.43);

\path[draw=drawColor,draw opacity=0.70,line width= 0.4pt,line join=round,line cap=round,fill=fillColor,fill opacity=0.70] (122.81, 46.43) circle (  1.43);

\path[draw=drawColor,draw opacity=0.70,line width= 0.4pt,line join=round,line cap=round,fill=fillColor,fill opacity=0.70] (119.02, 48.10) circle (  1.43);

\path[draw=drawColor,draw opacity=0.70,line width= 0.4pt,line join=round,line cap=round,fill=fillColor,fill opacity=0.70] (115.98, 49.78) circle (  1.43);
\definecolor{drawColor}{gray}{0.20}

\path[draw=drawColor,line width= 0.6pt,line join=round,line cap=round] (110.29, 20.71) rectangle (135.33, 53.18);
\end{scope}
\begin{scope}
\path[clip] (  0.00,  0.00) rectangle (155.81,180.67);
\definecolor{drawColor}{RGB}{0,0,0}

\node[text=drawColor,anchor=base,inner sep=0pt, outer sep=0pt, scale=  1.10] at (110.26,  7.64) {\small$\Re(\beta)$};
\end{scope}
\begin{scope}
\path[clip] (  0.00,  0.00) rectangle (155.81,180.67);
\definecolor{drawColor}{gray}{0.25}

\node[text=drawColor,anchor=base,inner sep=0pt, outer sep=0pt, scale=  0.88] at (110.26, 87.42) {samples};
\end{scope}
\end{tikzpicture}

%% file: figures/digitsim.tex
\begin{tikzpicture}[x=1pt,y=1pt]
\definecolor{fillColor}{RGB}{255,255,255}
\path[use as bounding box,fill=fillColor,fill opacity=0.00] (0,0) rectangle ( 93.08,180.67);
\begin{scope}
\path[clip] (  0.00,  0.00) rectangle ( 93.08,180.67);
\definecolor{drawColor}{RGB}{255,255,255}
\definecolor{fillColor}{RGB}{255,255,255}

\path[draw=drawColor,line width= 0.6pt,line join=round,line cap=round,fill=fillColor] (  0.00,  0.00) rectangle ( 93.08,180.67);
\end{scope}
\begin{scope}
\path[clip] ( 32.01, 89.73) rectangle ( 87.48,151.93);
\definecolor{drawColor}{gray}{0.80}

\path[draw=drawColor,line width= 0.3pt,line join=round] ( 32.01, 98.26) --
	( 87.48, 98.26);

\path[draw=drawColor,line width= 0.3pt,line join=round] ( 32.01,114.15) --
	( 87.48,114.15);

\path[draw=drawColor,line width= 0.3pt,line join=round] ( 32.01,130.04) --
	( 87.48,130.04);

\path[draw=drawColor,line width= 0.3pt,line join=round] ( 32.01,145.93) --
	( 87.48,145.93);
\definecolor{drawColor}{gray}{0.50}

\path[draw=drawColor,line width= 0.6pt,line join=round] ( 32.01, 90.32) --
	( 87.48, 90.32);

\path[draw=drawColor,line width= 0.6pt,line join=round] ( 32.01,106.21) --
	( 87.48,106.21);

\path[draw=drawColor,line width= 0.6pt,line join=round] ( 32.01,122.10) --
	( 87.48,122.10);

\path[draw=drawColor,line width= 0.6pt,line join=round] ( 32.01,137.99) --
	( 87.48,137.99);
\definecolor{drawColor}{gray}{0.20}

\path[draw=drawColor,line width= 0.4pt,line join=round,line cap=round] ( 42.41,145.10) circle (  1.96);

\path[draw=drawColor,line width= 0.4pt,line join=round,line cap=round] ( 42.41,149.10) circle (  1.96);

\path[draw=drawColor,line width= 0.6pt,line join=round] ( 42.41,123.28) -- ( 42.41,137.70);

\path[draw=drawColor,line width= 0.6pt,line join=round] ( 42.41,111.86) -- ( 42.41,103.57);
\definecolor{fillColor}{RGB}{255,255,255}

\path[draw=drawColor,line width= 0.6pt,line join=round,line cap=round,fill=fillColor] ( 35.91,123.28) --
	( 35.91,111.86) --
	( 48.91,111.86) --
	( 48.91,123.28) --
	( 35.91,123.28) --
	cycle;

\path[draw=drawColor,line width= 1.1pt,line join=round] ( 35.91,117.27) -- ( 48.91,117.27);

\path[draw=drawColor,line width= 0.4pt,line join=round,line cap=round] ( 59.75,123.45) circle (  1.96);

\path[draw=drawColor,line width= 0.4pt,line join=round,line cap=round] ( 59.75,123.34) circle (  1.96);

\path[draw=drawColor,line width= 0.6pt,line join=round] ( 59.75,112.03) -- ( 59.75,120.53);

\path[draw=drawColor,line width= 0.6pt,line join=round] ( 59.75,104.85) -- ( 59.75, 96.40);

\path[draw=drawColor,line width= 0.6pt,line join=round,line cap=round,fill=fillColor] ( 53.24,112.03) --
	( 53.24,104.85) --
	( 66.25,104.85) --
	( 66.25,112.03) --
	( 53.24,112.03) --
	cycle;

\path[draw=drawColor,line width= 1.1pt,line join=round] ( 53.24,108.28) -- ( 66.25,108.28);

\path[draw=drawColor,line width= 0.4pt,line join=round,line cap=round] ( 77.08,130.88) circle (  1.96);

\path[draw=drawColor,line width= 0.6pt,line join=round] ( 77.08,105.98) -- ( 77.08,112.49);

\path[draw=drawColor,line width= 0.6pt,line join=round] ( 77.08, 99.49) -- ( 77.08, 92.56);

\path[draw=drawColor,line width= 0.6pt,line join=round,line cap=round,fill=fillColor] ( 70.58,105.98) --
	( 70.58, 99.49) --
	( 83.58, 99.49) --
	( 83.58,105.98) --
	( 70.58,105.98) --
	cycle;

\path[draw=drawColor,line width= 1.1pt,line join=round] ( 70.58,102.73) -- ( 83.58,102.73);
\definecolor{drawColor}{RGB}{86,177,247}
\definecolor{fillColor}{RGB}{86,177,247}

\path[draw=drawColor,draw opacity=0.90,line width= 0.4pt,line join=round,line cap=round,fill=fillColor,fill opacity=0.90] ( 42.41,132.42) circle (  1.96);
\definecolor{drawColor}{RGB}{55,116,165}
\definecolor{fillColor}{RGB}{55,116,165}

\path[draw=drawColor,draw opacity=0.90,line width= 0.4pt,line join=round,line cap=round,fill=fillColor,fill opacity=0.90] ( 42.41,123.27) circle (  1.96);
\definecolor{drawColor}{RGB}{41,86,125}
\definecolor{fillColor}{RGB}{41,86,125}

\path[draw=drawColor,draw opacity=0.90,line width= 0.4pt,line join=round,line cap=round,fill=fillColor,fill opacity=0.90] ( 59.75,118.61) circle (  1.96);
\definecolor{drawColor}{RGB}{19,43,67}
\definecolor{fillColor}{RGB}{19,43,67}

\path[draw=drawColor,draw opacity=0.90,line width= 0.4pt,line join=round,line cap=round,fill=fillColor,fill opacity=0.90] ( 77.08,111.16) circle (  1.96);
\end{scope}
\begin{scope}
\path[clip] (  0.00,  0.00) rectangle ( 93.08,180.67);
\definecolor{drawColor}{RGB}{0,0,0}

\path[draw=drawColor,line width= 0.6pt,line join=round] ( 32.01,151.93) --
	( 87.48,151.93);
\end{scope}
\begin{scope}
\path[clip] (  0.00,  0.00) rectangle ( 93.08,180.67);
\definecolor{drawColor}{gray}{0.30}

\node[text=drawColor,anchor=base,inner sep=0pt, outer sep=0pt, scale=  0.88] at ( 42.41,156.88) {4};

\node[text=drawColor,anchor=base,inner sep=0pt, outer sep=0pt, scale=  0.88] at ( 59.75,156.88) {10};

\node[text=drawColor,anchor=base,inner sep=0pt, outer sep=0pt, scale=  0.88] at ( 77.08,156.88) {30};
\end{scope}
\begin{scope}
\path[clip] (  0.00,  0.00) rectangle ( 93.08,180.67);
\definecolor{drawColor}{RGB}{0,0,0}

\path[draw=drawColor,line width= 0.6pt,line join=round] ( 42.41,151.93) --
	( 42.41,154.68);

\path[draw=drawColor,line width= 0.6pt,line join=round] ( 59.75,151.93) --
	( 59.75,154.68);

\path[draw=drawColor,line width= 0.6pt,line join=round] ( 77.08,151.93) --
	( 77.08,154.68);
\end{scope}
\begin{scope}
\path[clip] (  0.00,  0.00) rectangle ( 93.08,180.67);
\definecolor{drawColor}{gray}{0.30}

\node[text=drawColor,anchor=base east,inner sep=0pt, outer sep=0pt, scale=  0.70] at ( 27.06, 87.89) {0.0};

\node[text=drawColor,anchor=base east,inner sep=0pt, outer sep=0pt, scale=  0.70] at ( 27.06,103.78) {0.5};

\node[text=drawColor,anchor=base east,inner sep=0pt, outer sep=0pt, scale=  0.70] at ( 27.06,119.67) {1.0};

\node[text=drawColor,anchor=base east,inner sep=0pt, outer sep=0pt, scale=  0.70] at ( 27.06,135.56) {1.5};
\end{scope}
\begin{scope}
\path[clip] (  0.00,  0.00) rectangle ( 93.08,180.67);
\definecolor{drawColor}{gray}{0.25}

\node[text=drawColor,anchor=base,inner sep=0pt, outer sep=0pt, scale=  0.88] at ( 59.75,167.40) {$n$ digits ``\textit{3}''};
\end{scope}
\begin{scope}
\path[clip] (  0.00,  0.00) rectangle ( 93.08,180.67);
\definecolor{drawColor}{gray}{0.25}

\node[text=drawColor,rotate= 90.00,anchor=base,inner sep=0pt, outer sep=0pt, scale=  0.88] at ( 12.63,120.83) {\small$d([\hat{\mu}_l], [\mu]) / \sigma$};
\end{scope}
\begin{scope}
\path[clip] (  5.50,  5.50) rectangle ( 87.58, 86.73);

\path[] (  5.50,  5.50) rectangle ( 87.58, 86.73);
\end{scope}
\begin{scope}
\path[clip] ( 32.01, 20.71) rectangle ( 87.48, 86.73);
\definecolor{fillColor}{RGB}{255,255,255}

\path[fill=fillColor] ( 32.01, 20.71) rectangle ( 87.48, 86.73);
\definecolor{drawColor}{gray}{0.92}

\path[draw=drawColor,line width= 0.3pt,line join=round] ( 32.01, 25.89) --
	( 87.48, 25.89);

\path[draw=drawColor,line width= 0.3pt,line join=round] ( 32.01, 38.82) --
	( 87.48, 38.82);

\path[draw=drawColor,line width= 0.3pt,line join=round] ( 32.01, 51.75) --
	( 87.48, 51.75);

\path[draw=drawColor,line width= 0.3pt,line join=round] ( 32.01, 64.68) --
	( 87.48, 64.68);

\path[draw=drawColor,line width= 0.3pt,line join=round] ( 32.01, 77.61) --
	( 87.48, 77.61);

\path[draw=drawColor,line width= 0.3pt,line join=round] ( 32.79, 20.71) --
	( 32.79, 86.73);

\path[draw=drawColor,line width= 0.3pt,line join=round] ( 45.71, 20.71) --
	( 45.71, 86.73);

\path[draw=drawColor,line width= 0.3pt,line join=round] ( 58.64, 20.71) --
	( 58.64, 86.73);

\path[draw=drawColor,line width= 0.3pt,line join=round] ( 71.57, 20.71) --
	( 71.57, 86.73);

\path[draw=drawColor,line width= 0.3pt,line join=round] ( 84.50, 20.71) --
	( 84.50, 86.73);

\path[draw=drawColor,line width= 0.6pt,line join=round] ( 32.01, 32.36) --
	( 87.48, 32.36);

\path[draw=drawColor,line width= 0.6pt,line join=round] ( 32.01, 45.29) --
	( 87.48, 45.29);

\path[draw=drawColor,line width= 0.6pt,line join=round] ( 32.01, 58.22) --
	( 87.48, 58.22);

\path[draw=drawColor,line width= 0.6pt,line join=round] ( 32.01, 71.14) --
	( 87.48, 71.14);

\path[draw=drawColor,line width= 0.6pt,line join=round] ( 32.01, 84.07) --
	( 87.48, 84.07);

\path[draw=drawColor,line width= 0.6pt,line join=round] ( 39.25, 20.71) --
	( 39.25, 86.73);

\path[draw=drawColor,line width= 0.6pt,line join=round] ( 52.18, 20.71) --
	( 52.18, 86.73);

\path[draw=drawColor,line width= 0.6pt,line join=round] ( 65.11, 20.71) --
	( 65.11, 86.73);

\path[draw=drawColor,line width= 0.6pt,line join=round] ( 78.04, 20.71) --
	( 78.04, 86.73);
\definecolor{drawColor}{RGB}{86,177,247}

\path[draw=drawColor,draw opacity=0.90,line width= 0.6pt,line join=round] ( 49.18, 30.67) --
	( 50.39, 29.75) --
	( 51.62, 28.89) --
	( 52.88, 28.10) --
	( 54.17, 27.37) --
	( 55.50, 26.71) --
	( 56.86, 26.10) --
	( 58.26, 25.55) --
	( 59.70, 25.06) --
	( 61.19, 24.63) --
	( 62.73, 24.26) --
	( 64.28, 23.97) --
	( 65.77, 23.79) --
	( 67.23, 23.72) --
	( 68.66, 23.74) --
	( 70.07, 23.87) --
	( 71.46, 24.09) --
	( 72.84, 24.42) --
	( 74.21, 24.85) --
	( 75.60, 25.39) --
	( 76.99, 26.04) --
	( 78.31, 26.78) --
	( 79.48, 27.58) --
	( 80.51, 28.45) --
	( 81.42, 29.38) --
	( 82.22, 30.39) --
	( 82.93, 31.49) --
	( 83.53, 32.68) --
	( 84.05, 33.98) --
	( 84.46, 35.42) --
	( 84.78, 37.00) --
	( 84.96, 38.66) --
	( 84.96, 40.34) --
	( 84.78, 42.06) --
	( 84.42, 43.83) --
	( 83.85, 45.67) --
	( 83.06, 47.61) --
	( 82.04, 49.64) --
	( 80.74, 51.80) --
	( 79.15, 54.09) --
	( 77.23, 56.52) --
	( 75.21, 58.87) --
	( 73.33, 60.95) --
	( 71.59, 62.76) --
	( 69.99, 64.33) --
	( 68.50, 65.67) --
	( 67.13, 66.80) --
	( 65.88, 67.75) --
	( 64.73, 68.52) --
	( 63.68, 69.14) --
	( 62.71, 69.61) --
	( 61.93, 70.00) --
	( 61.41, 70.33) --
	( 61.08, 70.61) --
	( 60.89, 70.85) --
	( 60.81, 71.06) --
	( 60.81, 71.26) --
	( 60.90, 71.47) --
	( 61.09, 71.73) --
	( 61.43, 72.05) --
	( 61.98, 72.43) --
	( 62.61, 72.85) --
	( 63.15, 73.25) --
	( 63.61, 73.65) --
	( 64.01, 74.04) --
	( 64.35, 74.43) --
	( 64.63, 74.81) --
	( 64.86, 75.19) --
	( 65.04, 75.57) --
	( 65.18, 75.95) --
	( 65.28, 76.35) --
	( 65.33, 76.74) --
	( 65.33, 77.12) --
	( 65.29, 77.51) --
	( 65.18, 77.91) --
	( 65.03, 78.32) --
	( 64.81, 78.75) --
	( 64.52, 79.20) --
	( 64.14, 79.68) --
	( 63.68, 80.19) --
	( 63.11, 80.73) --
	( 62.47, 81.26) --
	( 61.78, 81.74) --
	( 61.03, 82.17) --
	( 60.22, 82.55) --
	( 59.34, 82.88) --
	( 58.38, 83.16) --
	( 57.34, 83.39) --
	( 56.20, 83.57) --
	( 54.96, 83.68) --
	( 53.61, 83.73) --
	( 52.15, 83.73) --
	( 50.61, 83.68) --
	( 48.97, 83.58) --
	( 47.23, 83.44) --
	( 45.40, 83.24) --
	( 43.45, 82.99) --
	( 41.39, 82.67) --
	( 39.22, 82.30) --
	( 36.94, 81.86) --
	( 34.53, 81.36);
\definecolor{drawColor}{RGB}{41,86,125}

\path[draw=drawColor,draw opacity=0.90,line width= 0.6pt,line join=round] ( 46.67, 35.96) --
	( 47.58, 35.75) --
	( 48.54, 35.52) --
	( 49.55, 35.26) --
	( 50.62, 34.96) --
	( 51.76, 34.63) --
	( 52.95, 34.26) --
	( 54.21, 33.86) --
	( 55.53, 33.41) --
	( 56.93, 32.93) --
	( 58.39, 32.41) --
	( 59.87, 31.93) --
	( 61.32, 31.59) --
	( 62.75, 31.38) --
	( 64.18, 31.29) --
	( 65.61, 31.33) --
	( 67.05, 31.48) --
	( 68.52, 31.77) --
	( 70.02, 32.19) --
	( 71.56, 32.75) --
	( 73.16, 33.46) --
	( 74.65, 34.28) --
	( 75.90, 35.15) --
	( 76.95, 36.06) --
	( 77.81, 37.03) --
	( 78.52, 38.06) --
	( 79.08, 39.18) --
	( 79.50, 40.39) --
	( 79.78, 41.74) --
	( 79.93, 43.23) --
	( 79.91, 44.90) --
	( 79.72, 46.60) --
	( 79.36, 48.16) --
	( 78.83, 49.62) --
	( 78.13, 50.99) --
	( 77.25, 52.30) --
	( 76.17, 53.56) --
	( 74.86, 54.78) --
	( 73.30, 55.97) --
	( 71.44, 57.13) --
	( 69.25, 58.26) --
	( 67.01, 59.29) --
	( 65.01, 60.16) --
	( 63.23, 60.88) --
	( 61.67, 61.47) --
	( 60.30, 61.94) --
	( 59.11, 62.30) --
	( 58.09, 62.57) --
	( 57.22, 62.75) --
	( 56.49, 62.85) --
	( 55.89, 62.90) --
	( 55.45, 62.92) --
	( 55.19, 62.94) --
	( 55.07, 62.95) --
	( 55.03, 62.96) --
	( 55.02, 62.97) --
	( 55.04, 62.98) --
	( 55.13, 63.00) --
	( 55.33, 63.04) --
	( 55.70, 63.09) --
	( 56.28, 63.17) --
	( 57.03, 63.30) --
	( 57.89, 63.49) --
	( 58.87, 63.77) --
	( 59.98, 64.13) --
	( 61.22, 64.60) --
	( 62.61, 65.18) --
	( 64.15, 65.89) --
	( 65.86, 66.73) --
	( 67.75, 67.71) --
	( 69.82, 68.85) --
	( 71.74, 70.01) --
	( 73.21, 71.06) --
	( 74.29, 71.98) --
	( 75.05, 72.81) --
	( 75.56, 73.54) --
	( 75.85, 74.21) --
	( 75.98, 74.83) --
	( 75.96, 75.44) --
	( 75.78, 76.07) --
	( 75.42, 76.76) --
	( 74.92, 77.46) --
	( 74.33, 78.12) --
	( 73.65, 78.75) --
	( 72.85, 79.34) --
	( 71.92, 79.89) --
	( 70.86, 80.42) --
	( 69.64, 80.91) --
	( 68.25, 81.37) --
	( 66.67, 81.78) --
	( 64.88, 82.14) --
	( 63.03, 82.44) --
	( 61.26, 82.66) --
	( 59.57, 82.81) --
	( 57.96, 82.89) --
	( 56.42, 82.91) --
	( 54.96, 82.87) --
	( 53.55, 82.78) --
	( 52.21, 82.62) --
	( 50.92, 82.41) --
	( 49.69, 82.15);
\definecolor{drawColor}{RGB}{19,43,67}

\path[draw=drawColor,draw opacity=0.90,line width= 0.6pt,line join=round] ( 45.11, 37.45) --
	( 46.59, 36.58) --
	( 48.06, 35.78) --
	( 49.51, 35.08) --
	( 50.94, 34.44) --
	( 52.36, 33.88) --
	( 53.77, 33.40) --
	( 55.17, 32.98) --
	( 56.56, 32.63) --
	( 57.94, 32.35) --
	( 59.33, 32.13) --
	( 60.71, 31.99) --
	( 62.10, 31.96) --
	( 63.50, 32.03) --
	( 64.92, 32.20) --
	( 66.36, 32.48) --
	( 67.84, 32.87) --
	( 69.36, 33.38) --
	( 70.92, 34.02) --
	( 72.54, 34.79) --
	( 74.21, 35.70) --
	( 75.77, 36.68) --
	( 77.06, 37.66) --
	( 78.11, 38.64) --
	( 78.96, 39.61) --
	( 79.62, 40.59) --
	( 80.12, 41.59) --
	( 80.47, 42.63) --
	( 80.69, 43.72) --
	( 80.76, 44.88) --
	( 80.68, 46.14) --
	( 80.47, 47.42) --
	( 80.15, 48.62) --
	( 79.73, 49.77) --
	( 79.19, 50.86) --
	( 78.53, 51.92) --
	( 77.74, 52.95) --
	( 76.82, 53.96) --
	( 75.74, 54.95) --
	( 74.49, 55.92) --
	( 73.04, 56.88) --
	( 71.51, 57.80) --
	( 70.01, 58.65) --
	( 68.52, 59.43) --
	( 67.05, 60.15) --
	( 65.60, 60.80) --
	( 64.17, 61.38) --
	( 62.76, 61.91) --
	( 61.35, 62.38) --
	( 59.97, 62.79) --
	( 58.59, 63.15) --
	( 57.39, 63.44) --
	( 56.53, 63.64) --
	( 55.93, 63.77) --
	( 55.56, 63.84) --
	( 55.35, 63.88) --
	( 55.27, 63.89) --
	( 55.25, 63.89) --
	( 55.25, 63.88) --
	( 55.30, 63.86) --
	( 55.43, 63.81) --
	( 55.66, 63.74) --
	( 56.00, 63.72) --
	( 56.46, 63.75) --
	( 57.08, 63.84) --
	( 57.89, 64.03) --
	( 58.91, 64.32) --
	( 60.17, 64.75) --
	( 61.69, 65.32) --
	( 63.51, 66.07) --
	( 65.65, 67.01) --
	( 67.73, 68.03) --
	( 69.40, 69.00) --
	( 70.70, 69.93) --
	( 71.70, 70.80) --
	( 72.44, 71.63) --
	( 72.95, 72.43) --
	( 73.29, 73.21) --
	( 73.46, 74.00) --
	( 73.49, 74.82) --
	( 73.35, 75.71) --
	( 73.08, 76.59) --
	( 72.71, 77.38) --
	( 72.22, 78.12) --
	( 71.62, 78.80) --
	( 70.87, 79.44) --
	( 69.96, 80.05) --
	( 68.86, 80.62) --
	( 67.55, 81.15) --
	( 65.98, 81.63) --
	( 64.14, 82.06) --
	( 62.19, 82.41) --
	( 60.31, 82.69) --
	( 58.49, 82.88) --
	( 56.75, 83.01) --
	( 55.06, 83.07) --
	( 53.43, 83.06) --
	( 51.85, 82.98) --
	( 50.32, 82.84) --
	( 48.84, 82.64) --
	( 47.41, 82.38);
\definecolor{drawColor}{RGB}{55,116,165}

\path[draw=drawColor,draw opacity=0.90,line width= 0.6pt,line join=round] ( 46.50, 38.22) --
	( 47.51, 37.81) --
	( 48.56, 37.40) --
	( 49.67, 37.00) --
	( 50.84, 36.60) --
	( 52.08, 36.20) --
	( 53.37, 35.81) --
	( 54.73, 35.43) --
	( 56.16, 35.05) --
	( 57.65, 34.67) --
	( 59.21, 34.31) --
	( 60.79, 34.01) --
	( 62.31, 33.84) --
	( 63.81, 33.80) --
	( 65.27, 33.88) --
	( 66.72, 34.07) --
	( 68.17, 34.39) --
	( 69.62, 34.83) --
	( 71.07, 35.41) --
	( 72.55, 36.12) --
	( 74.06, 36.99) --
	( 75.43, 37.93) --
	( 76.54, 38.85) --
	( 77.40, 39.77) --
	( 78.07, 40.69) --
	( 78.56, 41.62) --
	( 78.90, 42.58) --
	( 79.09, 43.59) --
	( 79.13, 44.68) --
	( 79.01, 45.86) --
	( 78.72, 47.17) --
	( 78.29, 48.52) --
	( 77.77, 49.78) --
	( 77.17, 50.98) --
	( 76.48, 52.12) --
	( 75.71, 53.21) --
	( 74.83, 54.26) --
	( 73.85, 55.26) --
	( 72.76, 56.23) --
	( 71.55, 57.17) --
	( 70.20, 58.07) --
	( 68.79, 58.91) --
	( 67.38, 59.66) --
	( 65.98, 60.33) --
	( 64.58, 60.91) --
	( 63.17, 61.42) --
	( 61.76, 61.85) --
	( 60.34, 62.21) --
	( 58.90, 62.50) --
	( 57.45, 62.71) --
	( 55.98, 62.85) --
	( 54.74, 62.94) --
	( 53.92, 63.00) --
	( 53.44, 63.03) --
	( 53.20, 63.05) --
	( 53.12, 63.05) --
	( 53.12, 63.05) --
	( 53.14, 63.05) --
	( 53.27, 63.04) --
	( 53.59, 63.01) --
	( 54.20, 62.97) --
	( 55.03, 62.93) --
	( 55.97, 62.96) --
	( 57.05, 63.06) --
	( 58.27, 63.24) --
	( 59.64, 63.51) --
	( 61.18, 63.88) --
	( 62.90, 64.36) --
	( 64.81, 64.97) --
	( 66.93, 65.70) --
	( 69.26, 66.58) --
	( 71.41, 67.51) --
	( 73.05, 68.38) --
	( 74.25, 69.19) --
	( 75.10, 69.94) --
	( 75.66, 70.63) --
	( 75.99, 71.30) --
	( 76.14, 71.96) --
	( 76.12, 72.66) --
	( 75.92, 73.45) --
	( 75.48, 74.35) --
	( 74.89, 75.28) --
	( 74.22, 76.12) --
	( 73.47, 76.89) --
	( 72.64, 77.59) --
	( 71.71, 78.22) --
	( 70.68, 78.80) --
	( 69.52, 79.31) --
	( 68.23, 79.77) --
	( 66.79, 80.15) --
	( 65.19, 80.47) --
	( 63.53, 80.69) --
	( 61.94, 80.79) --
	( 60.39, 80.79) --
	( 58.89, 80.68) --
	( 57.43, 80.47) --
	( 56.00, 80.15) --
	( 54.59, 79.73) --
	( 53.20, 79.20) --
	( 51.81, 78.56) --
	( 50.43, 77.80);
\definecolor{drawColor}{RGB}{0,0,0}

\path[draw=drawColor,line width= 0.6pt,dash pattern=on 4pt off 4pt ,line join=round] ( 44.26, 36.51) --
	( 45.60, 35.80) --
	( 46.95, 35.14) --
	( 48.31, 34.53) --
	( 49.68, 33.96) --
	( 51.06, 33.44) --
	( 52.45, 32.96) --
	( 53.86, 32.52) --
	( 55.29, 32.13) --
	( 56.74, 31.77) --
	( 58.20, 31.47) --
	( 59.67, 31.22) --
	( 61.13, 31.08) --
	( 62.60, 31.02) --
	( 64.07, 31.06) --
	( 65.55, 31.19) --
	( 67.05, 31.41) --
	( 68.57, 31.73) --
	( 70.12, 32.15) --
	( 71.71, 32.68) --
	( 73.33, 33.32) --
	( 74.85, 34.05) --
	( 76.16, 34.83) --
	( 77.28, 35.67) --
	( 78.22, 36.56) --
	( 79.02, 37.52) --
	( 79.69, 38.56) --
	( 80.23, 39.70) --
	( 80.66, 40.94) --
	( 80.96, 42.32) --
	( 81.13, 43.86) --
	( 81.16, 45.43) --
	( 81.03, 46.89) --
	( 80.76, 48.28) --
	( 80.34, 49.60) --
	( 79.77, 50.88) --
	( 79.04, 52.12) --
	( 78.13, 53.35) --
	( 77.03, 54.57) --
	( 75.71, 55.78) --
	( 74.14, 57.00) --
	( 72.49, 58.15) --
	( 70.92, 59.17) --
	( 69.43, 60.08) --
	( 68.01, 60.87) --
	( 66.67, 61.56) --
	( 65.39, 62.15) --
	( 64.18, 62.65) --
	( 63.03, 63.06) --
	( 61.94, 63.39) --
	( 60.90, 63.64) --
	( 60.05, 63.84) --
	( 59.48, 64.00) --
	( 59.14, 64.12) --
	( 58.96, 64.21) --
	( 58.89, 64.27) --
	( 58.88, 64.31) --
	( 58.91, 64.35) --
	( 59.01, 64.41) --
	( 59.23, 64.50) --
	( 59.65, 64.61) --
	( 60.21, 64.76) --
	( 60.84, 64.97) --
	( 61.55, 65.24) --
	( 62.34, 65.59) --
	( 63.22, 66.01) --
	( 64.19, 66.51) --
	( 65.27, 67.11) --
	( 66.45, 67.82) --
	( 67.74, 68.62) --
	( 69.16, 69.55) --
	( 70.44, 70.49) --
	( 71.40, 71.36) --
	( 72.09, 72.15) --
	( 72.55, 72.88) --
	( 72.83, 73.56) --
	( 72.95, 74.24) --
	( 72.92, 74.92) --
	( 72.74, 75.64) --
	( 72.39, 76.44) --
	( 71.82, 77.32) --
	( 71.08, 78.22) --
	( 70.27, 79.02) --
	( 69.37, 79.74) --
	( 68.37, 80.38) --
	( 67.26, 80.95) --
	( 66.02, 81.45) --
	( 64.63, 81.87) --
	( 63.08, 82.22) --
	( 61.35, 82.47) --
	( 59.43, 82.62) --
	( 57.43, 82.70) --
	( 55.50, 82.73) --
	( 53.64, 82.72) --
	( 51.85, 82.66) --
	( 50.12, 82.56) --
	( 48.46, 82.42) --
	( 46.85, 82.24) --
	( 45.30, 82.02) --
	( 43.80, 81.76) --
	( 42.36, 81.47);
\definecolor{drawColor}{gray}{0.20}

\path[draw=drawColor,line width= 0.6pt,line join=round,line cap=round] ( 32.01, 20.71) rectangle ( 87.48, 86.73);
\end{scope}
\begin{scope}
\path[clip] (  0.00,  0.00) rectangle ( 93.08,180.67);
\definecolor{drawColor}{RGB}{0,0,0}

\node[text=drawColor,rotate= 90.00,anchor=base,inner sep=0pt, outer sep=0pt, scale=  1.10] at ( 24.37, 53.72) {\small$\Im(\beta)$};
\end{scope}
\begin{scope}
\path[clip] (  0.00,  0.00) rectangle ( 93.08,180.67);
\definecolor{drawColor}{RGB}{0,0,0}

\node[text=drawColor,anchor=base,inner sep=0pt, outer sep=0pt, scale=  1.10] at ( 59.75,  7.64) {\small$\Re(\beta)$};
\end{scope}
\end{tikzpicture}

%% file: figures/sophia_profile.pdf_tex
\begingroup%
  \makeatletter%
  \providecommand\color[2][]{%
    \errmessage{(Inkscape) Color is used for the text in Inkscape, but the package 'color.sty' is not loaded}%
    \renewcommand\color[2][]{}%
  }%
  \providecommand\transparent[1]{%
    \errmessage{(Inkscape) Transparency is used (non-zero) for the text in Inkscape, but the package 'transparent.sty' is not loaded}%
    \renewcommand\transparent[1]{}%
  }%
  \providecommand\rotatebox[2]{#2}%
  \newcommand*\fsize{\dimexpr\f@size pt\relax}%
  \newcommand*\lineheight[1]{\fontsize{\fsize}{#1\fsize}\selectfont}%
  \ifx\svgwidth\undefined%
    \setlength{\unitlength}{108bp}%
    \ifx\svgscale\undefined%
      \relax%
    \else%
      \setlength{\unitlength}{\unitlength * \real{\svgscale}}%
    \fi%
  \else%
    \setlength{\unitlength}{\svgwidth}%
  \fi%
  \global\let\svgwidth\undefined%
  \global\let\svgscale\undefined%
  \makeatother%
  \begin{picture}(1,2)%
    \lineheight{1}%
    \setlength\tabcolsep{0pt}%
    \put(0,0){\includegraphics[width=\unitlength,page=1]{sophia_profile.pdf}}%
    \put(0.45625474,0.97467856){\color[rgb]{0.22745098,0.22745098,0.22745098}\makebox(0,0)[lt]{\lineheight{1.25}\smash{\begin{tabular}[t]{l}\footnotesize$\mathtt{asa}$\end{tabular}}}}%
    \put(0.57462973,1.03862534){\color[rgb]{0.22745098,0.22745098,0.22745098}\makebox(0,0)[lt]{\lineheight{1.25}\smash{\begin{tabular}[t]{l}\footnotesize$\mathtt{isi}$\end{tabular}}}}%
    \put(0.10785462,0.57991296){\color[rgb]{0.22745098,0.22745098,0.22745098}\makebox(0,0)[lt]{\lineheight{0.85000002}\smash{\begin{tabular}[t]{l}\small tongue \\\quad \small muscle\end{tabular}}}}%
  \end{picture}%
\endgroup%

%% file: supplement.tex
\section{Hermitian covariance smoothing} 

\subsection{Complex processes and rotation invariance}

In the following, we detail prerequisites on linear operators and proof Theorem \ref{lem:bivariateCov} and \ref{lem:bivariateCov_cplxsymmetry}. Subsequently, Proposition \ref{prop:bivariateCov} substantiates the relation of complex and real covariance surfaces indicated in the main manuscript.


We widely follow \citet{Hsing2015TheoreticalFoundations} in their introduction of functional data fundamentals, but re-state required statements underlying Section 
\ref{sec:covariance:complex_processes} 
for the complex case, since they nominally focus on real Hilbert spaces. Moreover, we give a Bochner integral free definition of mean elements and covariance operators  to avoid introduction of additional notions. 

Let $\mathbb{H}$ denote a Hilbert space over $\mathbb{C}$ or $\mathbb{R}$.

\begin{theorem}
	\label{thm:eigenbasis}
	Let $\Omega$ be a compact self-adjoint operator on $\mathbb{H}$. 
		Then there exists a sequence of countably many real eigenvalues $\lambda_1, \lambda_2, \dots \in \mathbb{R}$ of $\Omega$ with corresponding orthogonal eigenvectors $e_1, e_2, \dots \in \mathbb{H}$ and $\lambda_1 \geq \lambda_2 \geq \dots$ such that $\{e_k\}_k$ (called eigenbasis of $\Omega$) is an orthonormal basis of the closure $\overline{\Omega(\mathbb{H})}$ of the image of $\Omega$ and for every $x \in \mathbb{H}$ 
		$$ \Omega(x) = \sum_{k\geq1} \lambda_k \langle e_k, x \rangle e_k.$$
	\end{theorem}
	\begin{proof}
		Compare \citet{Rynne2007linearFuncAna}, Chapter 7.3.
	\end{proof}
	
	\begin{definition}
		Let $Y$ be a random element in $\mathbb{H}$ with $\mathbb{E}\left( \|Y\|^2 \right) < \infty$. Then 
		\begin{enumerate}[i)]
			\item the mean element $\mu \in \mathbb{H}$ of $Y$ is defined by $\langle f, \mu \rangle = \mathbb{E}\left( \langle f, Y \rangle \right)$ for all $f\in\mathbb{H}$.
			\item the covariance operator $\Sigma: \mathbb{H} \rightarrow \mathbb{H}$ of $Y$ is defined by $ \langle \Sigma(e) , f \rangle = \mathbb{E}\left( \langle Y - \mu, f \rangle \langle e, Y - \mu \rangle \right) $ for all $e,f \in \mathbb{H}$.
		\end{enumerate}
	\end{definition}
	
	\begin{proposition} \label{lem:covop} Consider $\mu$ and $\Sigma$ as above.
		\begin{enumerate}[i)]
			\item $\mu$ and $\Sigma$ are well-defined.
			\item $\Sigma$ is a nonnegative-definite (thus self-adjoint), trace-class and, hence, also compact linear operator. 
		\end{enumerate}
	\end{proposition}
	\begin{proof}
		\begin{enumerate}[i)]
			\item Since $\mathbb{E}\left( \|Y\|^2 \right) < \infty$, Jensen's inequality yields $\mathbb{E}\left( \|Y\| \right) < \infty$, and therefore $\mathbb{E}\left( \langle f, Y \rangle \right) < \infty$ and also $\mathbb{E}\left( \langle Y - \mu, f \rangle \langle e, Y - \mu \rangle \right) < \infty$ for all $e,f \in \mathbb{H}$. Uniqueness of $\mu$ and $\Sigma$ follows from the Riesz Representation Theorem.
			\item Set $\mu = 0$ without loss of generality. Self-adjointness $ \langle \Sigma(e) , f \rangle = \mathbb{E}\left( \langle Y, f \rangle \langle e, Y\rangle \right)  = \langle e , \Sigma(f) \rangle$ and nonnegative-definiteness  $ \langle \Sigma(e) , e \rangle = \mathbb{E}\left( \langle Y, e \rangle \langle e, Y \rangle \right) = \mathbb{E}\left( |\langle e, Y \rangle|^2 \right)$ immediately follow from the definition. $\Sigma$ is trace-class, since for an orthonormal basis $\{e_k\}_k$ of $\mathbb{H}$ it holds that 
			$$\sum_k \langle \Sigma(e_k), e_k \rangle = \sum_k \mathbb{E}\left( |\langle e_k, Y \rangle|^2 \right) = \mathbb{E}\left(\| Y \|^2 \right) < \infty$$ as assumed in the definition. Trace-class operators are compact.
		\end{enumerate}
	\end{proof}
	
	\begin{corollary}
		The covariance operator $\Sigma$ of $Y$ with $\mathbb{E}(\|Y\|^2) <\infty$ 
		has an eigenbasis as descibed in Theorem \ref{thm:eigenbasis}.
	\end{corollary}
	\begin{proof}
		Immediately follows from Theorem \ref{thm:eigenbasis} and the self-adjointness and compactness of $\Sigma$ shown in Proposition \ref{lem:covop}.
	\end{proof}
	
	
	We proceed by proving Theorem \ref{lem:bivariateCov} and \ref{lem:bivariateCov_cplxsymmetry} in the main manuscript characterizing the relation of the covariance of a complex process $Y$ and the covariance of the corresponding bivariate real process $\mathbf{Y}$: 
	
	\begin{proof}[Theorem 1]
		For $x,y \in \mathbb{L}(\mathcal{T}, \mathbb{C})$ and assuming $\mu=0$ without loss of generality, 	
		$ \Re\left(\langle \Sigma(x) + \Omega(x), y \rangle)\right)
		= \Re\left(\mathbb{E}\left(  \langle x, Y \rangle \langle Y, y \rangle + \langle Y, x \rangle \langle Y, y \rangle \right) \right)
		=  \Re\left(\mathbb{E}\left( 2\,\Re\left( \langle Y, x \rangle \right)\langle Y, y \rangle \right)\right)
		=  2\, \mathbb{E}\left( \Re\left( \langle Y, x \rangle \right) \Re\left( \langle Y, y \rangle \right) \right)  
		= 2 \langle \mathbf{\Sigma} (\kappa(x)), \kappa(y) \rangle
		$. 
	\end{proof}
	
	\begin{proof}[Theorem 2]
		From complex symmetry of $\mathfrak{L}(Y)$ it follows that $\mathfrak{L}(\exp(\im\, \omega) Z_k) = \mathfrak{L}(\langle e_k, \exp(\im\, \omega)Y\rangle) = \mathfrak{L}(Z_k)$, $\langle \mu, f\rangle = \mathbb{E}(\langle Y, f\rangle) {=} \mathbb{E}[\langle -Y, f\rangle] = 0$, and $\langle \Omega(e), f\rangle = \mathbb{E}(\langle Y, e\rangle\langle Y, f\rangle) {=} \mathbb{E}(- \langle Y, e\rangle\langle Y, f\rangle) = 0$ for all $\omega, k,e,f$, which yields the first direction of the characterization via scores and, together with Theorem \ref{lem:bivariateCov}, 
		statement  \ref{lem:bivariateCov1}). 
		\ref{lem:bivariateCov2}) 
		follows from Theorem \ref{lem:bivariateCov}, 
		statement \ref{lem:bivariateCov1}) 
		and the fact that if $Z_k$ is complex symmetric, $\kappa(Z_k)$ has uncorrelated components with equal variance.
		Since $\exp(\im \omega) Y = \sum_{k\geq 1} \exp(\im \omega) Z_k e_k$ almost surely if $\mu=0$, the second direction of the characterization via scores follows.
	\end{proof}

	\begin{proposition}
		\label{prop:bivariateCov}
		Analogous to $\Sigma$, the bivariate covariance surface $\mathbf{C}(s,t)$ of $\mathbf{Y}=\kappa(Y)$ in $\mathbb{L}^2([0,1], \mathbb{R}^2)$ is characterized by the matrix of covariance and cross-covariance surfaces
		\begin{align*}\nonumber
			\mathbf{C}(s,t) &= \begin{pmatrix}
				\mathbb{E}\left(\Re\left(Y(s)\right)\Re\left(Y(t)\right)\right) & 
				\mathbb{E}\left(\Im\left(Y(s)\right)\Re\left(Y(t)\right)\right) \\
				\mathbb{E}\left(\Re\left(Y(s)\right)\Im\left(Y(t)\right)\right) & 
				\mathbb{E}\left(\Im\left(Y(s)\right)\Im\left(Y(t)\right)\right) \end{pmatrix} \\
			&= 
			\frac{1}{2} \begin{pmatrix}
				\Re\left(C(s,t) + R(s,t)\right) & \Im\left(R(s,t) - C(s,t)\right) \\
				\Im\left(C(s,t) + R(s,t)\right) & \Re\left(C(s,t) - R(s,t)\right)
			\end{pmatrix}
		\end{align*}
		determined by the pseudo-covariance surface $R(s,t) = \mathbb{E}\left(Y(s)Y(t)\right)$ in addition to the complex covariance surface $C(s,t)$. 
	\end{proposition}
	
	\begin{proof}
		\begin{align*}
			C(s,t) + R(s,t) &=
			\mathbb{E}\left(Y^\dagger(s) Y(t) + Y(s) Y(t)\right) = 
			\mathbb{E}\left(\left(2\,\Re\left(Y(s)\right) + 0\right) Y(t)\right)\\ 
			&= 2\underbrace{\mathbb{E}\left(\Re\left(Y(s)\right)\Re\left(Y(t)\right)\right)}_{
				\frac{1}{2}\Re\left(C(s,t) + R(s,t)\right)
			} + 2\im\,\underbrace{\mathbb{E}\left(\Re\left(Y(s)\right)\Im\left(Y(t)\right)\right)}_{
				\frac{1}{2}\Im\left(C(s,t) + R(s,t)\right)
			}\\
			C(s,t) - R(s,t) &=
			\mathbb{E}\left(Y^\dagger(s) Y(t) - Y(s) Y(t)\right) = 
			\mathbb{E}\left(\left(0 - 2\im\,\Im\left(Y(s)\right) \right) Y(t)\right) \\
			&= -2\im\,\underbrace{\mathbb{E}\left(\Im\left(Y(s)\right)\Re\left(Y(t)\right)\right)}_{
				\frac{1}{2}\Im\left(R(s,t) - C(s,t)\right)
			} + 2\,\underbrace{\mathbb{E}\left(\Im\left(Y(s)\right)\Im\left(Y(t)\right)\right)}_{
				\frac{1}{2}\Re\left(C(s,t) - R(s,t)\right)
			}
		\end{align*}
		which shows the desired form. 
	\end{proof}

	\section{Elastic full Procrustes analysis}
	
	\subsection{Full Procrustes analysis in the square-root-velocity framework}
	
	In the following, we start by proving Proposition \ref{lem:ifullProcMean} and use Proposition \ref{lem:ifullProcMean}
	i) to show Proposition \ref{lem:efullProcDist} before proving Proposition \ref{lem:efullProcMean}
	subsequently.
	
	\begin{proof}[Proposition \ref{lem:ifullProcMean} i) and ii)]
		$d_{\not \mathcal{E}}$ defines a metric on $\tilde{\mathfrak{B}}$:
			\begin{align}\nonumber
				d^2_{\not \mathcal{E}}((\beta_1), (\beta_2)) &= \inf_{u\in \mathbb{C}} \|q_1 - u\, q_2 \|^2 = \inf_{u\in \mathbb{C}} \big[1 - \overbrace{u}^{=r_1 \exp(\im\, \omega_1)} \underbrace{\langle q_1, q_2\rangle}_{=r_2 \exp(\im\, \omega_2)} - u^\dagger \langle q_2, q_1\rangle +  |u|^2 \big]\\
				\nonumber &= \inf_{r_1>0,\, \omega_1\in\mathbb{R}} \big[1 - r_1 r_2 \exp(\im\, (\omega_1 + \omega_2)) - r_1 r_2 \exp(-\im\, (\omega_1 + \omega_2)) + r_1^2\big]\\
				 &= \inf_{r_1>0,\, \omega_1\in\mathbb{R}} \big[1 - 2 r_1 r_2 \cos(\omega_1 + \omega_2) + r_1^2\big] \overset{\omega_1 = -\omega_2}{=} \inf_{r_1>0} \big[1 - 2 r_1 r_2  + r_1^2 \big]\label{optimRot} \\ 
				&= \inf_{r_1>0} \big[1 - r_2^2  + (r_1 - r_2)^2\big] \overset{r_1 = r_2}{=} 1 - |\langle q_1, q_2\rangle|^2 = \|q_1 - \langle q_2, q_1\rangle q_2\|^2 \label{fPfit}
			\end{align}
			Clearly, $d_{\not \mathcal{E}}$ is well-defined (i.e., does not depend on the choice of $\beta_i \in (\beta_i)$), symmetric, positive. It is zero if and only if $|\langle q_2, q_1\rangle|=1$ and, hence, $(\beta_1) = (\int_0^t q_1(s) |q_1(s)|\,ds) = (\langle q_2, q_1 \rangle \int_0^t q_2(s) |q_2(s)|\,ds) = (\beta_2)$. To show the triangle inequality let $(\beta_3) \in \tilde{\mathfrak{B}}$ with $q_3 = \Psi(\beta_3)$ and $v^* = \langle q_2, q_1\rangle$. Then $d_{\not \mathcal{E}}((\beta_1), (\beta_3)) = \inf_{u\in \mathbb{C}} \|q_1 - u\, q_3 \| \overset{\mathbb{L}^2}{\underset{\text{tr. ineq.}}{\leq}} 
			\underbrace{\|q_1 - v^*\, q_2 \|}_{\overset{(\ref{fPfit})}{=} 
				\inf_{v\in \mathbb{C}} \|q_1 - v\, q_2 \|} + 
			\underbrace{\inf_{u\in \mathbb{C}} \|v^*\, q_2 - u\, q_3 \|}_{= |v^*| \inf_{u\in \mathbb{C}} \|q_2 - u\, q_3 \|} \overset{|v^*|\leq1}{\leq}
			d_{\not \mathcal{E}}((\beta_1), (\beta_2)) + d_{\not \mathcal{E}}((\beta_2), (\beta_3))$.
			This shows i). ii) directly follows from (\ref{optimRot}), since $\exp(-\im\omega_2) = \langle q_1, q_2\rangle / |\langle q_1, q_2\rangle|$.
	\end{proof}
	
	\begin{proof}[Proposition \ref{lem:ifullProcMean} iii)]
		$\min_{(\beta) \in \tilde{\mathfrak{B}}} \mathbb{E}\left( d^2_{\not \mathcal{E}}((\beta), (B))  \right) = \min_{y : \|y\|=1} \mathbb{E}\left( 1 - |\langle y, Q \rangle|^2  \right) = 1 - \max_{y : \|y\| = 1} \, \mathbb{E}\left( | \langle y, Q \rangle |^2 \right)$. 
		Hence, $\psi_{\not \mathcal{E}} \in \operatorname{argmax}_{y : \|y\| = 1}  \mathbb{E}\left(| \langle y, Q \rangle |^2 \right)$, and $\mathbb{E}\left(| \langle y, Q \rangle |^2 \right) = \langle y, \Sigma(y) \rangle = \langle y, \sum_k\lambda_k \langle e_k, y\rangle e_k\rangle = \sum_k \lambda_k |\langle e_k, y\rangle|^2 \leq \lambda_1 \sum_k |\langle e_k, y\rangle|^2 = \lambda_1 \|y\|^2 \overset{}{=} \lambda_1$, due to $\lambda_k \leq \lambda_1$ and $\|y\|=1$, with equality attained by all $y = \frac{x}{\|x\|}$ with $x \in \mathcal{Y}_1$. This also yields $(\mu_{\not \mathcal{E}})$ and $\sigma_{\not \mathcal{E}}^2$.
	\end{proof}
	
	\begin{proof}[Proposition \ref{lem:efullProcDist}]
		$d_\mathcal{E}$ defines a metric on $\mathfrak{B}$ and allows for the provided expression:
		\begin{align*}
			d_\mathcal{E}^2([ \beta_1 ], [ \beta_2]) 
			&= \inf_{a \geq 0, v_i \in \mathbb{C}, \omega_i \in \mathbb{R}, \gamma_i \in \Gamma, i=1,2} \|\  \exp(\im\,\omega_1)\ q_1\circ\gamma_1 {\dot{\gamma}_1^{1/2}} - a\ \exp(\im\,\omega_2)\ q_2\circ\gamma_2 {\dot{\gamma}_2^{1/2}}\ \|^2\\
			&\overset{(*)}{=} \inf_{u\in \mathbb{C}, \gamma \in \Gamma} \|q_1 - u\, q_2\circ\gamma\, {\dot{\gamma}^{1/2}}\|^2 
			\overset{(**)}{=}  1 - \sup_{\gamma \in \Gamma} | \langle q_1, q_2 \circ \gamma {\dot{ \gamma}^{1/2}} \rangle |^2
		\end{align*}
		where $(*)$ follows from isometry of rotation and warping action setting $u = a\, \exp(\im\, (\omega_2 - \omega_1)),\ \gamma = \gamma_2\circ\gamma_1^{-1}$; and $(**)$ is analogous to the proof of Proposition 3. 
		
		As $\Gamma$ acts on $\tilde{\mathfrak{B}}$ by isometries, $\inf_{u\in \mathbb{C}, \gamma \in \Gamma} \|q_1 - u\, q_2\circ\gamma\, {\dot{\gamma}^{1/2}}\| = \inf_{\gamma \in \Gamma} d_{\not\mathcal{E}}((\beta_1), (\beta_2))$ is a semi-metric. To see that it is also positive-definite, assume $d_\mathcal{E}([\beta_1], [\beta_2]) = 0$. 
		Consider any minimizing sequence $\{u_l\}_l$  with
		$0 = d_\mathcal{E}([\beta_1], [\beta_2]) = \inf_{\gamma\in \Gamma} \lim_{l \rightarrow \infty} \| q_1 - u_l q_2 \circ \gamma {\dot{ \gamma}^{1/2}}\|$. Then, $\{u_l\}_l$ is bounded, since $|u_l| \|q_2\| = \inf_{\gamma\in\Gamma} |u_l| \|q_2 \circ \gamma {\dot{ \gamma}^{1/2}}\| = \inf_{\gamma\in\Gamma} \|u_l q_2 \circ \gamma {\dot{ \gamma}^{1/2}}\| \leq \inf_{\gamma\in\Gamma} \|u_l q_2 \circ \gamma {\dot{ \gamma}^{1/2}} - q_1\| + \|q_1\| = \|q_1\|$ and $\|q_2\|>0$ since $\beta_1$ is assumed non-constant. Hence, there is a convergent sub-sequence $\lim_{h\rightarrow\infty} u_{l_h} = u$, and $0 = \inf_{\gamma\in \Gamma} \lim_{h \rightarrow \infty} \| q_1 - u_{l_h} q_2 \circ \gamma {\dot{ \gamma}^{1/2}}\| \overset{\text{continuity}}{=} \inf_{\gamma \in \Gamma} \|q_1 - u\, q_2\circ \gamma {\dot{ \gamma}^{1/2}}\|$ which is known to be a metric on $\mathbf{q}_1 = \kappa(q_1), \mathbf{q}_2 = \kappa(q_2) \in \mathbb{L}^2([0,1], \mathbb{R}^2)$ \citep{bruveris2016optimal}. Hence, also $[\beta_1] = [\beta_2]$ which completes the proof.
	\end{proof}
	
	\begin{proof}[Proposition \ref{lem:efullProcMean}]
		In analogy to Proposition \ref{lem:ifullProcMean}, 
		$\min_{[\beta] \in \mathfrak{B}} \mathbb{E}\left( d^2_\mathcal{E}([\beta], [B])  \right) = \min_{y : \|y\|=1} \mathbb{E}\left( 1 - \sup_{\gamma \in \Gamma} |\langle y, Q \circ \gamma \, \dot{ \gamma}^{1/2} \rangle|^2  \right) = 1 - \max_{y : \|y\| = 1} \, \mathbb{E}\left( \sup_{\gamma \in \Gamma} | \langle y, Q \circ \gamma\, \dot{ \gamma}^{1/2} \rangle |^2 \right)$.
	\end{proof}
	
	
	\subsection{The square-root-velocity representation in a sparse/irregular setting}
	
	\begin{theorem}\label{app:lem:feasible}
		Let $\beta:[0,1] \rightarrow \mathbb{C}$ be continuous, injective, and, for all $t\in(0,1)$, continuously differentiable with $\dot{\beta}(t) = \frac{d}{dt}\Re \circ \beta(t) + \im \, \frac{d}{dt}\Im \circ \beta(t) \neq 0$. Then,  there exists a $c \in (0,1)$ such that $\dot{\beta}(c) = \delta\, (\beta(1) - \beta(0))$ for some $\delta>0$.
	\end{theorem}
	
	\begin{proof}
		Let $\rho = \Re\circ \beta$ and $\zeta = \Im\circ \beta$ denote the real and imaginary part of $\beta$. Without loss of generality assume $\beta(0) = 0$ and $\beta(1) = \im$.
		Choose $0 \leq t_0< t_1 \leq 1$ with $\rho(t_0) = \rho(t_1) = 0$ such that $\zeta(t) \geq \zeta(t_0)$ for all $t\in[0,1]$ with $\rho(t) = 0$
		and $\zeta(t) \leq \zeta(t_1)$ for all $t \in [t_0, 1]$ with $\rho(t) = 0$.
		If $\rho(t) = 0$ for all $t \in [t_0, t_1]$ and, hence, $\dot{\beta}(t) = \im\, \dot{\zeta}(t)$ within $(t_0, t_1)$, the Mean Value Theorem directly yields existence of the desired $c\in(t_0, t_1)$. 
		We may, thus, assume $\rho(t)\neq 0$ for some $t \in [t_0, t_1]$, say, with $\rho(t)>0$. 
		Accordingly, a maximizer $c \in [t_0, t_1]$ with $\rho(c) = \max_{t\in[t_0, t_1]} \rho(t) > 0$ lies in $(t_0, t_1)$ and $\dot{\rho}(c) = 0$, since $\rho$ is continuously differentiable. Hence $\dot{\beta}(c) = \im \dot{\zeta}(c) \neq 0$ as $\beta$ is regular. $t_0 \neq t_1$ and $c$  all exist due to compactness/continuity arguments.
		
		We will now assume $\delta =\dot{\zeta}(c) < 0$ and show that this leads to a contradiction.
		With some upper/lower bounds $\rho_{\sup} > \rho(c) (> 0)$ and $\zeta_{\inf} < \min_{t\in[0,1]} \zeta(t)$, we  construct the open polygonal curve $\alpha:[c, 1]$ connecting the points $a_1 = \beta(c)$, $a_2 = \rho_{\sup} + \im \zeta_{\inf}$, $a_3 = \im \zeta_{\inf}$ and $a_4 = \beta(t_0) \leq 0$. Then $\beta \ind_{[t_0, c]} + \alpha \ind_{[c, 1]}$ is a simple closed continuous curve on $[t_0, 1]$, hence splits $\mathbb{C}$ into two connected open components, the interior component $\mathcal{A}\subset \mathbb{C}$ which is bounded and the exterior component $\mathcal{U} = \mathbb{C}\setminus \bar{\mathcal{A}}$ (Jordan curve theorem) where $\bar{\mathcal{A}}$ denotes the closure of $\mathcal{A}$.
		The path $\phi:[0, \infty) \rightarrow \mathbb{C}$, $r \mapsto \beta(t_1) + r\, \im$ does not intersect the boundary $\beta([t_0, c]) \cup \alpha([c, 1]) = \bar{\mathcal{A}}\cap \bar{\mathcal{U}}$ for all $r\geq0$, since, by construction, $\zeta(t_1) > \Im(a_k)$ for $k=2,\dots, 4$ and, for all $t\in[t_0, c]$ with $\rho(t)=0$, $\zeta(t_1) > \zeta(t)$ as $\zeta(t_1) \geq \zeta(t)$, $c<t_1$ and $\beta$ injective.
		Thus, $\phi$ lies entirely in $\mathcal{A}$ or in $\mathcal{U}$. Since $\mathcal{A}$ is bounded, the path and, in particular, $\phi(0)=\beta(t_1) \in \mathcal{U}$.
		Due to the construction of $\alpha$ and injectivity of $\beta$ that do not permit intersection of the boundary (Jordan curve), $\beta(t)$ lies in $\mathcal{A}$ for all $t > c$ if it lies within $\mathcal{A}$ for some $t > c$. This makes the local behavior at $c$ crucial. Thus, the assumption of $\dot{\zeta}(c) < 0$ entailing $\beta(t) \in \mathcal{A}$ for some $t > 0$ yields, in particular, $\beta(t_1) \in \mathcal{A}$ and, hence, the desired contradiction.  
	\end{proof}
	
	\begin{corollary}[Feasible sampling]
		If $\beta^*:[0,1] \rightarrow \mathbb{C}$ is continuous and $\beta^*: (t_{j-1}^*, t_j^*) \rightarrow \mathbb{C}$ continuously differentiable for $j = 1, \dots, n_0$, $t^*_0 < \dots < t^*_{n_0}$ with non-vanishing derivative, then for any time points $0 < t_1 < \dots < t_{n_0}< 1$ and speeds $w_1, \dots, w_{n_0}>0$, there exists a $\gamma \in \Gamma$ such that 	for the \SRV-transform $q$ of $\beta = \beta^* \circ \gamma$,
		$q(t_j) = {w_j^{1/2}} \, (\beta^*(t^*_{j}) - \beta^*(t^*_{j-1})) = {w_j^{1/2}} \, \Delta_j$ for all $j = 1, \dots n_0$.
	\end{corollary}
	\begin{proof}
		Since this is a local property, it suffices to consider the case of $n_0=1$ and $t_0^* = 0, t_1^* = 1$.
		By Theorem \ref{app:lem:feasible}, there exists $c \in (0,1)$ with $\dot{\beta}(c)^* = a\, \Delta_1$ for some $a>0$.
		Choose $\gamma\in\Gamma$ such that $\gamma(t_1) = c$ and $\dot\gamma(t_1)={w_1}{a^{-2}}$. Then, $q(t_j) = \beta^*\circ\gamma(t_j) {\dot{\gamma(t_j)}^{1/2}} = a\, \Delta_1 {{w_1^{1/2}}{a^{-1}}} = {w_1^{1/2}} \Delta_1$ for all $j=1, \dots, n_0$. 
	\end{proof}
	
	
	\subsection{Estimating elastic full Procrustes means via Hermitian covariance smoothing}
	
	In the following, we provide additional details for three steps in our proposed elastic full Procrustes mean estimation algorithm. We commence with proposing a more efficient covariance estimation procedure for data with densely observed curves and continue with a discussion of conditional complex Gaussian processes in Proposition \ref{prop:conditionalgaussian} underlying our estimation of length and optimal rotation of curves. Finally, we detail the warping alignment strategy proposed for the re-parameterization step.
	
	\textbf{Covariance estimation for densely observed curves:}
	If curves $y_1, \dots, y_n$, are sampled densely enough,
	covariance estimation can be achieved computationally more efficient than by Hermitian covariance smoothing. 
	In fact, for say $n_i > 1000$ samples per curve and $m$ basis functions $\mathbf{f} = (f_1, \dots, f_m)^\top$ for each margin, setting up the joint $(\sum_{i=1}^{n}n_i^2) \times ({m^2\pm m})/{2}$ design matrices for tensor-product covariance smoothing may also cause working memory shortage. 
	Using the notation of Section 2.2, 
	we obtain a tensor-product covariance estimator $\hat{C}(s,t) = \mathbf{f}^\top\!(s)\, \widehat{\mathbf{\Xi}}\, \mathbf{f}(t)$ of the same form by setting $\widehat{\mathbf{\Xi}} = \frac{1}{n} \sum_{i=1}^{n} \hat{\boldsymbol{\vartheta}}_i \hat{\boldsymbol{\vartheta}}_i^\dagger$ to the empirical covariance matrix of complex coefficient vectors $\hat{\boldsymbol{\vartheta}}_i = (\hat{\vartheta}_{i1}, \dots, \hat{\vartheta}_{im})^\top \in \mathbb{C}^m$ of basis representations $y_i(t) \approx \sum_{k=1}^{m} \hat{\vartheta}_{ik} f_k(t)$ for $i = 1,\dots, n$. 
	Partitioning the data into $\mathcal{N}_1 \cup \dots \cup \mathcal{N}_N = \{1,\dots, n\}$ subsets for computational efficiency (which might simply be given by $\mathcal{N}_i = \{i\}$), the estimators $\hat{\boldsymbol{\vartheta}}_i$ are fit by minimizing the penalized least-squares criterion
	\begin{equation} \nonumber
		\textsc{pls}(\boldsymbol{\vartheta}_{i\Re}, \boldsymbol{\vartheta}_{i\Im}) = 
		\sum_{i\in \mathcal{N}_l} \sum_{j=1}^{n_i}  
		\left| y_{ij} - 
		{\boldsymbol{\vartheta}}_{i\Re}^\top \mathbf{f}(t_{ij}) - \im\,\boldsymbol{\vartheta}_{i\Im}^\top\mathbf{f}(t_{ij})
		\right|^2 + 
		\eta\, {\boldsymbol{\vartheta}}_{i\Re}^\top \mathbf{P} \, {\boldsymbol{\vartheta}}_{i\Re} + 
		\eta\, {\boldsymbol{\vartheta}}_{i\Im}^\top \mathbf{P}\, {\boldsymbol{\vartheta}}_{i\Im}
	\end{equation}
	with $\boldsymbol{\vartheta}_{i\Re} = \Re(\boldsymbol{\vartheta}_{i})$ and $\boldsymbol{\vartheta}_{i\Im} = \Im(\boldsymbol{\vartheta}_{i})$, 
	for $l=1,\dots, N$.
	In principle, real and imaginary parts can be separately fit with the same smoothing parameter $\eta \geq0$ in both parts to achieve rotation invariant penalization.
	As in Section 2.2, 
	we use the \texttt{mgcv} framework for fitting \citep{Wood2017} using restricted maximum likelihood (\REML)  estimation for $\eta$. 
	To speed up computation, $\eta$ can be estimated only on $\mathcal{N}_1$ and fixed for $l=2, \dots, N$, or set to $\eta=0$ if no measurement error is assumed or no penalization is desired. 
	The residual variance yields a constant estimate for $\tau^2$. Using for instance \texttt{mgcv}'s ``\texttt{gaulss}" family, a smooth estimator $\hat{\tau}^2(t)$ could be obtained as well but is not detailed here.
	
	\textbf{Rotation and length estimation:}
	As proposed by \cite{Yao:etal:2005} for predicting scores in functional principal component analysis, we propose to use conditional expectations under a working normality assumption to incorporate the covariance structure of the data into estimation of inner products and quadratic terms. These are used for predicting basis coefficients of a curve (Proposition \ref{prop:conditionalgaussian} iii) Equation \eqref{eq:expectedCoef}), its optimal rotation to the mean \eqref{eq:expectedScore}, its length \eqref{eq:expectedNorm}, and its distance from the mean \eqref{eq:expectedScoreSquare} or another given curve. 
	We provide required conditional expectations covering both the case of a positive white noise error variance $\tau^2(t)>0$ and of no white noise error ($\tau^2(t)=0$) for each time point $t$. 
	The distinction runs through all formulations and reading might be more convenient when assuming either of the cases is always fulfilled.
	
	\begin{proposition}[Conditional Gaussian process]
		\label{prop:conditionalgaussian}
		Consider a random element $Y$ in a complex Hilbert space $\mathbb{H}$ of functions $\mathcal{T} \rightarrow \mathbb{C}$ defined on some set $\mathcal{T}$. Assume $Y = \sum_{k=1}^{m} Z_k e_k$ finitely generated with probability one from a
		finite set $\mathbf{e}(t) = (e_1(t), \dots, e_m(t))^\top$ of functions $e_k\in\mathbb{H}$ with regular Gramian $\mathbf{G} = \{\langle e_k, e_{k'} \rangle\}_{k,k'} \in \mathbb{C}^{m\times m}$ 
		and with $\mathbf{Z}=(Z_1, \dots, Z_m)^\top$
		following a complex symmetric multivariate normal distribution with positive-definite covariance matrix $\mathbf{\Lambda}$. 
		Let further denote $\varepsilon$ an uncorrelated complex symmetric error process on $\mathcal{T}$ with variance function $\tau^2: \mathcal{T} \rightarrow \mathbb{R}$.
		We consider a sequence of $n_* = n_0 + n_+$ points $t_1, \dots, t_{n_*} \in \mathcal{T}$ and values $y_1, \dots, y_{n_*} \in \mathbb{C}$ with $\tau^2(t_1), \dots, \tau^2(t_{n_0}) = 0$ and $\tau^2(t_{n_0 + 1}), \dots, \tau^2(t_{n_0 + n_+}) > 0$. Write $\mathbf{E} = \{e_k(t_j)\}_{jk} = (\mathbf{E}_0^\top, \mathbf{E}_+^\top)^\top$ for the $n_*\times m$ design matrix of function evaluations subdivided into $\mathbf{E}_0\in \mathbb{C}^{n_0 \times m}$ and $\mathbf{E}_+ \in \mathbb{C}^{n_+ \times m}$  containing the evaluations with zero and positive error variance, respectively, and analogously $\mathbf{y} = (y_1, \dots, y_{n_*})^\top = (\mathbf{y}_0^\top, \mathbf{y}_+^\top)^\top$ for the values and $\mathbf{T}_+ = \operatorname{Diag}(\tau^2(t_1), \dots, \tau^2(t_{n_+}))$ for the diagonal $n_+\times n_+$ noise covariance matrix. 
		Let $r_0 = \operatorname{rank}(\mathbf{E}_0)$ denote the rank of $\mathbf{E}_0$ and $\mathbf{Q} = (\mathbf{M}, \mathbf{N})$ be an $m\times m$ Hermitian matrix such that $\mathbf{M}$ is $m \times r_0$ and $\mathbf{N}$ spans the null space of $\mathbf{E}_0$. $\mathbf{Q}$ is obtained, e.g., by the QR-decomposition $\mathbf{E}_0^\top  = \mathbf{Q} \mathbf{R}$. By convention, matrices are set to $0$ if their rank is zero (i.e., if $m-r_0$, $n_0$, or $n_+=0$, respectively).
		Conditioning on $Y(t_{j}) + \varepsilon(t_j) = \mathbf{Z}^\top \mathbf{e}(t_j) + \varepsilon(t_j) = y_j$ for $j=1,\dots, n_*$ we obtain:
		\begin{enumerate}[i)]
			\item 
			$\mathbf{Z} = \mathbf{Z}_+ + \mathbf{z}_0$ is split into a random part $\mathbf{Z}_+ = \mathbf{N}\tilde{\mathbf{Z}}_+$ constrained to the linear sup-space $\operatorname{span}(\mathbf{N})$ spanned by $\mathbf{N}$, with $\tilde{\mathbf{Z}}_+$ a complex random vector of length $m-r_0$, and a deterministic part $\mathbf{z}_0 = \mathbf{M} \left(\mathbf{M}^\dagger \mathbf{E}_0^\dagger \mathbf{E}_0 \mathbf{M} \right)^{-1} \mathbf{M}^\dagger \mathbf{E}_0^{\dagger}\mathbf{y}_0$. In fact, under the given assumptions $\mathbf{z}_0 = \mathbf{M} (\mathbf{M} \mathbf{E}_0)^{-\dagger} \mathbf{y}_0$ with probability one, but the generalized inverse is robust with respect to the case where $\mathbf{y}_0 \notin \operatorname{span}(\mathbf{E}_0)$, i.e. where no measurement error is assumed but the curve cannot be exactly fit by the chosen basis. 
			\item $\tilde{\mathbf{Z}}_+$ follows a complex normal 
			with covariance $\mathbf{S} = \left(\mathbf{N}^\dagger \left(\mathbf{E}_+^\dagger \mathbf{T}_+^{-1} \mathbf{E}_+ + \mathbf{\Lambda}^{-1}\right) \mathbf{N} \right)^{-1}$, mean $\hat{\mathbf{z}}_+ =  \mathbf{S} \mathbf{N}^\dagger \left( \mathbf{E}_+^\dagger \mathbf{T}_+^{-1} \left(\mathbf{y}_+ - \mathbf{E}_+ \mathbf{z}_0\right) - \mathbf{\Lambda}^{-1} \mathbf{z}_0 \right)$ and zero pseudo-covariance.
			\item For $x \in \mathbb{H}$ and $\mathbf{g}_x = (\langle e_1, x \rangle, \dots, \langle e_m, x \rangle)$, this provides conditional means 
			\begin{eqnarray}
				\label{eq:expectedCoef} \hat{\mathbf{z}} = \mathbb{E}\left( \mathbf{Z} \mid Y(t_j) + \varepsilon(t_j) = y_j, j = 1, \dots, n_* \right) &=& \mathbf{N} \hat{\mathbf{z}}_+ + \mathbf{z}_0\\	
				\label{eq:expectedScore} \mathbb{E}\left( \langle Y, x \rangle \mid Y(t_j) + \varepsilon(t_j) = y_j, j = 1, \dots, n_* \right) &=& \hat{\mathbf{z}}^\dagger  \mathbf{g}_x\\	
				\label{eq:expectedNorm} \mathbb{E}\left( \| Y \|^2 \mid Y(t_j) + \varepsilon(t_j) = y_j, j = 1, \dots, n_* \right) &=& \operatorname{tr}\left( \mathbf{S}\, \mathbf{G} \right) + \hat{\mathbf{z}}^\dagger \mathbf{G} \hat{\mathbf{z}}.\\
				\label{eq:expectedScoreSquare} \mathbb{E}\left( | \langle Y, x \rangle |^2 \mid Y(t_j) + \varepsilon(t_j) = y_j, j = 1, \dots, n_* \right) &=& \mathbf{g}_x^\dagger \mathbf{S}\, \mathbf{g}_x  + \mathbf{g}_x \hat{\mathbf{z}}^\dagger \hat{\mathbf{z}}\, \mathbf{g}_x^\dagger.
			\end{eqnarray} 
		\end{enumerate} 
	\end{proposition}
	\begin{proof}
		The computation is analogous to the real case. Defining $\bm{Y} = (Y(t_1),\dots, Y(t_{n_*})^\top$, i.e.\, $\bm{Y} = \mathbf{E}\mathbf{Z}$, and $\bm{\epsilon} = (\varepsilon(t_1),\dots, \varepsilon(t_{n_*}))^\top$, the distribution of $\tilde{\mathbf{Z}} = \mathbf{Q}^\dagger \mathbf{Z} = (\mathbf{M}^\dagger \mathbf{Z}, \mathbf{N}^\dagger \mathbf{Z})^\dagger = (\tilde{\mathbf{Z}}_0^\dagger,  \tilde{\mathbf{Z}}_+^\dagger)^\dagger$ conditional on $\bm{Y} + \bm{\epsilon} = \mathbf{y}$ has a density proportional to
		\begin{align*}
			\nonumber
			p_{\tilde{\mathbf{Z}}}(\tilde{\mathbf{z}} &\mid \bm{Y} + \bm{\epsilon} = \mathbf{y} ) 
			\propto 
			p_{\tilde{\mathbf{Z}}, \bm{Y} + \bm{\epsilon}}(\tilde{\mathbf{z}}, \bm{Y} + \bm{\epsilon} ) 
			\propto
			p_{\mathbf{Z}, \bm{\epsilon}}(\overbrace{\mathbf{Q}\,\tilde{\mathbf{z}}}^{=\mathbf{M}\tilde{\mathbf{z}}_0 + \mathbf{N} \tilde{\mathbf{z}}_+}, \mathbf{y}-\mathbf{E}\mathbf{Q}\,\tilde{\mathbf{z}}) \\
			\nonumber
			&\propto
			\exp\left(- \frac{1}{2} \tilde{\mathbf{z}}^\dagger \mathbf{Q}^\dagger \mathbf{\Lambda}^{-1} \mathbf{Q} \tilde{\mathbf{z}} \right) \cdot\\ 
			\nonumber
			&\quad \cdot \exp\left(- \frac{1}{2} (\mathbf{y}_+-\mathbf{E}_+\mathbf{Q}\,\tilde{\mathbf{z}})^\dagger \mathbf{T}_+^{-1} (\mathbf{y}_+-\mathbf{E}_+\mathbf{Q}\,\tilde{\mathbf{z}}) \right)\, \ind_{\{\mathbf{y}_0\}}(\mathbf{E}_0\mathbf{Q}\,\tilde{\mathbf{z}})\\
			&\overset{(*)}{\propto}
			\exp\left(- \frac{1}{2} \left( \tilde{\mathbf{z}}_+^\dagger \mathbf{N}^\dagger \mathbf{\Lambda}^{-1} \mathbf{N}\,\tilde{\mathbf{z}}_+ 
			\right) - 
			\Re\left( \tilde{\mathbf{z}}_+^\dagger \mathbf{N}^\dagger \mathbf{\Lambda}^{-1} \mathbf{z}_0\right) \right) \cdot \, \\
			\nonumber
			&\quad \cdot \exp\left(- \frac{1}{2} \tilde{\mathbf{z}}_+^\dagger \mathbf{N}^\dagger\mathbf{E}_+^\dagger \mathbf{T}_+^{-1} \mathbf{E}_+ \mathbf{N}\,\tilde{\mathbf{z}}_+ +
			\Re\left( \tilde{\mathbf{z}}_+^\dagger \mathbf{N}^\dagger \mathbf{E}_+^\dagger \mathbf{T}_+^{-1} (\mathbf{y}_+ -  \mathbf{E}_+ \mathbf{z}_0) \right) 
			\right)\, \ind_{\{\mathbf{M}^\dagger\mathbf{z}_0\}}(\tilde{\mathbf{z}}_0) \\ 
			\nonumber
			&\propto 
			\exp\Bigg(- \frac{1}{2} \tilde{\mathbf{z}}_+^\dagger \underbrace{\,\mathbf{N}^\dagger\left(\mathbf{\Lambda}^{-1} + \mathbf{E}_+^\dagger \mathbf{T}^{-1}_+ \mathbf{E}_+\right)\mathbf{N}\,}_{=\, \mathbf{S}^{-1}}  \tilde{\mathbf{z}}_+ + \\
			\nonumber
			&\quad + 
			\Re\bigg( \tilde{\mathbf{z}}_+^\dagger \underbrace{\mathbf{N}^\dagger \Big(\mathbf{E}_+^\dagger \mathbf{T}_+^{-1} (\mathbf{y}_+ -  \mathbf{E}_+ \mathbf{z}_0) - \mathbf{\Lambda}^{-1} \mathbf{z}_0 \Big)}_{=\mathbf{S}^{-1} \hat{\mathbf{z}}_+}  \bigg)  \Bigg)  \, \ind_{\{\mathbf{M}^\dagger\mathbf{z}_0\}}(\tilde{\mathbf{z}}_0) \\ 
			\nonumber
			&\propto \exp\left( -\frac{1}{2} \left(\tilde{\mathbf{z}}_+ - \hat{\mathbf{z}}_+ \right)^\dagger \mathbf{S}^{-1} \left(\tilde{\mathbf{z}}_+ - \hat{\mathbf{z}}_+ \right) \right)  \, \ind_{\{\mathbf{M}^\dagger\mathbf{z}_0\}}(\tilde{\mathbf{z}}_0).
		\end{align*}
		Solving $\mathbf{y}_0 = \mathbf{E}_0 \mathbf{Q} \tilde{\mathbf{z}} = \mathbf{E}_0 \mathbf{M} \tilde{\mathbf{z}}_0$ for $\tilde{\mathbf{z}}_0$ yields $(*)$ and shows i). Deriving the kernel of a Gaussian, the remainder of the computation shows ii). 
		In iii), (\ref{eq:expectedCoef}) and (\ref{eq:expectedScore}) follow directly by linearity and (\ref{eq:expectedNorm}) from variance decomposition (omitting conditions for brevity):
		\begin{align*}
			\mathbb{E}\left( \| Y \|^2 \right) &= 
			\mathbb{E}\left( \langle \sum_{k=1}^{m} Z_k e_k, \sum_{k=1}^{m} Z_k e_k \rangle \right) =
			\mathbb{E}\left( \mathbf{Z}^\dagger \mathbf{G} \mathbf{Z} \right) = 
			\mathbb{E}\left( \operatorname{tr}\left(\mathbf{Z} \mathbf{Z}^\dagger \mathbf{G} \right) \right)\\ 
			&\overset{\text{linearity}}{=} 
			\operatorname{tr}\left( \mathbb{E}\left( \mathbf{Z} \mathbf{Z}^\dagger \right)  \mathbf{G} \right)
			= \operatorname{tr}\left( \left( \operatorname{Var}\left( \mathbf{Z} \right) + \mathbb{E}\left(\mathbf{Z}\right) \mathbb{E}\left(\mathbf{Z}\right)^\dagger \right) \mathbf{G} \right) \overset{\text{ii)}}{=} 
			\operatorname{tr}\left( \mathbf{S}\, \mathbf{G} \right) + \hat{\mathbf{z}}^\dagger \mathbf{G} \hat{\mathbf{z}}. \nonumber 
		\end{align*} 
		The computation for (\ref{eq:expectedScoreSquare}) is analogous.
	\end{proof}
	
	\textbf{Warping alignment:}
	Generally, we consider it advisable to base warping alignment of the $i$th curve directly on its original \SRV-evaluations $q_{i1}^{[h]},\dots, q_{in_i}^{[h]}$ but, when considerable measurement error presents an issue, it might also be useful to employ a smoothed reconstruction $\tilde{q}_i: [0,1] \rightarrow \mathbb{C}$ of the \SRV-transform in the assumed basis. 
	Based on the working normality assumption used also for length and rotation estimation, such a reconstruction is obtained as  $\tilde{q}_i^{[h]}(t) =
	({\hat{\mathbf{z}}^{[h]}_{i}/{\|\hat{\mathbf{z}}_{i}\|})^\top \hat{\mathbf{e}}^{[h]}(t)}$ with $\hat{\mathbf{z}}^{[h]}_{i} = (\hat{z}_{i1}^{[h]}, \dots, \hat{z}_{im}^{[h]})^\top$ the predicted score vector for the eigenbasis $\hat{\mathbf{e}}^{[h]} = (\hat{e}^{[h]}_1, \dots, \hat{e}^{[h]}_m)^\top$.
	
	Following \citet{steyer2021elastic}, warping alignment to $\hat{\mu}^{[h]}$ is conducted using another, polygonal approximation of the curve given by a piece-wise constant approximation $\hat{q}_i^{[h]} \in \mathbb{L}^2([0,1], \mathbb{C})$ of $q_i^{[h]}$. 
	With a hyper-parameter $\rho \in [0,1]$, we control the balance between original $q_{ij}^{[h]}$ (for $\rho = 0$) and smoothed reconstruction $\tilde{q}_i$ (for $\rho = 1$) and set $\hat{q}_{ij}^{[h]} = \hat{u}_i^{[h]} \left(\varrho\,\tilde{q}_i^{[h]}(t_{ij}^{[h]})  + (1-\varrho)\,  q_{ij}^{[h]}\right)$ at nodes $s_{i0}^{[h]}=0$, $s_{ij}^{[h]} = 2 t^{[h]}_{ij} - s^{[h]}_{ij-1}$,  $j=1, \dots, n_i$. 
	This defines $\hat{q}_i^{[h]}(t) = \sum_{j=1}^{n_i} \hat{q}_{ij}^{[h]} \, \ind_{[s_{ij-1}^{[h]}, s_{ij}^{[h]})}(t)$ already rotated by $\hat{u}_i^{[h]}$.
	
	Warping alignment to $\hat{\mu}^{[h]}$ is achieved for $i=1,\dots, n$ by finding an optimal $\hat{q}_i^* \in \mathbb{L}^2([0,1], \mathbb{C})$ with 
	\begin{equation}
		\label{eq:optimalwarping1}
		\|\hat{q}_i^* - \hat{\psi}^{[h]}\| \leq \|\hat{q}_i^{[h]} \circ\gamma\, {\dot{\gamma}^{1/2}} - \hat{\psi}^{[h]}\| \quad \text{for all }\gamma\in\Gamma
	\end{equation}
	where the polygon approximation yields a practically feasible optimization problem and has proven suitable for sparse/irregular curves \citep{steyer2021elastic}. 
	As shown by \citet{steyer2021elastic}, the optimizers of (\ref{eq:optimalwarping1}) have the form $\hat{q}_i^*(t) = \sum_{j=1}^{n_i} w_{i}(t)\, \hat{q}_{ij}^{[h]} \ind_{[s^{[h+1]}_{ij-1}, s^{[h+1]}_{ij})}(t)$ almost-everywhere, where, denoting $a_+ = \max\{a, 0\}$ for $a\in\mathbb{R}$, the functions $w_i: [0,1] \rightarrow \mathbb{R}$ are given by $w_i^2(t) = {(s^{[h]}_{ij} - s^{[h]}_{ij-1})\Re\left( \psi^{[h]}(t)^\dagger \hat{q}_{ij}^{[h]} \right)^2_+}/{\int_{s^{[h+1]}_{ij-1}}^{s^{[h+1]}_{ij}} \Re\left( \psi^{[h]}(t)^\dagger \hat{q}_{ij}^{[h]} \right)^2_+\,dt}$ for $t \in [s^{[h+1]}_{ij-1}, s^{[h+1]}_{ij})$, and  fully determined by the warped time points 
	\begin{equation}
		(s_{i1}^{[h+1]}, \dots, s_{in_i-1}^{[h+1]}) = \underset{0 = s_{i0} \leq \dots \leq s_{in_i} = 1}{\arg\max} \sum_{j=1}^{n_i} \big({(s_{ij}^{[h]} - s_{ij-1}^{[h]}) \int_{s_{ij-1}}^{s_{ij}} \Re\left( \psi^{[h]}(t)^\dagger \hat{q}_{ij}^{[h]} \right)^2_+ \, dt } \big)^{1/2}. \nonumber
	\end{equation}
	If $s_{ij}^{[h+1]} = s_{ij-1}^{[h+1]}$ for some $j$, there is a minimizing sequence of functions of the form given for $\hat{q}^*_i$.  
	After optimization over the $s_{ij}^{[h]}$ with R package \texttt{elasdics} \citep{elasdics},
	we set new $t_{ij}^{[h+1]} = {(s^{[h+1]}_{ij-1} + s^{[h+1]}_{ij})}/{2}$ and $q_{ij}^{[h+1]} = w^*_j\, q_{ij}^{[h]}$ with $w^*_{ij} = (s_{ij}^{[h]} - s_{ij-1}^{[h]})^{1/2}\,(s^{[h+1]}_{ij} - s^{[h+1]}_{ij-1})^{-1/2}$ for $s^{[h+1]}_{ij} > s^{[h+1]}_{ij-1}$ and omit double time points for $j=1, \dots, n_i$. 
	The chosen time-points hereby approximate $t_{ij}^{[h+1]} \approx t_{ij}^*\in (s^{[h+1]}_{ij}, s^{[h+1]}_{ij-1})$ with $w_i(t_{ij}^*) = w^*_{ij}$ existing by the Mean Value Theorem.


\section{Adequacy and robustness of elastic full Procrustes mean estimation in realistic curve shape data}

While we focus on the first letter ``\textit{f}'' in our simulation studies, Figure \ref{fig:fda} exemplifies elastic full Procrustes mean estimation on the entire ``\textit{fda}'' handwritings contained in the dataset \texttt{handwrit.dat} in the R package \texttt{fda} \citep{RamsaySilverman2005}. To visualize different degrees of sparsity, means are fitted after subsampling recorded points to $n_i = n_{points}$, $i=1,\dots,n$, $n=20$, random sampling points for each curve placing higher acceptance probability on points more important for curve reconstruction, as illustrated in the bottom of the figure.
Means are fitted using piece-wise constant 0 order B-splines with 70 knots applying a 2nd order difference penalty in the Hermitian covariance estimation. This results in a nice gradual evolution from a rough ``\textit{fda}'' approximation for $n_{points} = 21$ to a detailed handwritten ``\textit{fda}'' for $n_{points} = 71$.

\begin{figure}
    \input{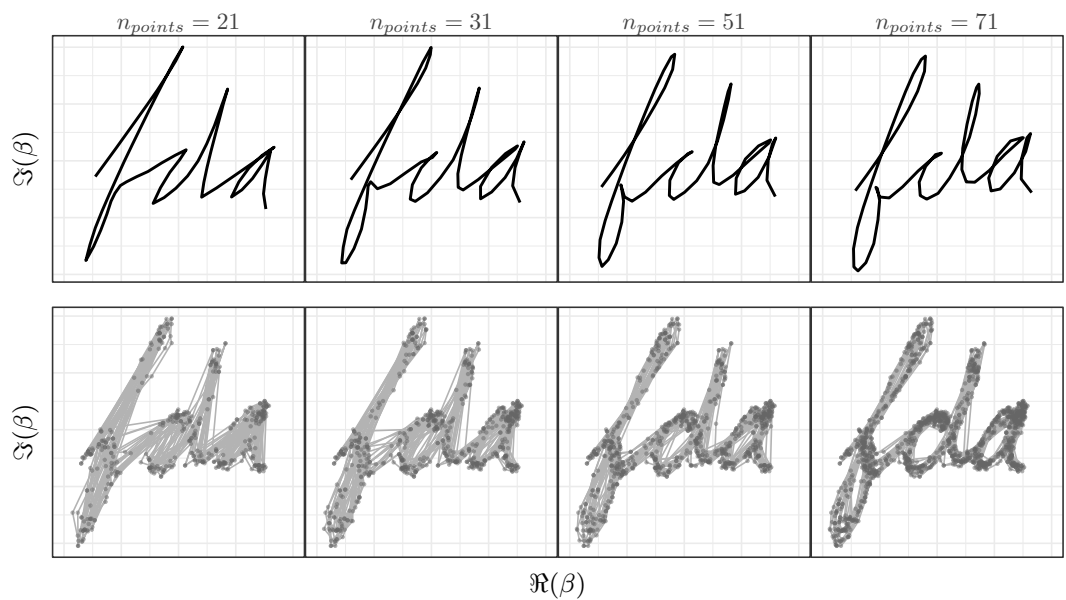}
	\caption{
		\textit{Top:} Elastic full Procrustes means estimated over 20 handwritten ``\textit{fda}''s  sampled with different degrees of sparsity. 
		\textit{Bottom:} Underlying datasets with 20 curves from the \texttt{handwrit.dat} dataset subsampled with higher acceptance probability on points important for curve reconstruction. Points sampled for each curve are connected by light-grey lines. 
		\label{fig:fda}
	}
\end{figure}